\documentclass[journal,onecolumn,draftclsnofoot]{IEEEtran}
\usepackage[T1]{fontenc}
\usepackage{cite}

\usepackage[pdftex]{graphicx}
\usepackage{multicol}
\usepackage{subcaption}

\usepackage[cmex10]{amsmath}
\usepackage{amssymb}
\usepackage{amsthm}
\usepackage{bbm}

\usepackage[shortlabels]{enumitem}

\usepackage{xcolor}

\usepackage{hyperref}
\hypersetup{
    colorlinks=true,       
    linkcolor=teal,        
    citecolor=teal, 
    urlcolor=teal          
}
\usepackage[capitalise]{cleveref}
\usepackage{thmtools}
\declaretheorem[thmbox=M]{theorem}
\declaretheorem[thmbox=M]{proposition}
\declaretheorem[thmbox=M]{definition}
\declaretheorem[thmbox=M]{lemma}
\declaretheorem[thmbox=M]{corollary}
\declaretheorem[thmbox=M]{example}
\declaretheorem[style=remark]{remark}

\newcommand{\thmboxsplit}{\par\penalty-10000\relax}

\makeatletter
\newif\ifthmfooter@active
\thmfooter@activefalse
\let\thmfooter@prooflabel\@empty

\newcommand{\thmfooterbegin}{%
    \global\thmfooter@activetrue
    \global\let\thmfooter@prooflabel\@empty
}

\newcommand{\thmfooterend}{%
    \global\thmfooter@activefalse
    \global\let\thmfooter@prooflabel\@empty
}

\newcommand{\appendixproofref}[1]{%
    \gdef\thmfooter@prooflabel{#1}%
}

\let\thmfooter@originaltail\thmbox@tail
\renewcommand{\thmbox@tail}{%
    \ifthmfooter@active
        \hrule height 0pt\relax
        \thmbox@dim=\hsize
        \advance\thmbox@dim-\thmbox@leftmargin
        \advance\thmbox@dim-\thmbox@rightmargin
        \advance\thmbox@dim\thmbox@hskip
        \advance\thmbox@dim\thmbox@thickness
        \noindent
        {\dimen@=\thmbox@leftmargin
         \advance\dimen@-\thmbox@hskip
         \advance\dimen@-\thmbox@thickness
         \hskip\dimen@}%
        \vbox{%
            \hrule width \thmbox@dim height \thmbox@thickness
            \ifx\thmfooter@prooflabel\@empty
            \else
                \kern-\thmbox@thickness
                \hbox to \thmbox@dim{%
                    \hfill
                    \begingroup
                    \fboxrule=\thmbox@thickness
                    \fbox{%
                        \normalfont\footnotesize
                        Proof in Appendix~\hyperref[\thmfooter@prooflabel]{%
                            \ref*{\thmfooter@prooflabel}%
                        }%
                    }%
                    \endgroup
                }%
            \fi
        }%
        \par
    \else
        \thmfooter@originaltail
    \fi
}
\makeatother

\addtotheorempreheadhook[theorem]{\thmfooterbegin}
\addtotheorempostfoothook[theorem]{\thmfooterend}
\addtotheorempreheadhook[proposition]{\thmfooterbegin}
\addtotheorempostfoothook[proposition]{\thmfooterend}
\addtotheorempreheadhook[lemma]{\thmfooterbegin}
\addtotheorempostfoothook[lemma]{\thmfooterend}
\addtotheorempreheadhook[corollary]{\thmfooterbegin}
\addtotheorempostfoothook[corollary]{\thmfooterend}
\addtotheorempreheadhook[definition]{\thmfooterbegin}
\addtotheorempostfoothook[definition]{\thmfooterend}
\addtotheorempreheadhook[example]{\thmfooterbegin}
\addtotheorempostfoothook[example]{\thmfooterend}

\newcommand{\renyiDiv}[2]{D_{\alpha}\left(#1\|#2\right)}
\newcommand{\dd}[1]{\mathrm{d}#1}
\newcommand{\mean}[2]{\mathbb{E}_{#1}\left[#2\right]}
\newcommand{\ip}[2]{\left\langle#1, #2\right\rangle}

\def \X{\mathsf X}
\def \Y{\mathsf Y}
\def \Z{\mathsf Z}

\def \sfS{\mathsf S}
\def \F{\mathcal F}
\def \calP{\mathcal P}
\def \calS{\mathcal S}
\def \calH{\mathcal H}

\def \reals {\mathbb R}
\def \supp {\mathrm{supp}}
\def \dualK{K^\star_\mu}
\def \Id{\mathrm{Id}}
\def \rnd{\frac{\dd\nu}{\dd\mu}}
\def \allones{\mathbf{1}}

\begin{document}
\title{Contraction of R\'enyi Divergences\\for Discrete Channels}
\author{Adrien~Vandenbroucque,
        Amedeo~Roberto~Esposito,
        and~Michael~Gastpar
\thanks{A. Vandenbroucque and M. Gastpar are with the School of Computer and Communication Sciences, EPFL, Lausanne, Switzerland, e-mails: \{adrien.vandenbroucque, michael.gastpar\}@epfl.ch.}
\thanks{A. R. Esposito is with the Okinawa Institute of Science and Technology (OIST), Onna-son, Japan, e-mail: amedeo.esposito@oist.jp.}
\thanks{This paper was presented in part at the \textit{2026 IEEE International Symposium on Information Theory}, Guangzhou, China}}

\maketitle

\begin{abstract}
    We investigate Strong Data-Processing Inequality (SDPI) constants for R\'enyi Divergences on finite spaces. We study their dependence on the R\'enyi order $\alpha$, proving that they are non-decreasing and that their scaling by $(\alpha-1)$ is convex for $\alpha\geq1$. We also identify several support restrictions on the probability measures involved in determining these constants. In particular, in the distribution-independent setting, measures supported on a common set of at most two points suffice to evaluate the R\'enyi-SDPI constant. For $\alpha\in[0,1]$, we further prove equality with the $\chi^2$-SDPI constant, while at order infinity we obtain a closed-form expression. In order to link contraction over product spaces to contraction along individual coordinates, we provide tensorisation bounds for arbitrary product channels analogous to those known for $\varphi$-Divergences. At finite orders, the Rényi-SDPI constants are bounded above and below through comparisons with the $\chi^2$-Divergence and Hellinger Divergences, with sharpness established in multiple cases. At order infinity, we instead relate these constants to the contraction of Total Variation Distance. Finally, our findings are applied to local differential privacy (LDP) and the analysis of Markov chains. This yields sharp contraction guarantees for pure-LDP mechanisms, and connects R\'enyi-LDP to R\'enyi-SDPI constants. For Markov chains, we derive finite-time convergence bounds and exhibit a family of chains for which R\'enyi-SDPIs improve on classical $\chi^2$-based bounds by arbitrarily large factors.
\end{abstract}


\begin{table*}[p]
\begingroup
\fontsize{7.0}{7.5}\selectfont
\setlength{\abovedisplayskip}{0.25\baselineskip}
\setlength{\belowdisplayskip}{0.25\baselineskip}
\setlength{\abovedisplayshortskip}{0.25\baselineskip}
\setlength{\belowdisplayshortskip}{0.25\baselineskip}
\setlength{\jot}{0.12\baselineskip}
\setlength{\columnsep}{1.5em}
\setlength{\multicolsep}{0.2\baselineskip}
\setlength{\fboxsep}{0.55em}
\crefname{equation}{Eq.}{Eqs.}
\Crefname{equation}{Eq.}{Eqs.}
\crefname{theorem}{Thm.}{Thms.}
\Crefname{theorem}{Thm.}{Thms.}
\crefname{proposition}{Prop.}{Props.}
\Crefname{proposition}{Prop.}{Props.}
\crefname{corollary}{Cor.}{Cors.}
\Crefname{corollary}{Cor.}{Cors.}
\newcommand{\summarysection}[1]{%
	\par\vspace{0.65\baselineskip}%
	\noindent{\fontsize{8.1}{8.8}\selectfont\bfseries #1}%
	\par\vspace{0.30\baselineskip}}
\newcommand{\summaryentry}[2]{%
	\par\vspace{0.40\baselineskip}\noindent
	\textbf{#1}~(\Cref{#2}):\ }
\newcommand{\summaryentrymixed}[3]{%
	\par\vspace{0.40\baselineskip}\noindent
	\textbf{#1}~(\Cref{#2}; \Cref{#3}):\ }

\noindent\fbox{%
\begin{minipage}[t]{\dimexpr\textwidth-2\fboxsep-2\fboxrule\relax}
	\vspace{0.12\baselineskip}
	\centering{\fontsize{10}{11}\selectfont\textbf{Summary}}\par
	\vspace{0.85\baselineskip}
	\fontsize{7.0}{7.5}\selectfont

	\begin{multicols}{2}
	\raggedcolumns
	\raggedright
	\fontsize{7.0}{7.5}\selectfont

	\summarysection{Definitions}

	\begin{minipage}{\linewidth}
	\summaryentry{Distribution-dependent contraction coefficient}{eq:dd_sdpi_constant}
	For $\alpha\in(0,\infty]$,
	\begin{equation*}
		\eta_\alpha(\mu,K)=\sup_{\substack{\nu\in\calP(\X):\\0<D_\alpha(\nu\|\mu)<\infty}}\frac{D_\alpha(\nu K\|\mu K)}{D_\alpha(\nu\|\mu)}.
	\end{equation*}
	At $\alpha=0$, the coefficient is defined using reverse KL divergence.

	\summaryentry{Distribution-independent contraction coefficient}{eq:di_sdpi_constant}
	\begin{equation*}
		\eta_\alpha(K)=\sup_{\mu\in\calP(\X)}\eta_\alpha(\mu,K), \qquad \alpha\in[0,\infty].
	\end{equation*}
	Throughout, write $K_x\triangleq K(\cdot|x)=\delta_xK$.
	\end{minipage}

	\summarysection{Dependence on the Order}

	\begin{minipage}{\linewidth}
	\summaryentrymixed{Monotonicity}{thm:monotonicity,prop:distribution_independent_sdpi_equal_chi_square_sdpi}{corollary:distribution_independent_order_properties}
	\begin{align*}
		\eta_\alpha(\mu,K)&\leq\eta_\beta(\mu,K), &1\leq\alpha\leq\beta\leq\infty,\\
		\eta_\alpha(K)&\leq\eta_\beta(K), &0\leq\alpha\leq\beta\leq\infty.
	\end{align*}

	\summaryentry{Scaled convexity}{thm:convexity,corollary:distribution_independent_order_properties}
	\begin{equation*}
		\begin{aligned}
			\alpha&\longmapsto(\alpha-1)\eta_\alpha(\mu,K),\\
			\alpha&\longmapsto(\alpha-1)\eta_\alpha(K)
		\end{aligned}
		\qquad\text{are convex on }(1,\infty).
	\end{equation*}

	\summaryentry{Continuity in the order}{corollary:continuity_in_order,corollary:distribution_independent_order_properties}
	\begin{equation*}
		\begin{aligned}
			\alpha\mapsto\eta_\alpha(\mu,K)&\text{ is continuous on }[1,\infty],\\
			\alpha\mapsto\eta_\alpha(K)&\text{ is continuous on }(1,\infty].
		\end{aligned}
	\end{equation*}
	\end{minipage}

	\summarysection{Simplified Representations}

	\begin{minipage}{\linewidth}
	\summaryentry{Boundary-support reduction}{thm:boundary_nu}
	For $\alpha\in[2,\infty)$,
	\begin{equation*}
		\eta_\alpha(\mu,K)=\max_{\substack{\nu\in\calP(\X):\\ \supp(\nu)\subsetneq\supp(\mu)}}\frac{D_\alpha(\nu K\|\mu K)}{D_\alpha(\nu\|\mu)}.
	\end{equation*}

	\summaryentry{Binary-support reduction}{prop:eta_parameter_renyi_binary}
	In the optimisation defining $\eta_\alpha(K)$:
	\par\vspace{0.12\baselineskip}
	\noindent For $\alpha\in[0,\infty]$, it suffices to consider pairs
	$(\nu,\mu)$ supported on at most two common points.
	\par\vspace{0.12\baselineskip}
	\noindent For $\alpha\in[2,\infty]$, one may further take $\nu$ to be
	a point mass and $\mu$ to be supported on that point and one
	additional point.
	\summaryentry{Distribution-dependent order $\infty$}{prop:infty_sdpi_achieved_on_conditionals}
	\begin{equation*}
		\eta_\infty(\mu,K)=\max_{\substack{A\in\Sigma_\X:\\0<\mu(A)<1}}\frac{D_\infty(\mu_{|A}K\|\mu K)}{D_\infty(\mu_{|A}\|\mu)}.
	\end{equation*}

	\summaryentry{Distribution-independent order $\infty$}{prop:infty_renyi_sdpi_closed_form}
	\begin{equation*}
		\eta_\infty(K)=\max_{x,x'\in\X}\left\{1-\min_{y\in\supp(K(\cdot|x))}\frac{K(y|x')}{K(y|x)}\right\}.
	\end{equation*}

	\summaryentry{Trivial contraction}{prop:non_trivial_sdpi_characterisation}
	For $\alpha\in[0,\infty]$,
	\begin{equation*}
		\eta_\alpha(K)=1\iff\! \begin{cases} \exists\,x,x':\ \supp(K_x)\cap\supp(K_{x'})=\emptyset, &\alpha\leq1,\\ \exists\,x,x':\ \supp(K_x)\neq\supp(K_{x'}), &\alpha>1. \end{cases}
	\end{equation*}
	\end{minipage}

	\summarysection{Product Kernels and Tensorisation Bounds}

	\begin{minipage}{\linewidth}
	Set
	\begin{equation*}
		\mu\triangleq\bigotimes_{i=1}^n\mu_i, \qquad K\triangleq\bigotimes_{i=1}^nK_i.
	\end{equation*}
	For any joint law $\rho$, let $\rho_T$ denote its marginal on the
	coordinates in $T$.

	\summaryentry{Samorodnitsky-type inequality}{prop:renyi_samorodnitsky}
	Let
	\begin{equation*}
		w_T\triangleq \prod_{i\in T}\eta_i \prod_{j\in[n]\setminus T}(1-\eta_j).
	\end{equation*}
	For $\nu\ll\mu$ and $\alpha\in(1,\infty]$, take
	$\eta_i=\eta_\alpha(\mu_i,K_i)$. Then
	\begin{equation*}
		D_\alpha(\nu K\|\mu K) \leq\sum_{T\subseteq[n]}w_TD_\alpha(\nu_T\|\mu_T).
	\end{equation*}
	\end{minipage}

	\columnbreak

	\fontsize{7.0}{7.5}\selectfont
	\setlength{\abovedisplayskip}{0.08\baselineskip}
	\setlength{\belowdisplayskip}{0.08\baselineskip}
	\setlength{\abovedisplayshortskip}{0.08\baselineskip}
	\setlength{\belowdisplayshortskip}{0.08\baselineskip}
	\setlength{\jot}{0.07\baselineskip}
	\renewcommand{\summarysection}[1]{%
		\par\vspace{0.35\baselineskip}%
		\noindent{\fontsize{8.1}{8.8}\selectfont\bfseries #1}%
		\par\vspace{0.16\baselineskip}}
	\renewcommand{\summaryentry}[2]{%
		\par\vspace{0.15\baselineskip}\noindent
		\textbf{#1}~(\Cref{#2}):\ }
	\renewcommand{\summaryentrymixed}[3]{%
		\par\vspace{0.15\baselineskip}\noindent
		\textbf{#1}~(\Cref{#2}; \Cref{#3}):\ }

	\begin{minipage}{\linewidth}
	\vspace*{0.80\baselineskip}
	For arbitrary $\nu,\gamma$ when $\alpha\in(0,1)$, and for $\nu\ll\gamma$ when $\alpha\in(1,\infty]$, take
	$\eta_i=\eta_\alpha(K_i)$. Then
	\begin{equation*}
		D_\alpha(\nu K\|\gamma K) \leq\sum_{T\subseteq[n]}w_TD_\alpha(\nu_T\|\gamma_T).
	\end{equation*}

	\vspace{0.25\baselineskip}
	\summaryentry{Product-coefficient bounds}{prop:tensorisation_of_sdpi_constant}
	For $\alpha\in(1,\infty]$,
	\begin{equation*}
		\max_{i\in[n]}\eta_\alpha(\mu_i,K_i) \leq\eta_\alpha(\mu,K) \leq1-\prod_{i=1}^n\bigl(1-\eta_\alpha(\mu_i,K_i)\bigr).
	\end{equation*}
	For $\alpha\in(0,1)\cup(1,\infty]$,
	\begin{equation*}
		\max_{i\in[n]}\eta_\alpha(K_i) \leq\eta_\alpha(K) \leq1-\prod_{i=1}^n\bigl(1-\eta_\alpha(K_i)\bigr).
	\end{equation*}
	\end{minipage}

	\summarysection{Comparisons with Other Contraction Coefficients}

	\begin{minipage}{\linewidth}
	\summaryentry{Equality with $\chi^2$-contraction}{prop:distribution_independent_sdpi_equal_chi_square_sdpi}
	\begin{equation*}
		\eta_\alpha(K)=\eta_{\chi^2}(K),\qquad \alpha\in[0,1].
	\end{equation*}

	\summaryentry{Upper bounds via $\chi^2$-contraction}{prop:ub_on_eta_alpha_by_eta_chi_squared}
	Let
	\begin{equation*}
		C_\alpha(t)\triangleq \frac{t^\alpha-\alpha t+\alpha-1}{(t-1)^2}\quad(t\in\reals_+\setminus\{1\}), \qquad C_\alpha(1)\triangleq\frac{\alpha(\alpha-1)}{2}.
	\end{equation*}
	Then, for $0<\alpha<1$ and $\supp(\mu)=\X$,
	\begin{equation*}
		\eta_\alpha(\mu,K)\leq\frac{\eta_{\chi^2}(\mu,K)}{\alpha\min_{x\in\X}\mu(x)}.
	\end{equation*}
	For $1\leq\alpha\leq2$,
	\begin{equation*}
		\eta_\alpha(\mu,K)\leq\frac{\eta_{\chi^2}(\mu,K)}{\min_{x\in\supp(\mu)}\mu(x)}.
	\end{equation*}
	For $\alpha>2$,
	\begin{equation*}
		\eta_\alpha(\mu,K)\leq\frac{C_\alpha\!\left(\bigl(\min_{y\in\supp(\mu K)}\mu K(y)\bigr)^{-1}\right)\eta_{\chi^2}(\mu,K)}{(\alpha-1)\min_{x\in\supp(\mu)}\mu(x)}.
	\end{equation*}

	\summaryentry{Comparison with Hellinger contraction}{corollary:sdpi_ub_lb_hellinger}
	\begin{align*}
		\eta_{\calH_\alpha}(\mu,K)&\leq\eta_\alpha(\mu,K),& \eta_{\calH_\alpha}(K)&\leq\eta_\alpha(K),&&\alpha>1,\\
		\eta_\alpha(\mu,K)&\leq\eta_{\calH_\alpha}(\mu,K),& \eta_\alpha(K)&\leq\eta_{\calH_\alpha}(K),&&0<\alpha<1.
	\end{align*}

	\summaryentry{Upper bounds via Hellinger contraction}{prop:sdpi_ub_hellinger}
	For $\alpha>1$,
	\begin{align*}
		\eta_\alpha(\mu,K)&\leq\frac{\log\left(1+(\alpha-1)M_\alpha(\mu,K)\right)}{\log\bigl(1+(\alpha-1)M_\alpha(\mu,K)/\eta_{\calH_\alpha}(\mu,K)\bigr)},\\
		\eta_\alpha(K)&\leq\frac{\log\left(1+(\alpha-1)M_\alpha(K)\right)}{\log\bigl(1+(\alpha-1)M_\alpha(K)/\eta_{\calH_\alpha}(K)\bigr)}.
	\end{align*}
	where
	\begin{equation*}
		M_\alpha(\mu,K)\triangleq\max_{x: \mu(x)>0}\calH_\alpha(K_x\|\mu K), \; M_\alpha(K)\triangleq\sup_{x,x'}\calH_\alpha(K_x\|K_{x'}).
	\end{equation*}

	\summaryentry{Order $\infty$ and Total Variation Distance}{prop:bounds_on_renyi_infty_sdpi_by_tv,prop:eta_tv_lower_bounds_eta_inf}
	If $\mu$ has full support, then
	\begin{equation*}
		\eta_{\sf TV}(\mu,K)\leq\eta_\infty(\mu,K)\leq\frac{\eta_{\sf TV}(\mu,K)}{\min_{y\in\supp(\mu K)}\mu K(y)}.
	\end{equation*}
    Moreover,
    \begin{equation*}
        \eta_{\sf TV}(K)\leq\eta_\infty(K).
    \end{equation*}
	\end{minipage}

	\vspace{0.20\baselineskip}
	\summarysection{Applications}

	\begin{minipage}{\linewidth}
	\summaryentry{Local differential privacy}{thm:renyi_ldp_implies_contraction}
	For $\alpha\in[2,\infty]$,
	\begin{align*}
		(\alpha,\varepsilon)\text{-RLDP}&\Longrightarrow\eta_\alpha(K)\leq1-e^{-\varepsilon},\\
		\varepsilon\text{-LDP}&\Longleftrightarrow\eta_\infty(K)\leq1-e^{-\varepsilon}.
	\end{align*}

	\summaryentry{Sharp contraction under pure LDP}{thm:ldp_sharp_renyi_contraction}
	Every $\varepsilon$-LDP mechanism $K$ satisfies, for $\alpha\in[0,\infty]$,
	\begin{equation*}
		\eta_\alpha(K)\leq\eta_\alpha\!\left(\mathrm{BSC}\!\left(\frac1{e^\varepsilon+1}\right)\right),
	\end{equation*}
	with equality for the binary randomised response mechanism.

	\summaryentry{Asymptotic Markov-chain contraction}{corollary:asymptotic_renyi_contraction}
	For irreducible, aperiodic, reversible $K$ with stationary law $\pi$,
	\begin{equation*}
		\lim_{t\to\infty}\eta_\alpha(\pi,K^t)^{1/t}=\eta_{\chi^2}(\pi,K),\qquad \alpha\in(0,\infty).
	\end{equation*}

	\summaryentry{Finite-time Markov-chain contraction}{prop:mcmt}
	If $\pi K=\pi$, $\nu\ll\pi$, $t\geq1$, and $\alpha\in(0,1)\cup(1,\infty)$:
	\begin{equation*}
		\calH_\alpha(\nu K^t\|\pi)\leq\frac{\bigl(1+(\alpha-1)\calH_\alpha(\nu\|\pi)\bigr)^{\eta_\alpha(\pi,K)^t}-1}{\alpha-1}.
	\end{equation*}
	\end{minipage}

	\end{multicols}
	\vspace{0.6\baselineskip}
\end{minipage}}
\endgroup
\end{table*}

\section{Introduction}\noindent
    The Data-Processing Inequality (DPI) is a fundamental result in information theory, formalising the fact that post-processing cannot increase information. Strong Data-Processing Inequalities (SDPIs) refine this qualitative statement by quantifying the amount of information dissipated by a given channel. The study of such quantitative contraction has a long history, with early results for relative entropy and subsequent extensions to more general divergence measures~\cite{ahlswede1976spreading,cohen1993relative,choi1994equivalence}. Over the past decade, this line of research has developed into a more systematic theory, particularly for $\varphi$-Divergences~\cite{raginsky2016strong,polyanskiy2017strong,makur2020comparison,Polyanskiy2025}. Beyond these theoretical developments, SDPIs have found applications to local differential privacy~\cite{zamanlooy2023strong,asoodeh2024contraction}, information reconstruction~\cite{polyanskiy2020application,gu2023non}, and concentration inequalities~\cite{esposito2024concentration}. More recently, information contraction has also been used to derive formal guarantees for machine-unlearning procedures~\cite{chien2024langevin,koloskova2025certified}.

    R\'enyi Divergences have likewise become increasingly prominent in recent years. Their order $\alpha$ allows different features of the probability measures to be emphasised, with higher orders being particularly sensitive to discrepancies in their tails. Together with their additivity over product measures, this flexibility makes R\'enyi Divergences well suited to problems involving tail behaviour and repeated observations~\cite{van2014renyi}. These features underlie their use in differential privacy~\cite{Mironov2017}, the convergence analysis of Langevin algorithms~\cite{vempala2019rapid,chewi2025analysis}, hypothesis testing~\cite{tomamichel2017operational,bruno2026finite}, estimation~\cite{esposito2024lower}, the generalisation behaviour of learning algorithms~\cite{EspositoGI:2021}, and universal prediction~\cite{bondaschi2023alphanml}. The breadth of these applications motivates a better understanding of how R\'enyi Divergences contract under stochastic transformations and, more specifically, of their SDPI constants.

    Despite this growing interest, the contraction of R\'enyi Divergences remains considerably less understood than that of $\varphi$-Divergences. Recent work by Jin et al.~\cite{jin2024properties} established a universal lower bound through the $\chi^2$-SDPI constant and studied the optimisation defining the contraction coefficient at order $\alpha=2$. Grosse et al.~\cite{grosse2025bounds} developed Pinsker-type bounds for restricted classes of probability measures and applied them to privacy amplification, while Abawonse et al.~\cite{abawonse2026generalized} obtained related functional inequalities in their study of noisy functions. Together, these works provide several important results, but a systematic understanding of the contraction properties of R\'enyi Divergences remains incomplete.

    Motivated by this gap, we study the properties of R\'enyi-SDPI constants for discrete channels. Particular attention is devoted to their dependence on the order $\alpha$ of the R\'enyi Divergence, to the probability measures that determine their value, and to their behaviour in product spaces. Parallels with the contraction of $\varphi$-Divergences complement this analysis, yielding comparisons that provide bounds on R\'enyi-SDPI constants and facilitate their applications.

    \subsection{Contributions and Organisation}\noindent
    \Cref{section:background_and_definitions} introduces the notation and definitions used throughout the paper. In~\cref{sec:properties_exact_characterisations}, we first derive a functional representation of the R\'enyi-SDPI constant, from which we establish its monotonicity in $\alpha$ and the convexity of its scaling by $\alpha-1$ for orders $\alpha\geq1$. We then identify restrictions on the probability measures that need to be considered when determining these constants, leading to reduced optimisation problems and simpler representations. In the distribution-independent setting, we further prove equality with the $\chi^2$-SDPI constant for $\alpha\in[0,1]$ and characterise exactly when strict contraction fails. The section concludes with tensorisation bounds for product channels.

    In~\cref{sec:bounds_on_contraction_coefficient}, we establish bounds on R\'enyi-SDPI constants. At finite orders, we complement the universal lower bound provided by the $\chi^2$-SDPI constant with upper bounds, and exploit the one-to-one relationship with Hellinger Divergences to obtain further comparisons. At order infinity, where these methods cannot be applied, we instead derive upper and lower bounds in terms of Total Variation Distance.

    Finally, \cref{sec:applications} applies these results to local differential privacy (LDP) and the convergence of Markov chains. We first relate privacy guarantees expressed in terms of R\'enyi Divergences to contraction at the same order. We then derive an upper bound on the R\'enyi-SDPI constant of any pure-LDP mechanism and prove it is attained by a randomised response mechanism. For reversible Markov chains, we show that R\'enyi Divergences with finite order contract asymptotically at the same rate as the $\chi^2$-Divergence. We also derive finite-time convergence bounds and exhibit a family of chains for which they improve upon the classical $\chi^2$-SDPI bound.

\section{Background and Definitions}\label{section:background_and_definitions}\noindent
    In the following, we consider spaces with finite alphabets. Nonetheless, we decide to use standard measure-theoretic notation throughout, as some of the results naturally extend to more general spaces.

    For a space $\X$, it is assumed that a $\sigma$-algebra $\Sigma_\X$ is associated to it, rendering it a measure space. Since the spaces under consideration are finite, we can assume $\Sigma_\X$ to be the power set of $\X$. Indeed, in case it is not, some symbols in $\X$ can be merged to yield a new measure space with this property. We always assume $|\X|\geq2$. We denote by $\calP(\X)$ the set of probability measures on $\X$ and by $\F(\X)$ the set of all real-valued functions on $\X$. The expectation of a function $f\in\F(\X)$ w.r.t. a probability measure $\mu\in\calP(\X)$ is ${\mean{\mu}{f}\triangleq\int_{\X}f\dd\mu}$ and its $L^\alpha$-norm is ${\|f\|_{L^\alpha(\mu)}\triangleq\big(\mean{\mu}{|f|^\alpha}\big)^{\frac1\alpha}}$ for $\alpha\in[1,\infty)$, with ${\|f\|_{L^\infty(\mu)}\triangleq\mathrm{ess}\sup_\mu|f|}$. For $\mu\in\calP(\X)$ and $\alpha\in[1,\infty]$, we denote by $L^\alpha(\mu)$ the space of functions identified up to $\mu$-almost-everywhere equality, with finite $L^\alpha$-norm. The constant function taking value 1 is denoted by $\allones$. Given two probability measures $\nu, \mu$, we write $\nu\ll\mu$ to denote that $\nu$ is absolutely continuous w.r.t. $\mu$, in which case there exists a Radon-Nikodym derivative $\frac{\dd\nu}{\dd\mu}$. The natural logarithm function is denoted by $\log$. When handling indices, we use $[n]\triangleq\{1, \dots, n\}$. For $\alpha\in[1,\infty]$, we denote its H\"older conjugate by $\alpha^\prime$, so that $\frac1\alpha+\frac1{\alpha^\prime}=1$. The Binary Symmetric Channel with crossover probability $\varepsilon$ is written as $\mathrm{BSC}(\varepsilon)$, and the Binary Erasure Channel (BEC) with erasure probability $\delta$ is denoted as $\mathrm{BEC}(\delta)$.

    \subsection{Markov Kernels}\label{sec:markov_kernels}\noindent
        Hereafter, a channel will be viewed as a Markov kernel, or stochastic transformation.
        \begin{definition}
            A Markov kernel $K:\Y\times\X\to[0, 1]$ specifies transition probabilities from elements of a set $\X$ to elements of a set $\Y$, and in particular
        \begin{enumerate}
            \item $\forall x\in\X$, the mapping $y\mapsto K(y|x)$ is a probability measure,
            \item $\forall y\in\Y$, the mapping $x\mapsto K(y|x)$ is a $\Sigma_\X$-measurable real-valued function.
        \end{enumerate}
        \end{definition}
        The set of all such kernels is denoted by $\calP(\Y|\X)$.

        A Markov kernel acts on probability measures $\mu\in\calP(\X)$ from the left as
        \begin{equation*}
            \mu K (y) = \sum_{x\in\X} K(y|x)\mu(x), \quad \text{for }y\in\Y,
        \end{equation*}
        so that $\mu K\in\calP(\Y)$. It can also be viewed as acting on functions $f\in\F(\Y)$ from the right as
        \begin{equation*}
            Kf(x) = \sum_{y\in\Y} f(y)K(y|x), \quad \text{for }x\in\X,
        \end{equation*}
        so that $Kf\in\F(\X)$. In terms of random variables, if $(X,Y)\in\X\times\Y$ is a random pair such that $X\sim\mu$ and $Y|X=x\sim K(\cdot|x)$, then $Y\sim\mu K$ and, for any $f\in\F(\Y)$,
        \begin{equation*}
            Kf(x)=\mathbb{E}\big[f(Y)| X=x\big],
            \qquad\text{for } x\in\supp(\mu).
        \end{equation*}
        For any pair $(\mu, K)\in\calP(\X)\times\calP(\Y|\X)$, there exists a dual kernel $\dualK\in\calP(\X|\Y)$ satisfying
        \begin{equation}\label{eq:adjoint_relation}
            \ip{Kf}{g}_{\mu} = \ip{f}{\dualK g}_{\mu K}, \quad\forall f\in\F(\Y)\text{ and }\forall g\in\F(\X)
        \end{equation}
        where we have used the inner-product notation ${\ip{f}{g}_\mu\triangleq \mean{\mu}{f\cdot g}=\int_\X f(x)g(x)\dd\mu(x)}$. \Cref{eq:adjoint_relation} determines the dual kernel up to a $\mu K$-null set: if $L,\widetilde L\in\calP(\X|\Y)$ both satisfy~\cref{eq:adjoint_relation}, then $L(\cdot|y)=\widetilde L(\cdot|y)$ for $\mu K$-almost every $y\in\Y$. In discrete settings like ours, the transition probabilities of this dual kernel are given by
        \begin{equation*}
            \dualK(x|y) = \frac{\mu(x)K(y|x)}{\mu K(y)},\qquad\forall x\in\X\text{ and }\forall y\in\supp(\mu K).
        \end{equation*}
        One useful consequence of~\cref{eq:adjoint_relation} is that for any non-negative $f\in\F(\X)$,
        \begin{equation}\label{eq:markov_kernel_mean_preservation}
            \|K_\mu^\star f\|_{L^1(\mu K)} = \ip{\allones}{K_\mu^\star f}_{\mu K} = \ip{K\allones}{ f}_{\mu} = \ip{\allones}{ f}_{\mu} = \|f\|_{L^1(\mu)}.
        \end{equation}
        We also note that for a pair $(\mu, K)$ and any other probability measure $\nu\ll\mu$, we have
        \begin{equation}\label{eq:dual_kernel_formula}
            \dualK f = \frac{\dd(\nu K)}{\dd(\mu K)}\quad \mu K\text{-a.e.},
        \end{equation}
        where $f=\frac{\dd\nu}{\dd\mu}$, see~\cite[Lemma 1]{esposito2024contraction}.

    \subsection{Divergences and (Strong) Data-Processing Inequalities}\noindent
        Let us define two families of divergences which will be central to this work, starting with R\'enyi Divergences~\cite{van2014renyi}.
        \begin{definition}
            For two probability measures $\nu,\mu\in\calP(\X)$, let $\lambda=\frac12(\nu+\mu)$. The R\'enyi divergence of order $\alpha\in(0,1)\cup(1,\infty)$ from $\nu$ to $\mu$ is defined as
            \begin{equation*}
                D_\alpha(\nu\|\mu) = \frac{1}{\alpha-1}\log\mean{\lambda}{\left(\frac{\dd\nu}{\dd\lambda}\right)^\alpha\left(\frac{\dd\mu}{\dd\lambda}\right)^{1-\alpha}},
            \end{equation*}
            where, for $\alpha>1$, the right-hand side is understood to equal $+\infty$ whenever $\nu\not\ll\mu$. The orders $0$, $1$ and $\infty$ are defined by continuous extensions.
        \end{definition}
        The KL Divergence, denoted as $D_{\sf KL}(\nu\|\mu)$, corresponds to the limit $\alpha\to1$.

        The second family of divergences we consider are the $\varphi$-Divergences~\cite{csiszar1967information}.
        \begin{definition}
            Let $\varphi:(0,\infty)\to\reals$ be a convex function with $\varphi(1)=0$. For two probability measures $\nu,\mu\in\calP(\X)$, let $\lambda=\frac12(\nu+\mu)$. The $\varphi$-Divergence from $\nu$ to $\mu$ is defined as
            \begin{equation*}
                D_\varphi(\nu\|\mu)=\mean{\lambda}{\frac{\dd\mu}{\dd\lambda}\varphi\left(\frac{\dd\nu/\dd\lambda}{\dd\mu/\dd\lambda}\right)},
            \end{equation*}
            with the conventions $\varphi(0)\triangleq\lim_{t\downarrow0}\varphi(t)$ and $0\varphi\left(\frac{a}{0}\right)\triangleq a\lim_{t\to\infty}\frac{\varphi(t)}{t}$ for $a>0$.
        \end{definition}

        A variety of well-known divergences belong to the family of $\varphi$-Divergences, notably:
        \begin{itemize}
            \item The KL Divergence, for which $\varphi(t) = t\log t$,
            \item The reverse KL Divergence, for which $\varphi(t) = -\log t$,
            \item The Total Variation Distance, denoted as $\|\nu-\mu\|_{\sf TV}$, for which $\varphi(t) = \frac12|t-1|$,
            \item The Hellinger Divergence of order $\alpha\in(0,1)\cup(1,\infty)$, denoted as $\calH_\alpha(\nu\|\mu)$, for which $\varphi(t) = \frac{1}{\alpha-1}(t^\alpha-1)$. The case $\alpha=2$ corresponds to the $\chi^2$-Divergence, denoted as $\chi^2(\nu\|\mu)$.
        \end{itemize}
        All the divergences defined above satisfy the Data-Processing Inequality (DPI), which for a divergence $D(\cdot\|\cdot)$ and a Markov kernel $K$ reads
        \begin{equation*}
            D(\nu K\|\mu K) \leq D(\nu\|\mu)
        \end{equation*}
        for all $\nu, \mu\in\calP(\X)$. Typically, the inequality is strict and one can give a more quantitative statement. For a fixed pair $(\mu, K)$, the DPI can be tightened by considering the associated distribution-dependent SDPI constant~\cite[Definition III.1]{raginsky2016strong}~\cite[Equation 5]{polyanskiy2017strong}
        \begin{equation}\label{eq:dd_sdpi_constant}
            \eta_{D}(\mu, K) = \sup_{\substack{\nu\in\calP(\X): \\0<D(\nu\|\mu)<\infty}} \frac{D(\nu K\|\mu K)}{D(\nu\|\mu)},
        \end{equation}
        yielding the stronger inequality ${D(\nu K\|\mu K) \leq \eta_{D}(\mu, K)D(\nu\|\mu)}$ for all $\nu\in\calP(\X)$. Whenever the constraint set in~\cref{eq:dd_sdpi_constant} is empty, we adopt the convention that $\eta_D(\mu,K)=0$.

        One can similarly define a quantity that only depends on $K$~\cite[Definition III.1]{raginsky2016strong}~\cite[Equation 6]{polyanskiy2017strong}, leading to the distribution-independent SDPI constant
        \begin{equation}\label{eq:di_sdpi_constant}
            \eta_{D}(K) = \sup_{\mu\in\calP(\X)}\eta_D(\mu, K) = \sup_{\substack{\nu, \mu\in\calP(\X)\\0<D(\nu\|\mu)<\infty}} \frac{D(\nu K\|\mu K)}{D(\nu\|\mu)}.
        \end{equation}
        In the rest of the article, we use the following notation to refer to SDPI constants under various common divergences: $\eta_{\sf TV}$ for Total Variation Distance, $\eta_{\sf KL}$ for KL Divergence, $\eta_{\sf revKL}$ for reverse KL Divergence, $\eta_{\chi^2}$ for $\chi^2$-Divergence, $\eta_{\calH_\alpha}$ for Hellinger Divergence of order $\alpha$, and $\eta_\alpha$ for R\'enyi Divergence of order $\alpha$. We often refer to SDPI constants as ``contraction coefficients''.

        Let us finally comment on how R\'enyi-SDPI constants are handled for extended orders of R\'enyi Divergences:
        \begin{itemize}
            \item For $\alpha=1$, the definition $\eta_1(\mu, K)\triangleq \sup_{\substack{\nu\in\calP(\X):\\0<D_{\sf KL}(\nu\|\mu)<\infty}}\frac{D_{\sf KL}(\nu K\|\mu K)}{D_{\sf KL}(\nu\|\mu)}=\eta_{\sf KL}(\mu, K)$ is used,
            \item For $\alpha=\infty$, the definition $\eta_\infty(\mu, K)\triangleq \sup_{\substack{\nu\in\calP(\X):\\0<D_\infty(\nu\|\mu)<\infty}}\frac{D_\infty(\nu K\|\mu K)}{D_\infty(\nu\|\mu)}$ is used,
            \item For $\alpha=0$, the definition $\eta_0(\mu, K)\triangleq \sup_{\substack{\nu\in\calP(\X):\\0<D_{\sf KL}(\mu\|\nu)<\infty}}\frac{D_{\sf KL}(\mu K\|\nu K)}{D_{\sf KL}(\mu\|\nu)}= \eta_{\sf revKL}(\mu, K)$ is used.
        \end{itemize}
        \begin{remark}
            One could alternatively define $\eta_0(\mu, K)$ through $\sup_{\substack{\nu\in\calP(\X):\\0<D_0(\nu\|\mu)<\infty}}\frac{D_0(\nu K\|\mu K)}{D_0(\nu\|\mu)}$. The reverse KL Divergence is chosen here because of the skew-symmetry property of R\'enyi Divergences $\frac{1}{\alpha}D_\alpha(\nu\|\mu) = \frac1{1-\alpha}D_{1-\alpha}(\mu\|\nu)$ for $\alpha\in(0,1)$, which allows us to view the contraction of $D_\alpha(\nu\|\mu)$ as the contraction of $D_{1-\alpha}(\mu\|\nu)$. In the limit, we have $\lim_{\alpha\downarrow0}\frac{1}{\alpha}D_\alpha(\nu\|\mu)=D_{\sf KL}(\mu\|\nu)$, so that our definition retains the corresponding symmetry at the endpoints. This enables us, for example, to write $\eta_\alpha(K)=\eta_{1-\alpha}(K)$ for all $\alpha\in[0,1]$.
        \end{remark}
\section{Properties of the Contraction Coefficient}\label{sec:properties_exact_characterisations}
    \subsection{Functional Representation and Dependence on the Order}\label{sec:monotonicity}\noindent
        The order $\alpha$ is one of the defining features of R\'enyi Divergences. For fixed probability measures $\nu\ll\mu$, increasing the order places progressively greater emphasis on large values of the density ratio $\frac{\dd\nu}{\dd\mu}$. This gives the well-known inequalities
        \begin{equation*}
            D_{\sf KL}(\nu\|\mu)\leq D_\alpha(\nu\|\mu)\leq D_\infty(\nu\|\mu),\qquad\alpha\in(1,\infty),
        \end{equation*}
        so that higher orders provide increasingly stringent measures of discrepancy, with $D_\infty$ determined by the largest value of this density ratio~\cite[Theorem 3]{van2014renyi}.

        It is natural to ask whether this hierarchy is inherited by the corresponding SDPI constants. Such a conclusion does not follow directly from the monotonicity of R\'enyi Divergences. Indeed, for a fixed probability measure $\nu\ll\mu$, both the numerator and the denominator in the contraction ratio
        \begin{equation*}
            \frac{D_\alpha(\nu K\|\mu K)}{D_\alpha(\nu\|\mu)}
        \end{equation*}
        vary with the order, and their individual monotonicity yields no immediate comparison between the resulting ratios. Passing from these ratios to $\eta_\alpha(\mu, K)$ introduces a further difficulty, since the supremum over probability measures is taken separately at each order. It is therefore unclear a priori whether increasing the order makes a given pair $(\mu,K)$ appear more or less contractive. We show below that the associated SDPI constant is in fact non-decreasing in the order.

        Rather than comparing the contraction ratios directly, our approach relies on the following functional representation of $\eta_\alpha(\mu,K)$.
        \begin{proposition}\label{prop:renyi_sdpi_functional_form}
            Let $\alpha\in(1, \infty]$ and consider a pair $(\mu, K)\in\calP(\X)\times\calP(\Y|\X)$. Define $\mathcal{D}_\mu \triangleq \big\{f\in L^1(\mu): f\geq 0~\mu\text{-a.e.} \text{ and }\mean{\mu}{f}=1\big\}$, that is, every $f\in\mathcal D_\mu$ is the Radon--Nikodym derivative $f=\frac{\dd\nu}{\dd\mu}$ of a probability measure $\nu\ll\mu$. Then,
            \begin{equation*}
                \eta_\alpha(\mu, K) = \sup_{f\in\mathcal{D}_\mu\setminus\{\allones\}} \frac{\log\|K_\mu^\star f\|_{L^\alpha(\mu K)}}{\log\|f\|_{L^\alpha(\mu)}},
            \end{equation*}
            and the supremum is understood to be zero when its constraint set is empty. Equivalently, for any $f\in\mathcal{D}_\mu$,
            \begin{equation}\label{eq:renyi_sdpi_functional_inequality}
                \left\|K_\mu^\star f\right\|_{L^\alpha(\mu K)} \leq \left\|f\right\|_{L^\alpha(\mu)}^{\eta_\alpha(\mu, K)}.
            \end{equation}
            \appendixproofref{appendix:proof_renyi_sdpi_functional_form}
        \end{proposition}
        \begin{remark}
            In fact, \cref{eq:renyi_sdpi_functional_inequality} extends from probability densities to arbitrary non-negative functions (cf.~\cite[Equation 3]{abawonse2026generalized}). More precisely, for any non-negative function $f\in L^\alpha(\mu)$ satisfying $\|f\|_{L^1(\mu)}>0$ and any $\lambda\in[\eta_\alpha(\mu,K),1]$,
            \begin{equation}\label{eq:renyi_sdpi_functional_inequality_non_negative_functions}
                \left\|K_\mu^\star f\right\|_{L^\alpha(\mu K)}\leq\left\|f\right\|_{L^\alpha(\mu)}^\lambda\left\|f\right\|_{L^1(\mu)}^{1-\lambda}.
            \end{equation}
            Indeed, applying~\cref{eq:renyi_sdpi_functional_inequality} to the density $f/\|f\|_{L^1(\mu)}$ and rearranging gives~\cref{eq:renyi_sdpi_functional_inequality_non_negative_functions} for $\lambda=\eta_\alpha(\mu, K)$. Since $\|f\|_{L^\alpha(\mu)}\geq\|f\|_{L^1(\mu)}$, increasing $\lambda$ can only increase the right-hand side.
        \end{remark}
        With this functional formulation in hand, we can now show that the hierarchy of R\'enyi Divergences is inherited by the corresponding distribution-dependent SDPI constants.
        \begin{theorem}\label{thm:monotonicity}
            Consider a pair $(\mu,K)\in\calP(\X)\times\calP(\Y|\X)$ and let $1\leq\alpha\leq\beta\leq\infty$. Then, $\eta_\alpha(\mu,K)\leq\eta_\beta(\mu,K)$.
            \appendixproofref{appendix:proof_monotonicity}
        \end{theorem}
        \begin{remark}
            The restriction to orders $\alpha\geq1$ in~\cref{thm:monotonicity} is necessary. On $(0,1)$, the SDPI constant $\eta_\alpha(\mu,K)$ can instead decrease with $\alpha$, as illustrated numerically in~\cref{fig:dd_bounds_on_renyi_sdpi}.
        \end{remark}
        Since~\cref{thm:monotonicity} plays a central role in the sequel, we give a high-level overview of its proof.

        Rather than comparing the contraction coefficients at two orders directly, we establish a comparison involving a third order. Fix $1<\alpha<\beta<\infty$ and set $\gamma\triangleq 1+\frac{\alpha}{\beta^\prime}\in(\alpha,\beta)$. We show in~\cref{lemma:order_insertion_monotonicity} that
        \begin{equation}\label{eq:main_text_three_order_inequality}
            \eta_\gamma(\mu,K)\leq\max\big\{\eta_\alpha(\mu,K),\eta_\beta(\mu,K)\big\}.
        \end{equation}
        To obtain this inequality, we employ a first-order optimality condition (\cref{lemma:first_order_optimality}) to derive a relation between a suitable probability density and its image under $K_\mu^\star$. Weighted AM--GM turns this relation into a lower bound (\cref{lemma:first_order_balance_identity}), while H\"older's inequality with conjugate exponents $\beta$ and $\beta^\prime$ provides a complementary upper bound on the same quantity by separating it into two norm estimates. These are controlled using~\cref{eq:renyi_sdpi_functional_inequality_non_negative_functions}: one by applying the functional inequality at order $\beta$ to a suitable power of the density, and the other by applying it directly at order $\alpha$. The choice of $\gamma$ is key: it ensures that both factors induced by H\"older's inequality reduce to quantities at order $\alpha$, allowing the lower and upper bounds to be combined.

        The three-order comparison from~\cref{eq:main_text_three_order_inequality} drives the monotonicity argument by allowing an upper bound at two orders to be propagated to a suitable order between them. A contradiction argument based on this property shows that, for every order $\beta$, the value $\eta_\beta(\mu,K)$ must bound the contraction coefficient at every smaller order.

        Beyond monotonicity, R\'enyi Divergences satisfy a further structural property with respect to the order. For fixed probability measures $\nu\ll\mu$, the function
        \begin{equation*}
            \alpha\longmapsto(\alpha-1)D_\alpha(\nu\|\mu)=\log\mean{\mu}{\left(\frac{\dd\nu}{\dd\mu}\right)^\alpha}
        \end{equation*}
        is convex~\cite[Corollary 2]{van2014renyi}. As with monotonicity, this property does not pass immediately to the corresponding contraction coefficients. Nevertheless, the three-order comparison in~\cref{eq:main_text_three_order_inequality} underlying~\cref{thm:monotonicity} can be refined to show that the scaled SDPI constant inherits the same convexity.
        \begin{theorem}\label{thm:convexity}
            Consider a pair $(\mu, K)\in\calP(\X)\times\calP(\Y|\X)$. Then, the function $\alpha\mapsto(\alpha-1)\eta_\alpha(\mu, K)$ is convex on the interval $(1, \infty)$.
            \appendixproofref{appendix:proof_convexity}
        \end{theorem}
        \begin{remark}
            The factor $\alpha-1$ in~\cref{thm:convexity} is essential: the map $\alpha\mapsto\eta_\alpha(\mu,K)$ need not itself be convex on $(1,\infty)$, as illustrated numerically in~\cref{fig:dd_bounds_on_renyi_sdpi}.
        \end{remark}
        The key difference from the monotonicity argument is how the contraction coefficients at the two comparison orders are used in the three-order inequality. For monotonicity, we use only a common upper bound on $\eta_\alpha(\mu,K)$ and $\eta_\beta(\mu,K)$, and therefore discard their individual values. Retaining these values separately in the same comparison yields the sharper estimate
        \begin{equation}\label{eq:main_text_refined_three_order_inequality}
            \eta_\gamma(\mu,K)\leq\frac{(\alpha-1)\eta_\alpha(\mu,K)+\eta_\beta(\mu,K)}{\alpha},
        \end{equation}
        where $1<\alpha<\beta<\infty$ and $\gamma\triangleq1+\frac{\alpha}{\beta^\prime}$, as established in~\cref{lemma:order_insertion_convexity}. To make the connection with convexity explicit, define $F(\theta)\triangleq(\theta-1)\eta_\theta(\mu,K)$. Then~\cref{eq:main_text_refined_three_order_inequality} reads
        \begin{equation*}
            F(\gamma)\leq \frac1{\beta^\prime} F(\alpha) + \frac1\beta F(\beta).
        \end{equation*}
        Since $\gamma=\frac1{\beta^\prime}\alpha+\frac1\beta\beta$, this means that there is a suitable intermediate order at which $F$ does not exceed the affine interpolation of its endpoint values. A contradiction argument based on this property shows that the same bound must in fact hold at every intermediate order, which is precisely convexity.

        With~\cref{thm:monotonicity,thm:convexity} in hand, we now collect their main consequences. We begin with the distribution-dependent setting and establish continuity in the order.
        \begin{corollary}\label{corollary:continuity_in_order}
            Consider a pair $(\mu,K)\in\calP(\X)\times\calP(\Y|\X)$. Then, the function $\alpha\mapsto\eta_\alpha(\mu,K)$ is continuous on $[1,\infty]$.
            \appendixproofref{appendix:proof_continuity_in_order}
        \end{corollary}
        Having described the dependence on the order for a fixed reference measure, we now turn to the distribution-independent setting. The monotonicity and scaled convexity established above are preserved under the additional optimisation over $\mu$, leading to the following corollary.
        \begin{corollary}\label{corollary:distribution_independent_order_properties}
            Consider a Markov kernel $K\in\calP(\Y|\X)$. Then,
            \begin{enumerate}[a)]
                \item The function $\alpha\mapsto\eta_\alpha(K)$ is non-decreasing on $[1,\infty]$,
                \item The function $\alpha\mapsto(\alpha-1)\eta_\alpha(K)$ is convex on $(1,\infty)$,
                \item The function $\alpha\mapsto\eta_\alpha(K)$ is continuous on the interval $(1,\infty]$.
            \end{enumerate}
            \appendixproofref{appendix:proof_distribution_independent_order_properties}
        \end{corollary}
        \begin{remark}
            The exclusion of the order one from Part c) of~\cref{corollary:distribution_independent_order_properties} is essential. An example of discontinuity is given by the $Z$-channel $K = \big[\begin{smallmatrix}
                1 & 0\\
                1-\lambda & \lambda
            \end{smallmatrix}\big]$ with parameter $\lambda\in(0,1)$. One has $\eta_1(K)=\lambda<1$, whereas $\eta_\alpha(K)=1$ for every $\alpha\in(1,\infty]$, as shown in~\cref{example:z_channel}.
        \end{remark}

    \subsection{Properties of Extremal Distributions}\label{sec:properties}\noindent
        Let us first focus on the distribution-dependent case, where a reference measure $\mu$ is fixed. Our objective is to study properties of the probability measures that achieve the optimisation in~\cref{eq:dd_sdpi_constant}. We refer to any probability measure attaining the supremum defining $\eta_\alpha(\mu,K)$ as an extremal distribution. The interest in analysing such properties stems from the literature on log-Sobolev-type inequalities, where ``extremal functions'' are studied and shown to satisfy particular variational equations~\cite[Section 6]{bobkov2006modified}. This approach has similarly been applied to characterise the extremal distributions of $\eta_\varphi(\mu, K)$ for a restricted family of $\varphi$-Divergences~\cite[Theorem III.12]{raginsky2016strong}. In~\cite[Lemma 4.7]{caputo2025entropy}, conditions on $(\mu, K)$ are given that guarantee the extremal distributions for $\eta_{\sf KL}(\mu, K)$ to have full support.

        For the contraction of R\'enyi Divergences, the following can be established about extremal distributions:
        \begin{theorem}\label{thm:extremal_functions}
            Let $\alpha\in(1,\infty)$ and consider a pair $(\mu,K)\in\calP(\X)\times\calP(\Y|\X)$. Then, exactly one of the following holds:
            \begin{enumerate}[a)]
                \item $\eta_\alpha(\mu,K)=\eta_{\chi^2}(\mu,K)$,
                \item $\eta_{\chi^2}(\mu,K)<\eta_\alpha(\mu,K)<1$, and the supremum defining $\eta_\alpha(\mu,K)$ is attained.
            \end{enumerate}
            Moreover, in either case, every extremal distribution $\nu^\star$ for $\eta_\alpha(\mu,K)$, with corresponding density $f^\star\triangleq\frac{\dd\nu^\star}{\dd\mu}$, satisfies
            \begin{equation}\label{eq:renyi_sdpi_foc}
                \frac{K\big((K_\mu^\star f^\star)^{\alpha-1}\big)}{\|K_\mu^\star f^\star\|_{L^\alpha(\mu K)}^\alpha}\leq\eta_\alpha(\mu,K)\frac{(f^\star)^{\alpha-1}}{\|f^\star\|_{L^\alpha(\mu)}^\alpha}+\bigl(1-\eta_\alpha(\mu,K)\bigr)\allones\qquad\mu\text{-almost everywhere},
            \end{equation}
            with equality on $\supp(\nu^\star)$.
            \appendixproofref{appendix:proof_extremal_functions}
        \end{theorem}
        The variational condition in~\cref{eq:renyi_sdpi_foc} also constrains the support of extremal distributions. For $\alpha=2$, it was shown in~\cite[Theorem 4]{jin2024properties} that whenever $\eta_{\chi^2}(\mu, K)<1$, the optimisation defining $\eta_2(\mu,K)$ may be restricted to probability measures whose support is strictly contained in $\supp(\mu)$. We extend this property to orders $\alpha\in[2,\infty)$ and strengthen it to a statement about every extremal distribution.
        \begin{theorem}\label{thm:boundary_nu}
            Let $\alpha\in[2,\infty)$ and consider a pair $(\mu,K)\in\calP(\X)\times\calP(\Y|\X)$ with $|\supp(\mu)|\geq2$. Then
            \begin{equation}\label{eq:boundary_nu_maximum_representation}
                \eta_\alpha(\mu,K)=\max_{\substack{\nu\in\calP(\X):\\\supp(\nu)\subsetneq\supp(\mu)}}\frac{D_\alpha(\nu K\|\mu K)}{D_\alpha(\nu\|\mu)},
            \end{equation}
            where the maximum is understood to be zero when its constraint set is empty.
            Moreover, if $0<\eta_\alpha(\mu,K)<1$, then every extremal distribution $\nu^\star$ for $\eta_\alpha(\mu,K)$ satisfies $\supp(\nu^\star)\subsetneq\supp(\mu)$.
            \appendixproofref{appendix:proof_boundary_nu}
        \end{theorem}
        Together with~\cref{thm:extremal_functions}, this shows that if $\alpha\in[2,\infty)$ and $\eta_\alpha(\mu,K)>\eta_{\chi^2}(\mu,K)$, then every extremal distribution has support strictly contained in $\supp(\mu)$.
        \begin{remark}
            \Cref{thm:boundary_nu} also clarifies why the variational relation in~\cref{eq:renyi_sdpi_foc} is stated as an inequality rather than an equality. Indeed, while~\cref{thm:extremal_functions} guarantees equality on $\supp(\nu^\star)$, the proof of~\cref{thm:boundary_nu} shows that the inequality is strict at some point of $\supp(\mu)\setminus\supp(\nu^\star)$. This contrasts with other extremal analyses such as the modified log-Sobolev setting of~\cite[Proof of Theorem 6.5]{bobkov2006modified}, where extremal functions are first shown to be strictly positive and the corresponding variational relation consequently holds with equality.
        \end{remark}
        In fact, when $\alpha=\infty$, a more precise result can be uncovered: the achieving distributions are not merely of smaller support, but are precisely restrictions of $\mu$ to some non-trivial measurable set $A\in\Sigma_\X$.
        \begin{theorem}\label{prop:infty_sdpi_achieved_on_conditionals}
            Consider a pair $(\mu,K)\in\calP(\X)\times\calP(\Y|\X)$. Then
            \begin{equation*}
                \eta_\infty(\mu,K)=\max_{\substack{A\in\Sigma_\X:\\0<\mu(A)<1}}\frac{D_\infty\left(\mu_{|A}K\|\mu K\right)}{D_\infty\left(\mu_{|A}\|\mu\right)},
            \end{equation*}
            where $\mu_{|A}(x) \triangleq \frac{\mathbbm{1}_A(x)\mu(x)}{\mu(A)}$ denotes the restriction of $\mu$ to $A$, and the maximum is understood to be zero when its constraint set is empty. Moreover, if $0<\eta_\infty(\mu,K)<1$, then every extremal distribution for $\eta_\infty(\mu,K)$ is of the form $\mu_{|A}$ for some $A\in\Sigma_\X$ with $0<\mu(A)<1$.
            \appendixproofref{appendix:proof_infty_sdpi_achieved_on_conditionals}
        \end{theorem}
        We now provide three examples that leverage the representation given in~\cref{prop:infty_sdpi_achieved_on_conditionals}.
        \begin{example}[Binary Symmetric Channel]\label{example:dd_infty_sdpi_bsc}
            Consider $K=\mathrm{BSC}(\varepsilon)$ and $\mu=\left[\frac12, \frac12\right]$. Since $\X=\{0,1\}$, we only need to consider $A=\{0\}$ and $A=\{1\}$ when applying~\cref{prop:infty_sdpi_achieved_on_conditionals}. One has $\mu_{|\{0\}}=\delta_0$ and $\mu_{|\{1\}}=\delta_1$, so that
            \begin{equation*}
                \eta_\infty(\mu, K) = \max\left(\frac{D_\infty(\delta_0 K\|\mu K)}{D_\infty(\delta_0\|\mu)}, \frac{D_\infty(\delta_1 K\|\mu K)}{D_\infty(\delta_1\|\mu)}\right) =\frac{\log\big(2\max(\varepsilon, 1-\varepsilon)\big)}{\log (2)} = 1 - \log_2\frac1{\max(\varepsilon, 1-\varepsilon)}.
            \end{equation*}
        \end{example}
        \begin{example}[Binary Erasure Channel]
            Let $\X=\{0,1\}, \Y=\{0, 1, e\}$, and consider $K=\mathrm{BEC}(\delta)$ with $\delta\in[0,1)$ and $\mu=\left[\frac12, \frac12\right]$. Since $\X=\{0,1\}$, we only need to consider subsets $A$ of size 1. In fact, selecting $A=\{0\}$ gives
            \begin{equation*}
                \frac{D_\infty\left(\mu_{|A} K\|\mu K\right)}{D_\infty\left(\mu_{|A}\|\mu\right)} = \frac{\log\max_{y\in\supp(\mu K)}\frac{\mu_{|\{0\}}K(y)}{\mu K(y)}}{\log(2)} = 1.
            \end{equation*}
            \Cref{prop:infty_sdpi_achieved_on_conditionals} gives $\eta_\infty(\mu, K)\geq 1$, so we must have $\eta_\infty(\mu, K)=1$.
        \end{example}
        The next example considers larger alphabet sizes.
        \begin{example}[$n$-ary Symmetric Channel]
            Suppose $\X=\Y=\{0, \dots, n-1\}$, $\mu=\left[\frac1n, \dots, \frac1n\right]$, and consider the $n$-ary Symmetric Channel with $\varepsilon\in\left[0,\frac{n-1}{n}\right]$, defined through $K(y|x) = (1-\varepsilon)\mathbbm{1}_{y=x}+\frac{\varepsilon}{n-1}\mathbbm{1}_{y\neq x}$ for $(x, y)\in\X\times\Y$. We have $\mu K=\mu$ and for any $A$ with $|A|=k\in[n-1]$,
            \begin{align*}
                \mu_{|A}K(y) &= \sum_{x\in A}K(y|x)\frac1k \\
                &= \begin{cases}
                    \frac1k\left[(1-\varepsilon) + (k-1)\frac{\varepsilon}{n-1}\right], \quad\text{if } y\in A\\
                    \frac{\varepsilon}{n-1}, \quad\text{otherwise}.
                \end{cases}
            \end{align*}
            Since $1-\varepsilon\geq\frac{\varepsilon}{n-1}$, we get $\max_{y\in\Y}\mu_{|A}K(y) = \frac1k\left[(1-\varepsilon) + (k-1)\frac{\varepsilon}{n-1}\right]= \frac1k\left(1-\varepsilon\frac{n-k}{n-1}\right)$, so that
            \begin{equation*}
                \frac{D_\infty\left(\mu_{|A} K\|\mu K\right)}{D_\infty\left(\mu_{|A}\|\mu\right)} = \frac{\log\max_{y\in\Y}\frac{\mu_{|A}K(y)}{\mu K(y)}}{-\log\mu(A)} = 1 + \frac{\log\left(1-\varepsilon\frac{n-k}{n-1}\right)}{\log\left(\frac{n}{k}\right)}.
            \end{equation*}
            Since the ratio of $\infty$-R\'enyi Divergences for a set $A$ depends solely on its size, \cref{prop:infty_sdpi_achieved_on_conditionals} gives
            \begin{equation}\label{eq:q_ary_symmetric}
                \eta_\infty(\mu, K) = \max_{k\in[n-1]}\left\{1 + \frac{\log\left(1-\varepsilon\frac{n-k}{n-1}\right)}{\log\left(\frac{n}{k}\right)}\right\} = 1 + \max_{k\in[n-1]}\frac{\log\left(1-\varepsilon\frac{n-k}{n-1}\right)}{\log\left(\frac{n}{k}\right)}.
            \end{equation}
            Under the substitution $t(k)=\log\left(\frac{n}{k}\right)$, the fraction in~\cref{eq:q_ary_symmetric} becomes $\frac{\log\left(1-\frac{n\varepsilon}{n-1}+\frac{n\varepsilon}{n-1}e^{-t(k)}\right)}{t(k)}$, which is non-decreasing in $t(k)$ by~\cref{corollary:log_sum_exp}. Since $t(k)$ is decreasing in $k$, the maximum is attained at $k=1$, and hence
            \begin{equation*}
                \eta_\infty(\mu,K)=1+\frac{\log(1-\varepsilon)}{\log n}=1-\log_n\frac1{1-\varepsilon}.
            \end{equation*}
        \end{example}
        For the distribution-independent SDPI constant, one can similarly study the properties of extremisers. In the context of $\varphi$-Divergences, it was recently proven that the optimisation $\sup_{\nu, \mu} \frac{D_\varphi(\nu K\|\mu K)}{D_\varphi(\nu\|\mu)}$ defining $\eta_\varphi(K)$ can be restricted to pairs of probability measures that are binary supported~\cite{ordentlich2021strong}. We show that the same result holds for R\'enyi Divergences of order $\alpha\in[0,\infty]$, and that furthermore when $\alpha\in[2, \infty]$, one of the two probability measures only needs to be supported on a single point.
        \begin{theorem}\label{prop:eta_parameter_renyi_binary}
            Consider a Markov kernel $K\in\calP(\Y|\X)$. Then:
            \begin{enumerate}[a)]
                \item For $\alpha\in[0,\infty]$, the optimisation defining $\eta_\alpha(K)$ can be restricted to pairs $(\nu,\mu)\in\calP(\X)\times\calP(\X)$ satisfying $\big|\supp(\nu)\cup\supp(\mu)\big|\leq2$,
                \item For $\alpha\in[2,\infty]$, the optimisation can be further restricted to pairs of the form
                \begin{equation*}
                    \nu=\delta_x,\qquad\mu=(1-t)\delta_x+t\delta_{x^\prime},
                \end{equation*}
                where $x,x^\prime\in\X$ are distinct and $t\in(0,1)$.
            \end{enumerate}
            \appendixproofref{appendix:proof_eta_parameter_renyi_binary}
        \end{theorem}
    \subsection{Simplified Representations and Conditions for Trivial Contraction}\label{sec:exact_characterisations}
        A hurdle when dealing with SDPI constants is that one must optimise over an entire probability simplex, or a Cartesian product of two simplices. As highlighted in~\cref{sec:properties}, the domain over which the optimisation is carried out can be reduced once the properties of extremal distributions are uncovered. This is of practical importance, as it renders numerical evaluations more efficient. In the following, we aim to obtain even further simplifications.

        A well-known example of a closed-form expression that characterises an SDPI constant is the one associated with Total Variation Distance in the distribution-independent case:
        \begin{equation}\label{eq:two_point_tv_di}
                \eta_{\sf TV}(K) = \sup_{x, x^\prime\in\X}\|K(\cdot|x)-K(\cdot|x^\prime)\|_{\rm TV}
        \end{equation}
        This ``two-point characterisation'' was proven in~\cite{cohen1993relative} for the discrete case, and extended to the general case in~\cite{del2003contraction} (see also~\cite{raginsky2016strong}). We also note that a similar representation holds more generally for Hockey-Stick Divergences~\cite{asoodeh2020contraction}, which are $\varphi$-Divergences that generalise Total Variation Distance.

        Two-point characterisations appear to be uncommon in the realm of SDPI constants. Indeed, for $\varphi$-Divergences, only Hockey-Stick Divergences (and hence the Total Variation Distance) are known to have such a representation. What they indicate, in fact, is that instead of searching over all pairs of probability measures to determine the SDPI constant, it suffices to look at pairs of Dirac measures. Our next result shows that this happens when considering the R\'enyi Divergence of order $\alpha=\infty$.
        \begin{theorem}\label{prop:infty_renyi_sdpi_closed_form}
            For any Markov kernel $K\in\calP(\Y|\X)$, the distribution-independent SDPI constant associated with the R\'enyi Divergence of order $\infty$ can be written as
            \begin{equation}\label{eq:infty_renyi_sdpi_closed_form}
                \eta_\infty(K) = \max_{x, x^\prime\in\X}\left\{1 -\min_{y\in\supp(K(\cdot|x))} \frac{K(y|x^\prime)}{K(y|x)}\right\}.
            \end{equation}
            \appendixproofref{appendix:proof_infty_renyi_sdpi_closed_form}
        \end{theorem}
        \begin{remark}
            \Cref{eq:two_point_tv_di} is obtained by directly considering pairs of distributions of the form $\nu=\delta_x$ and $\mu=\delta_{x^\prime}$ when evaluating $\sup_{\nu, \mu}\frac{\|\nu K-\mu K\|_{\sf TV}}{\|\nu-\mu\|_{\sf TV}}$. The same cannot be done for the $\infty$-R\'enyi Divergence, and instead~\cref{eq:infty_renyi_sdpi_closed_form} requires using $\nu=\delta_x$ and $\mu=(1-t)\delta_x+t\delta_{x^\prime}$ when evaluating $\sup_{\nu, \mu}\frac{D_\infty(\nu K\|\mu K)}{D_\infty(\nu\|\mu)}$. As shown in~\cref{appendix:proof_infty_renyi_sdpi_closed_form}, the closed-form expression is eventually obtained after taking the limit
            \begin{equation}\label{eq:two_point_limit_ratio_renyi_divergences}
                \max_{x, x^\prime\in\X}\lim_{t\to0^+} \frac{D_\infty(\delta_x K\|\mu_t K)}{D_\infty(\delta_x\|\mu_t)},
            \end{equation}
            which unlike Total Variation Distance, corresponds to picking the same Dirac measure $\delta_x$ for both $\nu$ and $\mu$.
        \end{remark}
        \begin{figure}
            \centering
            \includegraphics[width=0.75\linewidth]{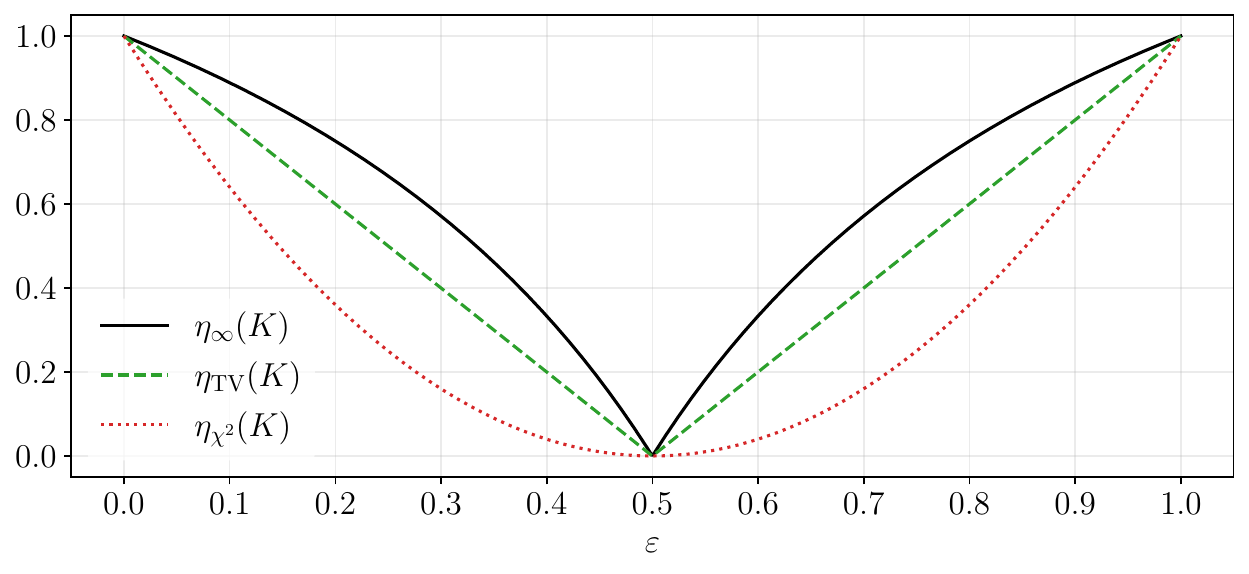}
            \caption{Comparison between $\eta_\infty(K)$ and other common distribution-independent SDPI constants for $K=\mathrm{BSC}(\varepsilon)$.}
            \label{fig:bsc_sdpi_constants_comparison}
        \end{figure}
        Next, we obtain closed-form expressions for various channels by applying~\cref{prop:infty_renyi_sdpi_closed_form} in the following examples.
        \begin{example}[Binary Input-Output Channels]
            Consider an arbitrary binary channel, that is $\X=\Y=\{0,1\}$ and the Markov kernel consists of the conditionals $K(\cdot|0)=[1-\lambda, \lambda]$ and $K(\cdot|1)=[\kappa, 1-\kappa]$, where $0\leq\lambda, \kappa\leq\frac12$. Considering $x=0, x^\prime=1$, we find $\min_y\frac{K(y|x^\prime)}{K(y|x)}= \frac{\kappa}{1-\lambda}$ and for $x=1, x^\prime=0$, we find $\min_y\frac{K(y|x^\prime)}{K(y|x)}= \frac{\lambda}{1-\kappa}$. Therefore, an application of~\cref{prop:infty_renyi_sdpi_closed_form} yields
            \begin{equation}\label{eq:eta_inf_binary_channel}
                \eta_\infty(K) = 1 - \min\left(\frac{\lambda}{1-\kappa}, \frac{\kappa}{1-\lambda}\right).
            \end{equation}
            In particular, for $K=\mathrm{BSC}(\varepsilon)$ with $0\leq\varepsilon\leq\frac12$, \cref{eq:eta_inf_binary_channel} gives $\eta_\infty(K) = 1 - \frac{\varepsilon}{1-\varepsilon}$ since $\lambda=\kappa=\varepsilon$ (see~\cref{fig:bsc_sdpi_constants_comparison} for a comparison with other common SDPI constants).
        \end{example}

        \begin{example}[Binary Erasure Channel]
             Let $\X=\{0,1\}, \Y=\{0, 1, e\}$, and consider $K=\mathrm{BEC}(\delta)$ with $\delta\in[0,1)$. For $x=0, x^\prime=1$, we find $\min_y\frac{K(y|x^\prime)}{K(y|x)}=\frac0{1-\delta}=0$, from which we directly conclude $\eta_\infty(K)=1$ by~\cref{prop:infty_renyi_sdpi_closed_form}.
        \end{example}

        \begin{example}[$n$-ary Symmetric Channel]
            Suppose $\X=\Y=\{0, \dots, n-1\}$, and consider the $n$-ary Symmetric Channel with $\varepsilon\in\left[0,\frac{n-1}{n}\right]$, defined through $K(y|x) = (1-\varepsilon)\mathbbm{1}_{y=x}+\frac{\varepsilon}{n-1}\mathbbm{1}_{y\neq x}$ for $(x, y)\in\X\times\Y$. By the symmetry of $K$, notice that any pair of distinct inputs $x\neq x^\prime$ will yield the same result when evaluating $\min_y \frac{K(y|x^\prime)}{K(y|x)}$. For $y=x$, the ratio is $\frac{K(x|x^\prime)}{K(x|x)}=\frac{\frac{\varepsilon}{n-1}}{1-\varepsilon}$, for $y=x^\prime$, the ratio is $\frac{K(x^\prime|x^\prime)}{K(x^\prime|x)}=\frac{1-\varepsilon}{\frac{\varepsilon}{n-1}}$, and for $y\notin\{x, x^\prime\}$, the ratio is $\frac{K(y|x^\prime)}{K(y|x)}=\frac{\frac{\varepsilon}{n-1}}{\frac{\varepsilon}{n-1}}=1$. Since $1-\varepsilon \geq\frac{\varepsilon}{n-1}$, we get $\min_y \frac{K(y|x^\prime)}{K(y|x)}=\frac{\frac{\varepsilon}{n-1}}{1-\varepsilon}$ and hence $\eta_\infty(K)=1 - \frac{\varepsilon}{(1-\varepsilon)(n-1)}$.
        \end{example}
        Beyond restricting the space over which the optimisation is carried out, one can simplify the expression of an SDPI constant when it mirrors a contraction coefficient that is well-understood. A prime example is the $\chi^2$-Divergence, for which the associated SDPI constant $\eta_{\chi^2}(\mu, K)$ has multiple deep interpretations. Not only is it equivalent to the squared Hirschfeld-Gebelein-R\'enyi maximal correlation~\cite{hirschfeld1935connection,gebelein1941statistische,renyi1959measures}, but it is also governed by the spectral properties of the Markov kernel~\cite{renyi1959measures}. Moreover, the role of the $\chi^2$-Divergence as a ``local approximation'' to many other divergences makes it a fundamental quantity in the study of divergence contraction~\cite[Theorem 2]{makur2020comparison}.

        The connection between the contraction coefficients of $\chi^2$-Divergence and other divergences is not new. It was first proven that the KL Divergence coincides with it through the identity $\eta_{\sf KL}(K)=\eta_{\chi^2}(K)$~\cite{ahlswede1976spreading}. This was further extended to all $\varphi$-Divergences generated by non-linear operator convex functions $\varphi$, for which it was shown $\eta_{\varphi}(K)=\eta_{\chi^2}(K)$~\cite[Theorem 1]{choi1994equivalence}\cite[Corollary III.1]{raginsky2016strong}\cite[Proposition 6]{makur2020comparison}. In the following theorem, we show that for R\'enyi Divergences of order $\alpha\in[0,1]$, the distribution-independent SDPI constant similarly collapses to $\eta_{\chi^2}(K)$.
        \begin{theorem}\label{prop:distribution_independent_sdpi_equal_chi_square_sdpi}
            For any Markov kernel $K\in\calP(\Y|\X)$ and $\alpha\in[0,1]$, we have $\eta_\alpha(K) = \eta_{\chi^2}(K)$.
            \appendixproofref{appendix:proof_distribution_independent_sdpi_equal_chi_square_sdpi}
        \end{theorem}
        \begin{remark}
            The restriction to orders $\alpha\in[0,1]$ is strictly necessary. For $\alpha\in(1, \infty]$, the SDPI constant $\eta_\alpha(K)$ can trivially evaluate to 1 even when the $\chi^2$-contraction is strictly bounded below 1. An example is the $Z$-channel with crossover probability $\lambda\in(0,1)$, for which $\eta_{\chi^2}(K)=\lambda<1$, yet $\eta_\alpha(K)=1$ for all $\alpha\in(1, \infty]$ (see~\cref{example:z_channel} for a derivation of the identity $\eta_\alpha(K)=1$).
        \end{remark}
        The stark divergence in the SDPI constant's behaviour across different ranges of $\alpha$ is not an isolated case, as we shall shortly see.
        \begin{remark}
            \Cref{prop:distribution_independent_sdpi_equal_chi_square_sdpi} also gives $\eta_C(K)=\eta_{\chi^2}(K)$ for the Chernoff information~\cite[Section IV-A]{van2014renyi}
            \begin{equation*}
                C(\nu\|\mu)\triangleq\sup_{0<\alpha<1}(1-\alpha)D_\alpha(\nu\|\mu).
            \end{equation*}
            Indeed, for every $0<\alpha<1$ and every pair with $0<C(\nu\|\mu)<\infty$,
            \begin{equation*}
                (1-\alpha)D_\alpha(\nu K\|\mu K) \leq\eta_{\chi^2}(K)(1-\alpha)D_\alpha(\nu\|\mu) \leq\eta_{\chi^2}(K)C(\nu\|\mu).
            \end{equation*}
            Taking the supremum over $\alpha$ proves the upper bound. For the reverse inequality, let $\nu_\varepsilon=(1+\varepsilon h)\mu$, where $\mean{\mu}{h}=0$ and $\mean{\mu}{h^2}>0$. A second-order expansion, uniform in $\alpha\in[0,1]$, yields
            \begin{equation*}
                \lim_{\varepsilon\downarrow0}\frac{C(\nu_\varepsilon K\|\mu K)}{C(\nu_\varepsilon\|\mu)}
                =\frac{\|K_\mu^\star h\|_{L^2(\mu K)}^2}{\|h\|_{L^2(\mu)}^2}.
            \end{equation*}
            The supremum over $h$ on the right is $\eta_{\chi^2}(\mu,K)$. Taking the supremum over $\mu$ establishes the lower bound.
        \end{remark}

        A final avenue for simplifying the expression of an SDPI constant lies in identifying conditions on $K$ that prevent contraction, yielding an SDPI constant of exactly $1$. For general $\varphi$-Divergences, the condition $\eta_\varphi(K)=1$ is known to be equivalent to the existence of conditionals $K(\cdot|x), K(\cdot|x^\prime)$ with disjoint support~\cite{cohen1993relative}. In the specific context of R\'enyi Divergences, the same condition has been shown to be sufficient to imply $\eta_\alpha(K)=1$ for $\alpha>0$~\cite[Proposition 3]{grosse2025bounds}. Our next result strengthens this relationship from a one-way implication to a complete equivalence, establishing the precise structural properties of $K$ that are both necessary and sufficient for $\eta_\alpha(K)=1$, depending on the range of $\alpha$.
        \begin{theorem}\label{prop:non_trivial_sdpi_characterisation}
            Consider a Markov kernel $K\in\calP(\Y|\X)$. Then we have:
            \begin{enumerate}[a)]
                \item For $\alpha\in[0,1]$, $\eta_\alpha(K)=1$ if and only if there exist  $x, x^\prime\in\X$ such that ${\supp\big(K(\cdot|x)\big)\cap\supp\big(K(\cdot|x^\prime)\big)=\emptyset}$,
                \item For $\alpha\in(1, \infty]$, $\eta_\alpha(K)=1$ if and only if there exist  ${x, x^\prime\in\X}$ such that ${\supp\big(K(\cdot|x)\big)\neq\supp\big(K(\cdot|x^\prime)\big)}$.
            \end{enumerate}
            \appendixproofref{appendix:proof_non_trivial_sdpi_characterisation}
        \end{theorem}
        We now provide an example of a Markov kernel to illustrate~\cref{prop:non_trivial_sdpi_characterisation}.
        \begin{example}\label{example:z_channel}
            Consider the Z-channel $K = \big[\begin{smallmatrix}
                1 & 0\\
                1-\lambda & \lambda
            \end{smallmatrix}\big]$ with $\lambda\in(0,1)$. Then \cref{prop:non_trivial_sdpi_characterisation} gives $\eta_\alpha(K)=1$ for $\alpha>1$ because the conditionals of $K$ have non-identical support, and $\eta_\alpha(K)<1$ for $\alpha\in[0,1]$.
        \end{example}

    \subsection{Tensorisation}
        In this section, we prove tensorisation bounds for the contraction of R\'enyi Divergences. Tensorisation allows one to reduce the problem of analysing contraction in higher-dimensional spaces to simply studying contraction along each dimension individually. It has been greatly studied in the context of functional inequalities such as Sobolev-type inequalities~\cite{Bakry2014,chewi2026log}, and it has also appeared in the context of SDPIs for $\varphi$-Divergences~\cite[Theorem 3.9]{raginsky2016strong},~\cite[Equation 61]{makur2020comparison},~\cite[Corollary 6]{polyanskiy2017strong}.

        Given a divergence $D(\cdot\|\cdot)$, a tensorisation result should be understood as follows. Consider the spaces $\{\X_i\}_{i=1}^n$ and $\{\Y_i\}_{i=1}^n$, and let $(\mu_i,K_i)\in\calP(\X_i)\times\calP(\Y_i|\X_i)$ for each $i\in[n]$. Then the distribution-dependent (resp. distribution-independent) SDPI constant for $D(\cdot\|\cdot)$ satisfies a tensorisation result if $\eta_D(\mu_1\otimes\cdots\otimes \mu_n, K_1\otimes\dots\otimes K_n)$ (resp. $\eta_D(K_1\otimes\cdots \otimes K_n)$) can be ``controlled'' by some function of $\eta_D(\mu_1, K_1), \dots, \eta_D(\mu_n, K_n)$ (resp. $\eta_D(K_1), \dots, \eta_D(K_n)$). For example, a tensorisation result can occur through exact identities, as is the case with the distribution-dependent SDPI constant of some $\varphi$-Divergences. Specifically, for any convex $\varphi$ that induces a homogeneous and subadditive $\varphi$-Entropy, it is shown in~\cite[Theorem III.9]{raginsky2016strong} that
        \begin{equation}\label{eq:varphi_tensorisation_equality}
            \eta_\varphi(\mu_1\otimes\cdots\otimes\mu_n, K_1\otimes\cdots\otimes K_n) = \max_{i\in[n]}\eta_\varphi(\mu_i, K_i).
        \end{equation}
        Tensorisation can also manifest in the shape of bounds, notably seen with $\varphi$-Divergences generated by non-linear operator convex $\varphi$, for which
        \begin{equation}\label{eq:varphi_tensorisation_inequality}
            \eta_\varphi(K_1\otimes\cdots\otimes K_n)\leq 1 - \prod_{i=1}^n (1-\eta_\varphi(K_i)).
        \end{equation}
        This was initially proven for the KL Divergence~\cite[Corollary 6]{polyanskiy2017strong}, and extended in~\cite[Equation 61]{makur2020comparison}.

        Another tensorisation property that pertains to $\varphi$-Divergences comes from a result initially due to Samorodnitsky involving Shannon Entropy~\cite{samorodnitsky2016entropy}. It was later reformulated in terms of KL Divergence in~\cite[Theorem 20]{polyanskiy2017strong} and further extended to $\varphi$-Divergences with non-linear operator convex $\varphi$ in~\cite[Theorem 7]{makur2020comparison}. Its final form can be stated as follows. For any $\nu, \mu\in\calP(\X_1\times\cdots\times\X_n)$,
        \begin{equation}\label{eq:varphi_samorodnitsky_tensorisation}
            D_\varphi\big(\nu (K_1\otimes\cdots\otimes K_n)\|\mu (K_1\otimes\cdots\otimes K_n)\big) \leq \sum_{T\subseteq[n]} \left(\prod_{i\in T}\eta_\varphi(K_i)\prod_{j\in[n]\setminus T}\big(1-\eta_\varphi(K_j)\big)\right) D_\varphi(\nu_T\|\mu_T),
        \end{equation}
        where for a probability measure $\gamma\in\calP(\X_1\times\cdots\times\X_n)$ and set of indices $T\subseteq [n]$, $\gamma_T\in\calP(\prod_{i\in T}\X_i)$ denotes its marginal over the corresponding indices, with the convention that $\gamma_\varnothing$ is the (unique) probability measure on a singleton.

        Regarding R\'enyi Divergences, tensorisation properties of the form of~\cref{eq:varphi_tensorisation_equality,eq:varphi_tensorisation_inequality} are unknown. However, some advances have been partially made recently, establishing the same type of tensorisation as in~\cref{eq:varphi_samorodnitsky_tensorisation}. The result below is an extension of~\cite[Theorem 1]{abawonse2026generalized}.
        \begin{theorem}\label{prop:renyi_samorodnitsky}
            Consider $n\geq1$ pairs $(\mu_1, K_1), \dots, (\mu_n, K_n)$, where $(\mu_i, K_i)\in\calP(\X_i)\times\calP(\Y_i|\X_i)$ for each $i\in[n]$.
            \begin{enumerate}[a)]
                \item For all $\nu\in\calP(\X_1\times\cdots\times\X_n)$ with $\nu\ll\mu_1\otimes\cdots\otimes\mu_n$ and any $\alpha\in(1, \infty]$,
                \begin{equation}\label{eq:dd_renyi_samorodnitsky_tensorisation}
                    D_\alpha\big(\nu (K_1\otimes\cdots\otimes K_n)\|\mu_1K_1\otimes\cdots\otimes\mu_nK_n\big) \leq \sum_{T\subseteq[n]} \left(\prod_{i\in T}\eta_\alpha(\mu_i, K_i)\prod_{j\in[n]\setminus T}\big(1-\eta_\alpha(\mu_j, K_j)\big)\right) D_\alpha(\nu_T\|\mu_T),
                \end{equation}
                \item For all $\nu,\gamma\in\calP(\X_1\times\cdots\times\X_n)$ and any $\alpha\in(0,1)\cup(1,\infty]$, assuming additionally that $\nu\ll\gamma$ when $\alpha>1$,
                \begin{equation}\label{eq:di_renyi_samorodnitsky_tensorisation}
                    D_\alpha\big(\nu (K_1\otimes\cdots\otimes K_n)\|\gamma (K_1\otimes\cdots\otimes K_n)\big) \leq \sum_{T\subseteq[n]} \left(\prod_{i\in T}\eta_\alpha(K_i)\prod_{j\in[n]\setminus T}\big(1-\eta_\alpha(K_j)\big)\right) D_\alpha(\nu_T\|\gamma_T).
                \end{equation}
            \end{enumerate}
            \appendixproofref{appendix:proof_renyi_samorodnitsky}
        \end{theorem}
        \begin{remark}
            Abawonse et al.~\cite{abawonse2026generalized} established a Samorodnitsky-type inequality with R\'enyi Divergences for tensor products of resampling noise operators. The product-reference assertion of~\cref{prop:renyi_samorodnitsky} extends this form of tensorisation to tensor product of arbitrary Markov kernels, while the second assertion gives a distribution-independent version in which the reference measure may be correlated across coordinates, and which holds for a larger range of $\alpha$. While the functional tensorisation argument in~\cite{abawonse2026generalized} can be adapted to obtain the product-reference extension, we provide in~\cref{appendix:proof_renyi_samorodnitsky} a unified proof of both assertions directly at the level of probability measures.
        \end{remark}
        From~\cref{prop:renyi_samorodnitsky}, one can deduce tensorisation bounds of the same type as~\cref{eq:varphi_tensorisation_inequality} for both the distribution-dependent and distribution-independent R\'enyi-SDPI constant.
        \begin{corollary}\label{prop:tensorisation_of_sdpi_constant}
            Consider $n\geq1$ pairs $(\mu_1, K_1), \dots, (\mu_n, K_n)$, where $(\mu_i, K_i)\in\calP(\X_i)\times\calP(\Y_i|\X_i)$ for each $i\in[n]$. Denote by $\mu = \mu_1\otimes\cdots\otimes\mu_n\in\calP(\X_1\times\cdots\times\X_n)$ the product distribution and $K = K_1 \otimes\cdots\otimes K_n\in\calP(\Y_1\times\cdots\times\Y_n|\X_1\times\cdots\times\X_n)$ the product Markov kernel.
            \begin{enumerate}[a)]
                \item For any $\alpha\in(1, \infty]$,
                \begin{equation*}
                    \max_{i\in[n]}\eta_{\alpha}(\mu_i, K_i) \leq \eta_{\alpha}(\mu, K) \leq 1 - \prod_{i=1}^n\big(1-\eta_\alpha(\mu_i, K_i)\big),
                \end{equation*}
            \item For any $\alpha\in(0,1)\cup(1, \infty]$,
                \begin{equation*}
                    \max_{i\in[n]}\eta_{\alpha}(K_i) \leq \eta_{\alpha}(K) \leq 1 - \prod_{i=1}^n\big(1-\eta_\alpha(K_i)\big).
                \end{equation*}
            \end{enumerate}
            \appendixproofref{appendix:proof_tensorisation_of_sdpi_constant}
        \end{corollary}
        \begin{remark}
            In both~\cref{prop:renyi_samorodnitsky,prop:tensorisation_of_sdpi_constant}, the distribution-dependent results do not include the range $\alpha\in(0,1)$. This is not a limitation of the proof, as illustrated by~\cref{example:no_tensorisation_alpha_frac12}, which provides a counterexample to both results.
        \end{remark}
        \begin{example}\label{example:no_tensorisation_alpha_frac12}
            Let $\alpha=\frac12$, $\X=\Y=\{0,1\}$, $\mu=\delta_0$, and $K=\big[\begin{smallmatrix}
                \frac15 & \frac45\\
                1 & 0
            \end{smallmatrix}\big]$. Writing $\nu_t\triangleq t\delta_0+(1-t)\delta_1$, where $0<t<1$, one has
            \begin{equation*}
                \eta_{\frac12}(\mu,K)=\sup_{0<t<1}\frac{D_\frac12(\nu_tK\|\mu K)}{D_\frac12(\nu_t\|\mu)}=\sup_{0<t<1}\frac{-2\log\left(\frac{\sqrt{5-4t}+4\sqrt t}{5}\right)}{-\log t}=\frac13.
            \end{equation*}
            Now consider $K\otimes K$ and $\mu\otimes \mu$ with $\X^2=\Y^2=\{00, 01, 10, 11\}$. Choosing
            \begin{equation*}
                \nu = \frac{16}{49}\delta_{00}+\frac{33}{49}\delta_{11},
            \end{equation*}
            one finds that $D_\frac12(\nu\|\mu\otimes\mu)=-2\log(4/7)$ and $D_\frac12\big(\nu (K\otimes K)\|\mu K\otimes\mu K\big)=-2\log(5/7)$. Hence
            \begin{equation*}
                \eta_{\frac12}(\mu\otimes\mu,K\otimes K)\geq\frac{\log(7/5)}{\log(7/4)}>\frac59=1-\left(1-\eta_{1/2}(\mu,K)\right)^2,
            \end{equation*}
            showing that, in the distribution-dependent setting, tensorisation results such as~\cref{prop:tensorisation_of_sdpi_constant}, and thus~\cref{prop:renyi_samorodnitsky}, cannot hold.
        \end{example}
        In general, R\'enyi Divergences do not admit tensorisation identities of the form of~\cref{eq:varphi_tensorisation_equality}. The next example illustrates this fact for $\alpha=\infty$.
        \begin{example}
            Consider $\X=\Y=\{0,1\}$, $\mu=\left[\frac12,\frac12\right]$, and $K=\mathrm{BSC}(\frac13)$. From~\cref{example:dd_infty_sdpi_bsc},
            \begin{equation*}
                \eta_\infty(\mu,K)=\frac{\log(4/3)}{\log 2}.
            \end{equation*}
            Let $A=\{00,01,10\}$ and $\nu=(\mu\otimes\mu)_{|A}$. Since $\mu K=\mu$, a direct calculation gives
            \begin{equation*}
                \nu K^{\otimes 2}=\frac1{27}\begin{bmatrix}8&7&7&5\end{bmatrix},
            \end{equation*}
            where the entries are ordered according to $00,01,10,11$. Consequently,
            \begin{equation*}
                D_\infty(\nu\|\mu\otimes\mu)=\log\frac43,
                \qquad
                D_\infty\bigl(\nu (K\otimes K)\|\mu K\otimes\mu K\bigr)=\log\frac{32}{27}.
            \end{equation*}
            Hence, the identity $\eta_\infty(\mu_1\otimes\cdots\otimes \mu_n, K_1\otimes\cdots\otimes K_n)=\max_i\eta_\infty(\mu_i, K_i)$ does not hold in general since
            \begin{equation*}
                \eta_\infty(\mu\otimes\mu, K\otimes K)
                \geq\frac{\log(32/27)}{\log(4/3)}
                >\frac{\log(4/3)}{\log 2}
                =\eta_\infty(\mu,K).
            \end{equation*}
        \end{example}
\section{Bounds on the Contraction Coefficient}\label{sec:bounds_on_contraction_coefficient}\noindent
    The results in~\cref{sec:properties_exact_characterisations} provide different avenues for simplifying the computation of $\eta_\alpha(\mu, K)$ and $\eta_\alpha(K)$. Whether by restricting the optimisation space or obtaining exact characterisations for certain orders $\alpha$, they enable the derivation of analytical formulas for the SDPI constants of well-known channels. However, analytically computing the exact contraction coefficient of a general Markov kernel for an arbitrary order $\alpha$ is a difficult endeavour. Unlike the $\chi^2$-contraction, which essentially reduces to an eigenvalue problem, evaluating R\'enyi-SDPI constants demands a different approach.

    To overcome this bottleneck, this section shifts the focus from exactly calculating contraction coefficients to bounding them. By developing systematic upper and lower bounds for $\eta_\alpha(\mu, K)$ and $\eta_\alpha(K)$, we aim to tightly capture a channel's contractive behaviour across all values of $\alpha$, providing clear insights even where closed-form expressions are unattainable. The central idea behind our results is to exploit the relationship between R\'enyi Divergences and other divergences whose contraction properties are well understood. This approach allows us to inherit established contractive properties and map them back into bounds on R\'enyi-SDPI constants.
    \subsection{Linear Bounds between R\'enyi-SDPI Constants and $\chi^2$-SDPI Constants}\label{section:linear_bounds_renyi_chi_squared_sdpi}\noindent
        To anchor our bounds for R\'enyi-SDPI constants, we return to the $\chi^2$-contraction coefficient as our primary reference point. The foundational role of the $\chi^2$-divergence stems from its behaviour as the local quadratic approximation for any $\varphi$-Divergence with $\varphi$ appropriately smooth~\cite[Chapter 7.10]{Polyanskiy2025}. Thanks to this property, it is a well-established result that the $\chi^2$-SDPI constant naturally acts as a universal lower bound on the SDPI constant of many $\varphi$-Divergences. Specifically, whenever $\varphi$ is twice differentiable at unity and $\varphi^{\prime\prime}(1)>0$, one has $\eta_{\chi^2}(\mu, K) \leq \eta_\varphi(\mu, K)$~\cite[Theorem III.3]{raginsky2016strong}\cite[Proposition 3, Item 7]{makur2020comparison} and $\eta_{\chi^2}(K)\leq\eta_\varphi(K)$~\cite[Theorem 5.4]{cohen1993relative}\cite[Theorem 2]{polyanskiy2017strong}\cite[Proposition 5, Item 6]{makur2020comparison}. In the realm of R\'enyi Divergences, $\chi^2$-Divergence also acts as a local quadratic approximation, which directly leads to the following result.
        \begin{theorem}[{\cite[Theorem 1 and Corollary 2]{jin2024properties}}]\label{thm:local_chi_squared_lb_on_renyi_sdpi}
            Let $\alpha\in(0,1)\cup(1, \infty)$. Then we have
            \begin{equation}
                \eta_{\chi^2}(\mu, K) \leq \eta_\alpha(\mu, K)\quad\text{ and }\quad\eta_{\chi^2}(K) \leq \eta_\alpha(K).
            \end{equation}
        \end{theorem}
        While these lower bounds, illustrated in~\cref{fig:di_bounds_on_renyi_sdpi_via_hellinger,fig:dd_bounds_on_renyi_sdpi}, confirm that a channel cannot contract R\'enyi Divergences more aggressively than it contracts the $\chi^2$-Divergence, they do not dictate how far the R\'enyi-SDPI constant might deviate from $\eta_{\chi^2}$. In this section, we derive upper bounds that allow us to ``sandwich'' the R\'enyi-SDPI constant and more deeply understand its relationship to $\eta_{\chi^2}$.

        A natural strategy for obtaining linear bounds between contraction coefficients is to establish bounds between the divergences themselves. By upper and lower bounding a target divergence as a function of a reference divergence, one directly obtains an inequality governing their respective SDPI constants. This approach has been successfully applied in the context of $\varphi$-Divergences, establishing precise bounds that relate $\eta_\varphi(\mu, K)$ to well-known contraction coefficients like $\eta_{\chi^2}(\mu, K)$~\cite{makur2020comparison,george2025divergence,grosse2025bounds}. Following this approach, we first  establish an inequality between R\'enyi Divergence and $\chi^2$-Divergence.
        \begin{lemma}\label{lemma:bound_between_renyi_chi_squared}
        Let $\nu, \mu\in\calP(\X)$ with $\nu\ll\mu$ and define $M\triangleq \left\|\frac{\dd\nu}{\dd\mu}\right\|_{L^\infty(\mu)}$. Then for $\alpha>2$, we have
            \begin{align*}
                D_\alpha(\nu\|\mu) \leq \frac1{\alpha-1}\log\left(1+C_\alpha\left(M\right)\chi^2(\nu\|\mu)\right),
            \end{align*}
            where $C_\alpha(t)\triangleq\frac{t^\alpha-\alpha t+\alpha-1}{(t-1)^2}$ for $t\in\reals_+\setminus\{1\}$ and $C_\alpha(1)\triangleq\frac{\alpha(\alpha-1)}{2}$.
            \appendixproofref{appendix:proof_bound_between_renyi_chi_squared}
        \end{lemma}
        \begin{remark}
            \Cref{lemma:bound_between_renyi_chi_squared} can be seen as an extension of the inequality
            \begin{equation}\label{eq:renyi_ub_by_chi_squared}
                D_\alpha(\nu\|\mu) \leq D_2(\nu\|\mu) = \log\left(1+\chi^2(\nu\|\mu)\right),
            \end{equation}
            which holds for $\alpha\in[0,2]$ due to the monotonicity of R\'enyi Divergences~\cite[Theorem 3]{van2014renyi}.
        \end{remark}
        We are now in a position to state the main result of this section.
        \begin{proposition}\label{prop:ub_on_eta_alpha_by_eta_chi_squared}
            Consider a pair $(\mu, K)\in\calP(\X)\times\calP(\Y|\X)$. Then, for $\alpha>2$ we have
            \begin{equation}
                \eta_\alpha(\mu, K) \leq C_\alpha\left(\frac1{\min_{y\in\supp(\mu K)}\mu K(y)}\right) \cdot\frac{\eta_{\chi^2}(\mu, K)}{(\alpha-1)\min_{x\in\supp(\mu)}\mu(x)},
            \end{equation}
            where $C_\alpha(t)$ is the function defined in~\cref{lemma:bound_between_renyi_chi_squared}. For $\alpha\in[1, 2]$ we have
            \begin{equation}
                \eta_\alpha(\mu, K) \leq \frac{\eta_{\chi^2}(\mu, K)}{\min_{x\in\supp(\mu)}\mu(x)}
            \end{equation}
            while for $\alpha\in(0,1)$, assuming $\mu$ has full support, we have
            \begin{equation}
                \eta_\alpha(\mu, K) \leq \frac{\eta_{\chi^2}(\mu, K)}{\alpha\min_{x\in\X}\mu(x)}.
            \end{equation}
            \appendixproofref{appendix:proof_ub_on_eta_alpha_by_eta_chi_squared}
        \end{proposition}
        \begin{remark}\label{remark:limitation_renyi_chi_squared_bounds}
            Due to the dependence on $\min_{x\in\supp(\mu)}\mu(x)$ and $\min_{y\in\supp(\mu K)}\mu K(y)$, the bounds in~\cref{prop:ub_on_eta_alpha_by_eta_chi_squared} can rapidly become vacuous. Since $\frac1{\min_{x\in\supp(\mu)}\mu(x)}\geq|\supp(\mu)|$ and $\frac1{\min_{y\in\supp(\mu K)}\mu K(y)}\geq|\supp(\mu K)|$, the bounds will typically be informative only when the relevant supports are small. Moreover, the additional dependence on $\alpha$ can render the bounds vacuous either when $\alpha$ gets close to 0, or when it is sufficiently larger than 2. The bounds across the different regimes of $\alpha$ are illustrated in~\cref{fig:dd_bounds_on_renyi_sdpi}.
        \end{remark}
        Despite these limitations, the bounds in~\cref{prop:ub_on_eta_alpha_by_eta_chi_squared} yield an exact result about the convergence of reversible Markov chains, which is discussed in~\cref{section:mixing_time}.
    \subsection{Bounds via One-To-One Mapping with Hellinger Divergences}
        The relationship between the R\'enyi and $\chi^2$-SDPI constants explored in~\cref{section:linear_bounds_renyi_chi_squared_sdpi} provides a useful initial comparison, but the resulting bounds can become vacuous outside of specific regimes. In this section we shift focus to Hellinger Divergences, driven by the one-to-one mapping that relates them to R\'enyi Divergences. By exploiting this connection, we obtain refined bounds that remain non-vacuous even in the regimes where the $\chi^2$-based bounds deteriorate.

        The starting point for this analysis is the identity connecting the R\'enyi and Hellinger Divergences:
        \begin{equation*}
            \renyiDiv{\nu}{\mu}=\frac1{\alpha-1}\log\big(1+(\alpha-1)\calH_\alpha(\nu\|\mu)\big)
        \end{equation*}
        By leveraging~\cref{lemma:sdpi_relationship_concave_function} (see~\cref{appendix:sdpi_relationship_concave_function}), which allows comparing the SDPI constants of divergences related through a convex or concave transformation, the following proposition bounds $\eta_\alpha$ using the SDPI constant of Hellinger Divergences.
        \begin{proposition}\label{corollary:sdpi_ub_lb_hellinger}
            For any pair $(\mu, K)\in\calP(\X)\times\calP(\Y|\X)$:
            \begin{enumerate}[a)]
                \item If $\alpha>1$, then
                \begin{equation}\label{eq:sdpi_hellinger_lb_sdpi renyi}
                    \eta_\alpha(\mu, K) \geq \eta_{\calH_\alpha}(\mu, K)\quad\text{and}\quad\eta_\alpha(K) \geq \eta_{\calH_\alpha}(K)
                \end{equation}
                \item If $\alpha\in(0,1)$, then
                \begin{equation}\label{eq:sdpi_hellinger_ub_sdpi renyi}
                    \eta_\alpha(\mu, K) \leq \eta_{\calH_\alpha}(\mu, K)\quad\text{and}\quad\eta_\alpha(K) \leq \eta_{\calH_\alpha}(K)
                \end{equation}
            \end{enumerate}
            \appendixproofref{appendix:sdpi_relationship_concave_function}
        \end{proposition}
        \begin{remark}
            At order $\alpha=1$, both the R\'enyi and Hellinger Divergences coincide with the KL Divergence, and hence their SDPI constants coincide as well.
        \end{remark}
        Part a) of~\cref{corollary:sdpi_ub_lb_hellinger} specifies new lower bounds on the R\'enyi-SDPI constant, and it is thus natural to compare them to the $\eta_{\chi^2}$-based ones. An improvement can be expected since in general $\eta_{\chi^2}(\mu,K)\leq\eta_{\calH_\alpha}(\mu,K)$ and $\eta_{\chi^2}(K)\leq\eta_{\calH_\alpha}(K)$, as illustrated in~\cref{fig:di_bounds_on_renyi_sdpi_via_hellinger,fig:dd_bounds_on_renyi_sdpi} for $\alpha>1$. Part b) of~\cref{corollary:sdpi_ub_lb_hellinger}, where $\eta_{\calH_\alpha}$ instead acts as an upper bound, is illustrated in~\cref{fig:dd_bounds_on_renyi_sdpi} for $\alpha\in(0,1)$. In both figures, the SDPI constants are approximated numerically by grid search, which is feasible here because the underlying optimisation problems are low-dimensional.
        \begin{figure}
            \centering
            \includegraphics[width=0.8\linewidth]{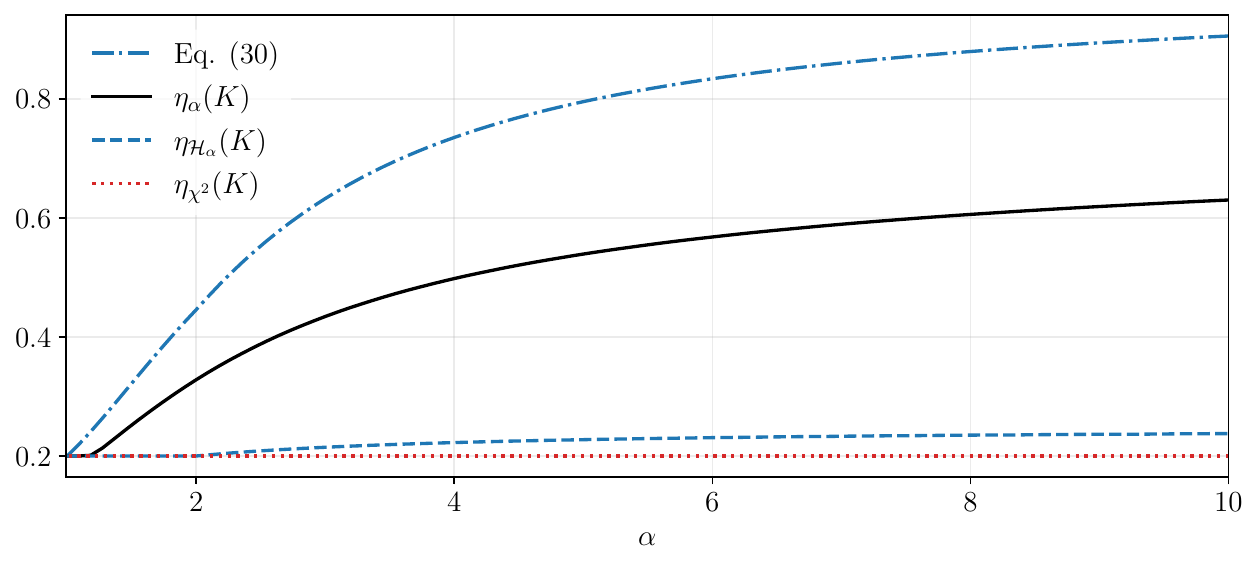}
            \caption{Comparison between $\eta_\alpha(K)$ and the bounds in~\cref{thm:local_chi_squared_lb_on_renyi_sdpi,corollary:sdpi_ub_lb_hellinger,prop:sdpi_ub_hellinger} for $\alpha\in(1, 10]$, where $K=\big[\begin{smallmatrix}0.5 & 0.5\\ 0.1 & 0.9\end{smallmatrix}\big]$.}
            \label{fig:di_bounds_on_renyi_sdpi_via_hellinger}
        \end{figure}

        \begin{figure}
            \centering
            \includegraphics[width=\textwidth]{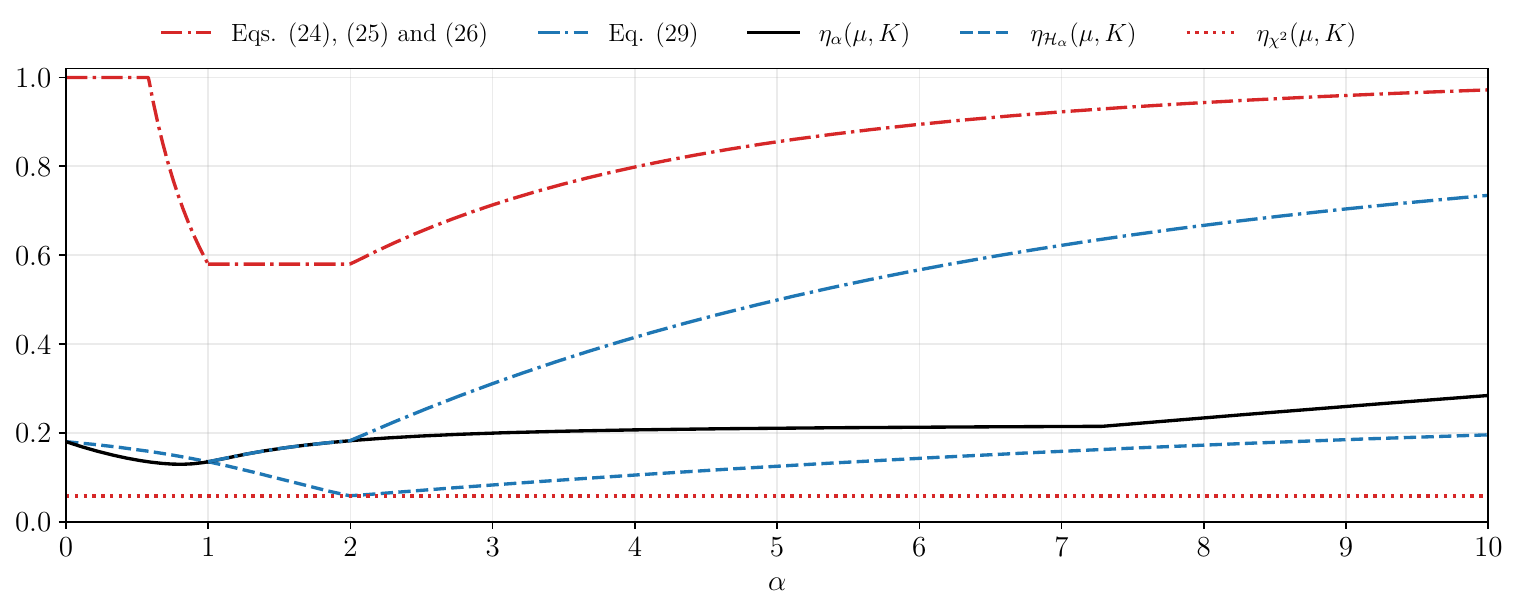}
            \caption{Comparison between $\eta_\alpha(\mu, K)$ and the bounds from~\cref{thm:local_chi_squared_lb_on_renyi_sdpi,prop:ub_on_eta_alpha_by_eta_chi_squared,corollary:sdpi_ub_lb_hellinger,prop:sdpi_ub_hellinger} for $\alpha\in(0,1)\cup(1,10]$, where $\mu=[0.9,0.1]$ and $K=\big[\begin{smallmatrix}0.5 & 0.5\\ 0.1 & 0.9\end{smallmatrix}\big]$.}
            \label{fig:dd_bounds_on_renyi_sdpi}
        \end{figure}
        Furthermore, the bounds in~\cref{corollary:sdpi_ub_lb_hellinger} enable us to revisit a fundamental property shared by all general $\varphi$-Divergences: the universal upper bound $\eta_\varphi(K) \leq\eta_{\sf TV}(K)$~\cite[Theorem III.1]{raginsky2016strong}. While it is known that this inequality does not extend to R\'enyi divergences of arbitrary order~\cite[Example 1]{esposito2024lower}, our Hellinger-based approach allows us to prove that this relationship is actually preserved in some cases. This is made precise in the statement below, which delineates the range of $\alpha$ where the Dobrushin coefficient upper bound is valid, as well as the region where it is violated.
        \begin{proposition}\label{prop:bound_on_distribution_independent_sdpi}
            Consider $\alpha\geq 0$. Then:
            \begin{enumerate}[a)]
                \item If $\alpha\in[0,1]$, $\eta_{\alpha}(K) \leq \eta_{\sf TV}(K)$ for any Markov kernel $K\in\calP(\Y|\X)$,
                \item If $\alpha>1$, there exists a Markov kernel $K\in\calP(\Y|\X)$ such that $\eta_{\alpha}(K) > \eta_{\sf TV}(K)$.
            \end{enumerate}
            \appendixproofref{appendix:proof_bound_on_distribution_independent_sdpi}
        \end{proposition}
        \begin{example}
            To illustrate Part b) of~\cref{prop:bound_on_distribution_independent_sdpi}, consider the $Z$-channel $K = \big[\begin{smallmatrix}
                1 & 0\\
                1-\lambda & \lambda
            \end{smallmatrix}\big]$ with $\lambda\in(0,1)$. One has $\eta_\alpha(K)>\eta_{\sf TV}(K)$ since $\eta_{\sf TV}(K)=\lambda$, and as seen in~\cref{example:z_channel}, $\eta_\alpha(K)=1$.
        \end{example}
        \Cref{prop:bound_on_distribution_independent_sdpi} indicates that the upper bound $\eta_\alpha(K) \leq \eta_{\sf TV}(K)$ is only guaranteed when $\alpha\in[0,1]$, unlike the universal lower bound $\eta_\alpha(K)\geq \eta_{\chi^2}(K)$, which holds for $\alpha>0$. Nonetheless, we can still upper bound $\eta_\alpha$ as a function of a $\varphi$-Divergence when $\alpha>1$. Once again, this is made possible thanks to the relationship between R\'enyi and Hellinger Divergences.
        \begingroup
        \thmboxoptions{nocut}
        \begin{proposition}\label{prop:sdpi_ub_hellinger}
            For any pair $(\mu, K)\in\calP(\X)\times\calP(\Y|\X)$ and $\alpha>1$,
            \begin{equation}
                \eta_\alpha(\mu, K) \leq \frac{\log\big(1+(\alpha-1)M_\alpha(\mu, K)\big)}{\log\left(1+(\alpha-1)\frac{M_\alpha(\mu, K)}{\eta_{\calH_\alpha}(\mu, K)}\right)},
            \end{equation}
            where $M_\alpha(\mu, K) \triangleq \max_{x\in\supp(\mu)}\calH_\alpha(\delta_x K\|\mu K)$. Moreover,
            \begin{equation}
                \eta_\alpha(K) \leq \frac{\log\big(1+(\alpha-1)M_\alpha(K)\big)}{\log\left(1+(\alpha-1)\frac{M_\alpha(K)}{\eta_{\calH_\alpha}(K)}\right)},
            \end{equation}
            where $M_\alpha(K) \triangleq \sup_{x, x^\prime\in\X}\calH_\alpha\big(K(\cdot|x)\|K(\cdot|x^\prime)\big)$. In either inequality, the right-hand side is understood to equal zero when the corresponding $M_\alpha$ equals zero.
            \appendixproofref{appendix:proof_sdpi_ub_hellinger}
        \end{proposition}
        \endgroup
        The distribution-independent and distribution-dependent bounds in~\cref{prop:sdpi_ub_hellinger} are illustrated in~\cref{fig:di_bounds_on_renyi_sdpi_via_hellinger,fig:dd_bounds_on_renyi_sdpi}, respectively. We now provide an example showing that the distribution-dependent bound is tight for every order $\alpha>1$.
        \begin{example}
            Let $\X=\{0,1\}$, $\Y=\{0,1,e\}$, and consider
            $K=\mathrm{BEC}(\delta)$ with $\delta\in(0,1)$ and
            $\mu=\left[\frac12,\frac12\right]$. For any
            $\nu\in\calP(\X)$, we have
            \begin{align}
                \calH_\alpha(\nu K\|\mu K)
                &=\frac{\sum_{y\in\Y}\nu K(y)^\alpha\mu K(y)^{1-\alpha}-1}{\alpha-1}\nonumber\\
                &=\frac{\delta+(1-\delta)\sum_{x\in\X}\nu(x)^\alpha\mu(x)^{1-\alpha}-1}{\alpha-1}\nonumber\\
                &=(1-\delta)\calH_\alpha(\nu\|\mu),\label{eq:hellinger_contraction_bec}
            \end{align}
            so that $\eta_{\calH_\alpha}(\mu,K)=1-\delta$. Next, we use~\cref{eq:hellinger_contraction_bec} to obtain
            \begin{align*}
                M_\alpha(\mu,K)&=\max\big\{\calH_\alpha(\delta_0K\|\mu K),\calH_\alpha(\delta_1K\|\mu K)\big\}\\
                &=(1-\delta)\max\big\{\calH_\alpha(\delta_0\|\mu),\calH_\alpha(\delta_1\|\mu)\big\}\\
                &=(1-\delta)\frac{2^{\alpha-1}-1}{\alpha-1},
            \end{align*}
            so that the upper bound from~\cref{prop:sdpi_ub_hellinger} becomes
            \begin{equation*}
                \frac{\log\bigl(1+(\alpha-1)M_\alpha(\mu,K)\bigr)}{\log\left(1+(\alpha-1)\frac{M_\alpha(\mu,K)}{\eta_{\calH_\alpha}(\mu,K)}\right)} = \frac{\log\bigl(\delta+(1-\delta)2^{\alpha-1}\bigr)}{(\alpha-1)\log 2}.
            \end{equation*}
            On the other hand, for every $\nu$ such that $0<D_\alpha(\nu\|\mu)<\infty$, \cref{eq:hellinger_contraction_bec} also gives,
            \begin{equation*}
                \frac{D_\alpha(\nu K\|\mu K)}
                     {D_\alpha(\nu\|\mu)}
                =
                \frac{\log\bigl(1+(\alpha-1)(1-\delta)
                \calH_\alpha(\nu\|\mu)\bigr)}
                     {\log\bigl(1+(\alpha-1)
                \calH_\alpha(\nu\|\mu)\bigr)}.
            \end{equation*}
            The right-hand side is non-decreasing as a function of
            $\calH_\alpha(\nu\|\mu)$ (see~\cref{corollary:log_ratio_increasing}), and is therefore maximised at
            $\nu=\delta_0$ or $\nu=\delta_1$, so that
            \begin{equation*}
                \eta_\alpha(\mu,K) = \frac{\log\bigl(\delta+(1-\delta)2^{\alpha-1}\bigr)}{(\alpha-1)\log 2}.
            \end{equation*}
            Thus, the distribution-dependent upper bound in
            \cref{prop:sdpi_ub_hellinger} is tight for every $\alpha>1$.
        \end{example}
        We next show that, at order $\alpha=2$, the distribution-dependent bound in~\cref{prop:sdpi_ub_hellinger} recovers the known exact R\'enyi-SDPI constant of the Binary Symmetric Channel.
        \begin{example}
            Let $\alpha=2$ and consider $K=\mathrm{BSC}(\varepsilon)$ with $\mu=[p, 1-p]$ where $p\in(0,1)$. Since $\mu K=[q, 1-q]$ with $q=p(1-\varepsilon)+(1-p)\varepsilon$, we have
            \begin{equation}\label{eq:chi_square_evaluation_bsc}
                \chi^2(\nu\|\mu) = \frac{(r-p)^2}{p(1-p)} \qquad\text{and}\qquad\chi^2(\nu K\|\mu K) = \frac{(r-p)^2(1-2\varepsilon)^2}{q(1-q)}
            \end{equation}
            for any $\nu=[r, 1-r]$ with $r\in[0,1]$. Hence, $\eta_{\chi^2}(\mu, K)=\frac{p(1-p)}{q(1-q)}(1-2\varepsilon)^2$. Next, we evaluate $M_2(\mu, K)$ in two ways in order to make our bound easier to read. First, from~\cref{eq:chi_square_evaluation_bsc}, we find
            \begin{align*}
                \chi^2(\delta_0K\|\mu K) = \frac{(1-p)^2(1-2\varepsilon)^2}{q(1-q)}=\frac{1-p}{p}\eta_{\chi^2}(\mu, K)\quad\text{and}\quad\chi^2(\delta_1K\|\mu K)= \frac{p^2(1-2\varepsilon)^2}{q(1-q)}=\frac{p}{1-p}\eta_{\chi^2}(\mu, K),
            \end{align*}
            \thmboxsplit
            so that $M_2(\mu, K)=\max\left(\frac{p}{1-p}, \frac{1-p}{p}\right)\eta_{\chi^2}(\mu, K)$. Alternatively, viewing the $\chi^2$-Divergence as the $\varphi$-Divergence stemming from $\varphi(t)=t^2-1$ gives
            \begin{equation*}
                M_2(\mu, K) = \max\left(\frac{(1-\varepsilon)^2}{q} + \frac{\varepsilon^2}{1-q}, \frac{\varepsilon^2}{q} + \frac{(1-\varepsilon)^2}{1-q}\right) -1.
            \end{equation*}
            Using the expressions for $\eta_{\chi^2}(\mu, K)$ and $M_2(\mu, K)$, \cref{prop:sdpi_ub_hellinger} yields
            \begin{align*}
                \eta_{2}(\mu, K) &\leq \frac{\log\big(1+M_2(\mu, K)\big)}{\log\left(1+\frac{M_2(\mu, K)}{\eta_{\chi^2}(\mu, K)}\right)}\\
                &=\frac{\log\max\left(\frac{(1-\varepsilon)^2}{q} + \frac{\varepsilon^2}{1-q}, \frac{\varepsilon^2}{q} + \frac{(1-\varepsilon)^2}{1-q}\right)}{\log\left(1+\max\left(\frac{p}{1-p}, \frac{1-p}{p}\right)\right)}\\
                &=\frac{\log\max\left(\frac{(1-\varepsilon)^2}{q} + \frac{\varepsilon^2}{1-q}, \frac{\varepsilon^2}{q} + \frac{(1-\varepsilon)^2}{1-q}\right)}{\log\max\left(\frac1p, \frac1{1-p}\right)},
            \end{align*}
            which matches with the exact value computed in~\cite[Corollary 6]{jin2024properties}, showing that the bound is tight. In particular when ${\mu=\left[\frac12, \frac12\right]}$, we recover $\eta_2(\mu, K) = \log_2\big(2\big(1-2\varepsilon(1-\varepsilon)\big)\big)$.
        \end{example}
    \subsection{Bounds at Order $\alpha=\infty$}\noindent
        In~\cref{sec:exact_characterisations}, we proved an exact characterisation for $\eta_\infty(K)$, offering a simple formula to evaluate it. However, for the distribution-dependent case, even with the simpler expression
        \begin{equation*}
            \eta_\infty(\mu, K) = \max_{\substack{A\in\Sigma_\X:\\0<\mu(A)<1}}\frac{D_\infty\left(\mu_{|A} K\|\mu K\right)}{D_\infty\left(\mu_{|A}\|\mu\right)}
        \end{equation*}
        put forth in~\cref{prop:infty_sdpi_achieved_on_conditionals}, its computation can still be prohibitive. This can be addressed by bounding this quantity with the SDPI constant of a different divergence that is more amenable to analytical evaluation, which is the focus of this section.

        Before we start, let us emphasise that the bounding strategies developed thus far rely on the assumption that the order $\alpha$ is finite. However, the behaviour of the R\'enyi-SDPI constant undergoes a transition as $\alpha \to \infty$. In this regime, the R\'enyi divergence evaluates to $D_\infty(\nu\|\mu)=\log\max_{x\in\supp(\mu)}\frac{\nu(x)}{\mu(x)}$, shifting the measure of discrepancy from having an expectation inside the logarithm to a maximum.

        This transition dismantles our previous bounding techniques. To see why the local approximation with the $\chi^2$-Divergence fails, consider $\nu_t$ to be a small perturbation of $\mu$, defined by the ratio $\frac{\nu_t(x)}{\mu(x)} = 1 + t \cdot h(x)$, where $t \downarrow 0$ and $\mean{\mu}{h} = 0$. For any finite $\alpha$, the function $x\mapsto x^\alpha$ is smooth, allowing a Taylor expansion of the R\'enyi Divergence:
        \begin{align*}
            D_\alpha(\nu_t\|\mu)&=\frac1{\alpha-1}\log\mathbb{E}_\mu\left[(1 + t \cdot h)^\alpha\right] \\
            &\approx \frac1{\alpha-1}\log\mathbb{E}_\mu\left[1 + \alpha t \cdot h + \frac{\alpha(\alpha-1)}{2} t^2 h^2\right]\\
            &\approx \frac1{\alpha-1}\mathbb{E}_\mu\left[\alpha t \cdot h + \frac{\alpha(\alpha-1)}{2} t^2 h^2\right]
        \end{align*}
        Because the linear term vanishes, the divergence grows on the order of $O(t^2)$, proportional to the $\chi^2$-Divergence. This behaviour is precisely what enables $\eta_\alpha(K)$ to be linked to $\eta_{\chi^2}(K)$. However, as $\alpha \to \infty$, the R\'enyi Divergence becomes $D_\infty(\nu_t\|\mu)=\log\max_{x\in\supp(\mu)}\bigl(1+t\cdot h(x)\bigr)\approx t\cdot\max_{x\in\supp(\mu)}h(x)$. The divergence now grows on the order of  $O(|t|)$, and is consequently no longer captured by the $\chi^2$-Divergence.

        Similarly, the connection to Hellinger Divergences explored in the previous section that allowed relating $\eta_\alpha$ to $\eta_{\calH_\alpha}$ becomes uninformative. This is because in the limit, the Hellinger Divergence degenerates into
        \begin{equation*}
                \lim_{\alpha \to \infty} \calH_\alpha(\nu\|\mu) =
            \begin{cases}
                  0 & \text{if } \nu = \mu \\
                  +\infty & \text{if } \nu \neq \mu
            \end{cases},
        \end{equation*}
        so that the associated SDPI constant becomes trivially zero by convention.

        Stripped of both the $\chi^2$-Divergence local approximation and the mapping to Hellinger Divergences, bounding $\eta_\infty(\mu, K)$ necessitates a different approach. Because the $\infty$-R\'enyi Divergence measures worst-case distinguishability, its contractive behaviour proves to be more closely related to that of the Total Variation Distance.

        Our first approach to bounding $\eta_\infty(\mu,K)$ exploits (reverse) Pinsker-type inequalities for R\'enyi Divergences~\cite{sason2015reverse} and follows the same strategy as in~\cref{section:linear_bounds_renyi_chi_squared_sdpi}: we upper- and lower-bound $D_\infty(\nu\|\mu)$ in terms of $\|\nu-\mu\|_{\sf TV}$.
        \begin{proposition}\label{prop:bounds_on_renyi_infty_sdpi_by_tv}
            Consider a pair $(\mu, K)\in\calP(\X)\times\calP(\Y|\X)$. Then, we have
            \begin{equation}\label{eq:ub_on_renyi_infty_sdpi_by_tv}
            \eta_\infty(\mu, K) \leq \frac{\eta_{\sf TV}(\mu, K)}{\min_{y\in\supp(\mu K)}\mu K(y)}.
            \end{equation}
            \appendixproofref{appendix:proof_bounds_on_renyi_infty_sdpi_by_tv}
        \end{proposition}
        \begin{remark}
            Unlike the Dobrushin coefficient $\eta_{\sf TV}(K)$ for which well-known characterisations such as
            \begin{equation*}
                \eta_{\sf TV}(K) = \max_{x, x^\prime\in\X}\|K(\cdot|x)-K(\cdot|x^\prime)\|_{\sf TV}
            \end{equation*}
            exist, its distribution-dependent counterpart $\eta_{\sf TV}(\mu, K)$ is comparatively not well studied. We show in~\cref{lemma:distribution_dependent_eta_tv} that in fact, a similar closed-form expression exists for $\eta_{\sf TV}(\mu, K)$.
        \end{remark}
        \begin{lemma}\label{lemma:distribution_dependent_eta_tv}
            Consider a pair $(\mu, K)\in\calP(\X)\times\calP(\Y|\X)$. Then, we have
            \begin{equation*}
                \eta_{\sf TV}(\mu, K) = \max_{x\in\X, x^\prime\in\supp(\mu)}\|K(\cdot|x)-K(\cdot|x^\prime)\|_{\sf TV}.
            \end{equation*}
            \appendixproofref{appendix:proof_distribution_dependent_eta_tv}
        \end{lemma}
        Leveraging the closed-form expression given in~\cref{lemma:distribution_dependent_eta_tv}, one can establish the lower bound $\eta_\infty(\mu, K)\geq\eta_{\sf TV}(\mu, K)$, improving on the bound in~\cite[Proposition 4]{vandenbroucque2026contraction} by removing a $\min_{x\in\X}\mu(x)$ factor. Moreover, this extends to the distribution-independent setting.
        \begin{theorem}\label{prop:eta_tv_lower_bounds_eta_inf}
            Consider a Markov kernel $K\in\calP(\Y|\X)$.
            \begin{enumerate}[a)]
                \item For every $\mu\in\calP(\X)$ with full support,
                \begin{equation}\label{eq:eta_tv_lower_bounds_eta_inf}
                    \eta_{\sf TV}(\mu,K)\leq\eta_\infty(\mu,K),
                \end{equation}
                with equality if and only if $\eta_{\sf TV}(\mu,K)\in\{0,1\}$.
                \item Moreover,
                \begin{equation}\label{eq:eta_tv_lower_bounds_eta_inf_distribution_independent}
                    \eta_{\sf TV}(K)\leq\eta_\infty(K),
                \end{equation}
                with equality if and only if $\eta_{\sf TV}(K)\in\{0,1\}$.
            \end{enumerate}
            \appendixproofref{appendix:proof_eta_tv_lower_bounds_eta_inf}
        \end{theorem}
        The Total Variation contraction coefficient plays a distinguished role in the study of $\varphi$-Divergences through the universal bound $\eta_\varphi(K)\leq\eta_{\sf TV}(K)$~\cite[Theorem~III.1]{raginsky2016strong}. Part b) of the preceding theorem places the order-$\infty$ R\'enyi-SDPI constant above this bound. Together with~\cref{prop:distribution_independent_sdpi_equal_chi_square_sdpi,corollary:distribution_independent_order_properties}, we therefore obtain
        \begin{equation*}
            \eta_\varphi(K)\leq\eta_{\sf TV}(K)\leq\eta_\infty(K),\qquad\text{and}\qquad \eta_\alpha(K)\leq\eta_\infty(K),\quad\alpha\in[0,\infty].
        \end{equation*}
        Thus, $\eta_\infty(K)$ extends the role traditionally played by $\eta_{\sf TV}(K)$ by providing a common upper bound for all $\varphi$-Divergence contraction coefficients and all R\'enyi orders.
        \begin{remark}
            To understand what happens when $\sfS\triangleq\supp(\mu)\subsetneq\X$, consider the two quantities
            \begin{equation*}
                D_{\rm intra} \triangleq \max_{x, x^\prime\in \sfS}\|K(\cdot|x)-K(\cdot|x^\prime)\|_{\sf TV}\text{ and }D_{\rm inter} \triangleq\max_{x\notin \sfS, x^\prime\in \sfS}\|K(\cdot|x)-K(\cdot|x^\prime)\|_{\sf TV}.
            \end{equation*}
            First, observe that in the definition $\eta_\infty(\mu, K)=\sup_{\nu\in\calP(\X):\nu\ll\mu\text{ and }\nu\neq\mu}\frac{D_\infty(\nu K\|\mu K)}{D_\infty(\nu\|\mu)}$, the condition $\nu\ll\mu$ implies $\eta_\infty(\mu, K)=\eta_\infty\left(\tilde{\mu}, \tilde{K}\right)$, where $\tilde{\mu}\in\calP(\sfS)$ is the fully-supported restriction of $\mu$ with $\tilde{\mu}(x)=\mu(x)\;\forall x\in \sfS$ and $\tilde{K}\in\calP(\Y|\sfS)$ is the reduced Markov kernel containing only the conditionals $K(\cdot|x)$ with $x\in \sfS$. Since $\tilde{\mu}$ has full support, \cref{prop:eta_tv_lower_bounds_eta_inf} applies and hence
            \begin{equation*}
                \eta_\infty(\mu, K) = \eta_\infty\left(\tilde{\mu}, \tilde{K}\right) \geq \eta_{\sf TV}\left(\tilde{\mu}, \tilde{K}\right) = D_{\rm intra}.
            \end{equation*}
            Moreover, by~\cref{lemma:distribution_dependent_eta_tv}, we have $\eta_{\sf TV}(\mu, K)=\max(D_{\rm intra}, D_{\rm inter})$, so that $\eta_\infty(\mu, K)\geq\eta_{\sf TV}(\mu, K)$ if and only if $\eta_\infty\left(\tilde{\mu}, \tilde{K}\right)\geq D_{\rm inter}$. A simple example of a sufficient condition which guarantees it is $D_{\rm intra}\geq D_{\rm inter}$.
        \end{remark}
        Next, we show through an example that the requirement that $\mu$ is fully supported in~\cref{prop:eta_tv_lower_bounds_eta_inf} is crucial in order to get $\eta_{\sf TV}(\mu, K) \leq \eta_{\infty}(\mu, K)$.
        \begingroup
        \thmboxoptions{nocut}
        \begin{example}
            Consider $\X=\{0,1,2\}$, $\Y=\{0,1\}$, and
            \begin{equation*}
                K = \begin{bmatrix}
                    0.5 & 0.5 \\
                    0.5 & 0.5\\
                    1. &0.
                \end{bmatrix} \text{ and } \mu=[0.5, \;0.5,\; 0].
            \end{equation*}
            We have $D_{\rm intra}=0$ and $D_{\rm inter}=0.5$, so that $\eta_{\sf TV}(\mu, K)=0.5$, while we trivially get $\eta_\infty(\mu, K)=0$. Hence, $\eta_{\sf TV}(\mu, K) \not\leq \eta_\infty(\mu, K)$.
        \end{example}
        \endgroup
        \begin{figure}
            \centering
            \includegraphics[width=0.75\linewidth]{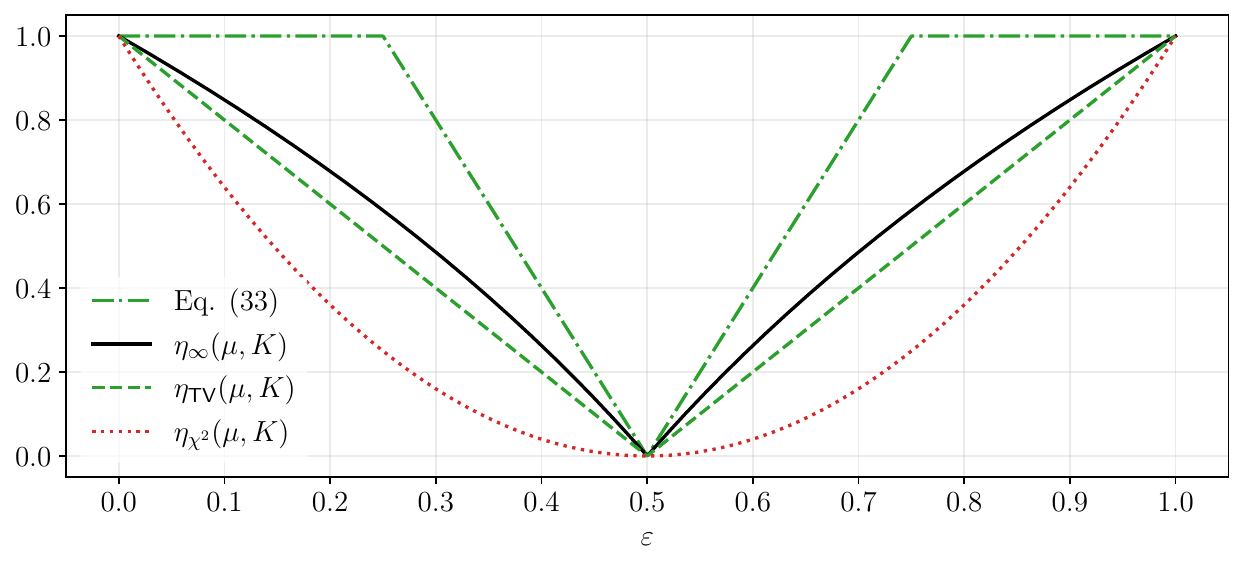}
            \caption{Bounds on $\eta_\infty(\mu, K)$ from~\cref{prop:bounds_on_renyi_infty_sdpi_by_tv,prop:eta_tv_lower_bounds_eta_inf}, where $K=\mathrm{BSC}(\varepsilon)$ and $\mu=[0.5,0.5]$. For comparison, $\eta_{\chi^2}(\mu, K)$ is also plotted.}
            \label{fig:bounds_on_renyi_infty_sdpi_by_tv}
        \end{figure}

        \Cref{fig:bounds_on_renyi_infty_sdpi_by_tv} displays the bounds in~\cref{prop:bounds_on_renyi_infty_sdpi_by_tv,prop:eta_tv_lower_bounds_eta_inf} for the choice $K=\mathrm{BSC}(\varepsilon)$ and $\mu = \left[\frac12, \frac12\right]$. In this setting, ${\eta_\infty(\mu, K) = 1-\log_2\frac1{\max(\varepsilon, 1-\varepsilon)}}$ (see~\cref{example:dd_infty_sdpi_bsc}) and ${\eta_{\sf TV}(\mu, K)=\eta_{\sf TV}(K)=|1-2\varepsilon|}$. The upper bound from~\cref{eq:ub_on_renyi_infty_sdpi_by_tv} gives a non-trivial result when $\varepsilon\in(\frac14, \frac34)$, and is tight for $\varepsilon=\frac12$. The lower bound from~\cref{eq:eta_tv_lower_bounds_eta_inf}, on the other hand, is always non-trivial. Moreover, it is tight for $\varepsilon\in\{0, \frac12, 1\}$.
\section{Applications}\label{sec:applications}\noindent
    \subsection{Connection to Local Differential Privacy}\noindent
        Local Differential Privacy (LDP) protects information about an input by randomising it before disclosure. It requires any two possible inputs to induce similar distributions over the disclosed output, thereby limiting what can be inferred about the original input. Such a randomisation is described by a privacy mechanism, represented by a Markov kernel $K\in\calP(\Y|\X)$. Information-theoretic inequalities have been used to study how much information can be preserved under this privacy constraint~\cite{DuchiJordanWainwright2013}. Privacy controls how $K(\cdot|x)$ varies with $x$, whereas contraction coefficients quantify how divergences between arbitrary input distributions change after applying $K$. For Hockey-Stick Divergences, this connection is exact: $(\varepsilon,\delta)$-LDP is equivalent to a bound on the corresponding contraction coefficient~\cite{asoodeh2020contraction}. Motivated by this equivalence, we investigate the corresponding relationship for R\'enyi Divergences. We begin by recalling the definition of pure LDP~\cite{KairouzOhViswanath2016}.
        \begin{definition}
            Let $\varepsilon\geq0$. A privacy mechanism $K\in\calP(\Y|\X)$ is $\varepsilon$-LDP if
            \begin{equation*}
                \sup_{x, x^\prime\in\X}\sup_{B\in\Sigma_\Y}\frac{K(B|x^\prime)}{K(B|x)}\leq e^\varepsilon,
            \end{equation*}
            with the conventions $\frac00=0$ and $\frac{a}{0}=\infty$ for $a>0$.
        \end{definition}
        R\'enyi Differential Privacy instead requires $D_\alpha\bigl(K(\cdot|x)\|K(\cdot|x^\prime)\bigr)$ to be bounded whenever $x$ and $x^\prime$ are neighbouring inputs~\cite{Mironov2017}. In the local setting, this requirement is imposed on every pair of inputs~\cite{GilaniKurriKosutSankar2024}, leading to the following definition.
        \begin{definition}
            Let $\alpha\in(1,\infty]$ and $\varepsilon\geq0$. A privacy mechanism $K\in\calP(\Y|\X)$ is $(\alpha,\varepsilon)$-RLDP if
            \begin{equation}\label{eq:renyi_diameter}
                \sup_{x,x^\prime\in\X}D_\alpha\bigl(K(\cdot|x)\|K(\cdot|x^\prime)\bigr)\leq\varepsilon.
            \end{equation}
        \end{definition}
        Recent work has studied the amplification of an existing RLDP guarantee under post-processing using contraction estimates for a rescaled version of Hellinger Divergences~\cite{grosse2025bounds}. These estimates are tailored to the privacy mechanism and post-processing channel under consideration. Here, we instead ask what the privacy constraint alone implies about the contraction of R\'enyi Divergence, uniformly over all input distributions.

        We first show that RLDP implies a bound on the R\'enyi-SDPI constant at the same order. For $\alpha=\infty$, where $(\infty,\varepsilon)$-RLDP coincides with $\varepsilon$-LDP, this implication is reversible.
        \begin{theorem}\label{thm:renyi_ldp_implies_contraction}
            Let $\alpha\in[2,\infty]$ and $\varepsilon\geq0$. If a privacy mechanism $K\in\calP(\Y|\X)$ is $(\alpha,\varepsilon)$-RLDP, then
            \begin{equation*}
                \eta_\alpha(K)\leq1-e^{-\varepsilon}.
            \end{equation*}
            When $\alpha=\infty$, the reverse implication also holds, so that
            \begin{equation*}
                K \text{ is $\varepsilon$-LDP} \quad\iff \quad\eta_\infty(K)\leq1-e^{-\varepsilon}.
            \end{equation*}
            \appendixproofref{appendix:proof_renyi_ldp_implies_contraction}
        \end{theorem}
        \begin{remark}
            The reverse implication fails at every finite order $\alpha\in[2,\infty)$. Indeed, the proof of~\cref{thm:renyi_ldp_implies_contraction} shows that the resulting bound on $\eta_\alpha(K)$ is strict whenever the supremum in~\cref{eq:renyi_diameter} is in $(0,\infty)$. Consequently, any $\varepsilon$ satisfying
            \begin{equation*}
                -\log\bigl(1-\eta_\alpha(K)\bigr)<\varepsilon<\sup_{x,x^\prime\in\X}D_\alpha\bigl(K(\cdot|x)\|K(\cdot|x^\prime)\bigr)
            \end{equation*}
            gives $\eta_\alpha(K)<1-e^{-\varepsilon}$ even though $K$ is not $(\alpha,\varepsilon)$-R\'enyi LDP.
        \end{remark}
        We now return to pure LDP. As mentioned above, pure LDP is equivalent to $(\infty,\varepsilon)$-RLDP. This fact, together with Part~a) of~\cref{corollary:distribution_independent_order_properties}, shows that every $\varepsilon$-LDP mechanism $K$ satisfies
        \begin{equation*}
            \eta_\alpha(K)\leq\eta_\infty(K)\leq1-e^{-\varepsilon},\qquad\alpha\in[1,\infty].
        \end{equation*}
       Since this bound is inherited from order $\infty$, a natural question is whether it can be sharpened at finite orders. For comparison, \cite[Theorem~1]{asoodeh2024contraction} establishes that every $\varepsilon$-LDP mechanism $K$ satisfies
        \begin{equation*}
            \eta_1(K)\leq\left(\frac{e^\varepsilon-1}{e^\varepsilon+1}\right)^2,
        \end{equation*}
        and that this bound is tight. We next identify the corresponding sharp bound for R\'enyi Divergences at every order.
        \begin{theorem}\label{thm:ldp_sharp_renyi_contraction}
            Let $\alpha\in[0,\infty]$ and $\varepsilon\geq0$. Every $\varepsilon$-LDP mechanism $K$ satisfies
            \begin{equation*}
                \eta_\alpha(K)\leq\Upsilon_{\alpha, \varepsilon}\triangleq\eta_\alpha\left(\mathrm{BSC}\left(\frac1{e^\varepsilon+1}\right)\right).
            \end{equation*}
            \appendixproofref{appendix:proof_ldp_sharp_renyi_contraction}
        \end{theorem}
        \begin{remark}
            The channel $\mathrm{BSC}\left(\frac1{e^\varepsilon+1}\right)$ is precisely the binary randomised response mechanism with privacy parameter $\varepsilon$. Since it is itself $\varepsilon$-LDP, the bound in~\cref{thm:ldp_sharp_renyi_contraction} is attained. Moreover, Part~a) of~\cref{corollary:distribution_independent_order_properties} shows that $\Upsilon_{\alpha,\varepsilon}$ is non-decreasing in $\alpha$. Its values at orders one and infinity therefore give
            \begin{equation*}
                \left(\frac{e^\varepsilon-1}{e^\varepsilon+1}\right)^2=\Upsilon_{1,\varepsilon}\leq\Upsilon_{\alpha,\varepsilon}\leq\Upsilon_{\infty,\varepsilon}=1-e^{-\varepsilon},\qquad\alpha\in[1,\infty].
            \end{equation*}
        \end{remark}
        Taken together, these results further develop the existing connection between local privacy and divergence contraction in the setting of R\'enyi Divergences. Parallel developments for Hockey-Stick and more general $\varphi$-Divergences include non-linear contraction inequalities under pure LDP and extensions to approximate LDP~\cite{zamanlooy2023strong,nuradha2025non}. This broader line of work suggests investigating how other privacy constraints relate to the contraction of R\'enyi Divergences.
    \subsection{Convergence of Markov Chains}\label{section:mixing_time}\noindent
        Markov chain Monte Carlo methods generate approximate samples from a target distribution by simulating a Markov chain for which the target is stationary~\cite{gilks1995markov}. The quality of this approximation depends on how rapidly the chain approaches stationarity, which motivates the study of quantitative convergence bounds. To describe this convergence, recall that a Markov chain on $\X$ is determined by a Markov kernel $K\in\calP(\X|\X)$. When initialised according to $\nu\in\calP(\X)$, the distribution of the chain after $t$ steps is $\nu K^t$. Classical Markov chain theory characterises the asymptotic behaviour of this sequence~\cite{levin2017markov,montenegro2006mathematical}. If the chain is irreducible and aperiodic, then it admits a unique stationary distribution $\pi\in\calP(\X)$ satisfying $\pi K=\pi$, and $\lim_{t\to\infty}\|\nu K^t-\pi\|_{\sf TV}=0$ for all $\nu\in\calP(\X)$.

        This guarantee, however, is asymptotic and does not by itself specify how close $\nu K^t$ is to $\pi$ after a finite number of steps. Finite-time convergence is commonly quantified in Total Variation Distance~\cite[Chapter~4]{levin2017markov}. For reversible chains, spectral methods also yield bounds in $\chi^2$-Divergence, or equivalently in $L^2$-norm~\cite[Section~12.6]{levin2017markov}\cite[Section~2.1]{montenegro2006mathematical}. SDPIs place these approaches within a common framework: since $\pi K^t=\pi$, the SDPI associated with a divergence $D$ controls $D(\nu K^t\|\pi)$ in terms of $D(\nu\|\pi)$. Applying this viewpoint to R\'enyi Divergences gives, for every integer $t\geq1$,
        \begin{equation*}
            D_\alpha(\nu K^t\|\pi)\leq\eta_\alpha(\pi,K^t)D_\alpha(\nu\|\pi).
        \end{equation*}
        The coefficient $\eta_\alpha(\pi,K^t)$ thus directly controls the speed of convergence. For reversible chains, the bounds in~\cref{prop:ub_on_eta_alpha_by_eta_chi_squared} allow us to determine its exact asymptotic behaviour. Recall that a Markov chain with stationary distribution $\pi$ is reversible if it satisfies the detailed-balance condition
        \begin{equation*}
            \pi(x)K(y|x)=\pi(y)K(x|y),\qquad x,y\in\X,
        \end{equation*}
        which in view of the definition of the dual kernel in~\cref{eq:adjoint_relation}, is equivalent to $K_\pi^\star=K$.
        \begin{corollary}\label{corollary:asymptotic_renyi_contraction}
            Consider a Markov kernel $K\in\calP(\X|\X)$ inducing an irreducible, aperiodic and reversible Markov chain, and let $\pi\in\calP(\X)$ be its stationary distribution. Then, for every $\alpha\in(0,\infty)$,
            \begin{equation*}
                \lim_{t\to\infty}\eta_\alpha(\pi,K^t)^{\frac1t}=\eta_{\chi^2}(\pi,K).
            \end{equation*}
            \appendixproofref{appendix:proof_asymptotic_renyi_contraction}
        \end{corollary}
        The preceding corollary shows that, for every finite positive order, $\eta_\alpha(\pi,K^t)$ decays at the same exponential rate as $\eta_{\chi^2}(\pi,K)^t$. Evaluating $\eta_\alpha(\pi,K^t)$ directly is generally difficult, since it requires analysing $K^t$. In practice, finite-time bounds are instead obtained by iterating the one-step SDPI, which gives $D_\alpha(\nu K^t\|\pi)\leq\eta_\alpha(\pi,K)^tD_\alpha(\nu\|\pi)$. We follow this approach below and show that R\'enyi-SDPIs can improve upon classical $\chi^2$-based convergence bounds.

        The one-to-one relationship between R\'enyi and Hellinger Divergences allows us to express the R\'enyi-SDPI as a bound involving Hellinger Divergences.
        \begin{proposition}\label{prop:mcmt}
            Let $K\in\calP(\X|\X)$ and $\pi\in\calP(\X)$ be such that $\pi K=\pi$. Then, for every $\alpha\in(0,1)\cup(1,\infty)$, every integer $t\geq1$, and every $\nu\in\calP(\X)$ such that $\nu\ll\pi$,
            \begin{equation}
                \calH_\alpha(\nu K^t\|\pi)\leq\frac{\big(1+(\alpha-1)\calH_\alpha(\nu\|\pi)\big)^{\eta_\alpha(\pi,K)^t}-1}{\alpha-1}.\label{eq:mcmt_non_linear_sdpi}
            \end{equation}
            \appendixproofref{appendix:proof_mcmt}
        \end{proposition}
        Contraction bounds of this form have previously appeared in the literature~\cite{polyanskiy2017strong,du2017strong,gu2023non}, where they are referred to as ``non-linear''.

        We next compare the non-linear bound in~\cref{eq:mcmt_non_linear_sdpi} with the corresponding linear Hellinger-contraction bound. For general $\alpha$, this comparison would require evaluating $\eta_{\calH_\alpha}(\pi,K)$, which is generally difficult. We therefore focus on $\alpha=2$, where
        \begin{equation*}
            \calH_2(\nu\|\pi)=\chi^2(\nu\|\pi)=\left\|\frac{\dd\nu}{\dd\pi}-1\right\|_{L^2(\pi)}^2
        \end{equation*}
        and $\eta_{\calH_2}(\pi, K)=\eta_{\chi^2}(\pi, K)$. In this case, the linear Hellinger-contraction bound reduces to the classical $\chi^2$-SDPI, while~\cref{eq:mcmt_non_linear_sdpi} becomes $ \chi^2(\nu K^t\|\pi)\leq\left(1+\chi^2(\nu\|\pi)\right)^{\eta_2(\pi,K)^t}-1$. Whenever $\eta_2(\pi,K)<1$, a Taylor expansion gives
        \begin{equation*}
            \left(1+\chi^2(\nu\|\pi)\right)^{\eta_2(\pi,K)^t}-1=\eta_2(\pi,K)^t\log\left(1+\chi^2(\nu\|\pi)\right)+O\left(\eta_2(\pi,K)^{2t}\right)
        \end{equation*}
        for every fixed $\nu\ll\pi$. This expansion shows that, for large $t$, comparing the non-linear and classical $\chi^2$-bounds amounts to comparing
        \begin{equation}\label{eq:mcmt_first_order_comparison}
            \eta_2(\pi,K)^t\log\left(1+\chi^2(\nu\|\pi)\right)\qquad\text{and}\qquad\eta_{\chi^2}(\pi,K)^t\chi^2(\nu\|\pi),
        \end{equation}
        respectively. This comparison does not determine a priori which bound is sharper: the classical bound benefits from the smaller contraction coefficient by~\cref{thm:local_chi_squared_lb_on_renyi_sdpi}, while the non-linear bound benefits from its logarithmic dependence on the initial divergence. We illustrate this trade-off on a family of chains for which both contraction coefficients and the exact divergence can be computed explicitly.
        \begin{example}\label{example:community_chain}
            Let $\mathcal A=\{A_1,\dots,A_q\}$ be a partition of a finite set $\X$ into $q\geq2$ communities, labelled so that $|A_1|=\min_{j\in[q]}|A_j|$, and let $\pi$ be the uniform distribution on $\X$. Fix $p\in(0,1)$ and consider the Markov kernel defined, for $x\in A_j$, through
            \begin{equation*}
                K(y|x)=p\frac{\mathbbm{1}_{A_j}(y)}{|A_j|}+(1-p)\frac1{|\X|},\qquad y\in\X.
            \end{equation*}
            At each step, the chain resamples uniformly within the community containing its current state with probability $p$, and uniformly over $\X$ otherwise. The chain is reversible with respect to $\pi$, and its contraction coefficients satisfy
            \begin{equation*}
                \eta_{\chi^2}(\pi,K)=p^2,\qquad \eta_2(\pi,K)=\frac{\log\left(1+p^2\left(\frac{|\X|}{|A_1|}-1\right)\right)}{\log\left(\frac{|\X|}{|A_1|}\right)}.
            \end{equation*}
            Since $0<p<1$, the strict concavity of the logarithm gives $\eta_{\chi^2}(\pi,K)<\eta_2(\pi,K)<1$. In view of~\cref{eq:mcmt_first_order_comparison}, the classical bound is therefore eventually sharper. However, the non-linear bound may remain tighter for a substantial number of steps.

            To compare the two bounds in~\cref{eq:mcmt_first_order_comparison} uniformly over the initial distribution, we maximise their right-hand sides over $\nu$. Both are non-decreasing functions of $\chi^2(\nu\|\pi)$. Since $\pi$ is uniform, $\chi^2(\nu\|\pi)=|\X|\sum_{z\in\X}\nu(z)^2-1\leq|\X|-1$, with equality precisely when $\nu$ is a point mass. We therefore fix $x\in A_1$ and take $\nu=\delta_x$ so that the R\'enyi-SDPI and $\chi^2$-SDPI bounds become, respectively,
            \begin{equation*}
                \chi^2(\nu K^t\|\pi)\leq|\X|^{\eta_2(\pi,K)^t}-1\qquad\text{and}\qquad\chi^2(\nu K^t\|\pi)\leq p^{2t}\bigl(|\X|-1\bigr).
            \end{equation*}
            On the other hand, the actual $\chi^2$-Divergence after $t\geq1$ steps is given by
            \begin{equation*}
                \chi^2(\delta_xK^t\|\pi)=p^{2t}\left(\frac{|\X|}{|A_1|}-1\right).
            \end{equation*}
            For $q=5$, $p=0.99$, and community sizes $1000$, $1500$, $2000$, $2500$, and $3000$,~\cref{fig:mcmt_chi_squared_comparison} compares the two bounds with the exact divergence and shows that the R\'enyi-SDPI bound remains sharper for several hundred steps. The general result underlying these computations, together with its proof, is provided in Appendix~\ref{appendix:partition_chain_computations}.
        \end{example}
        \begin{figure}[t]
            \centering
            \begin{subfigure}[t]{0.49\textwidth}
                \centering
                \includegraphics[width=\linewidth]{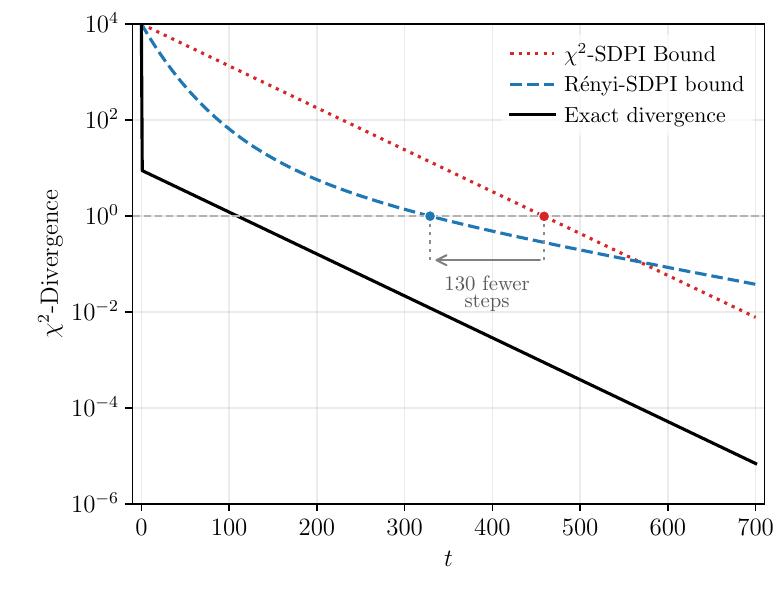}
                \caption{}
            \end{subfigure}
            \hfill
            \begin{subfigure}[t]{0.49\textwidth}
                \centering
                \includegraphics[width=\linewidth]{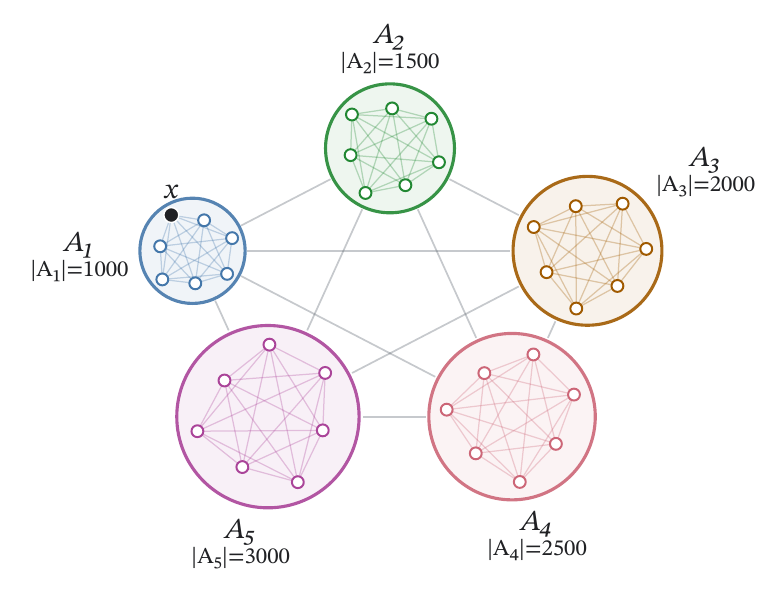}
                \caption{}
            \end{subfigure}
            \caption{Comparison of the exact $\chi^2$-Divergence with the bounds obtained from $\eta_2(\pi,K)$ and $\eta_{\chi^2}(\pi,K)$ for~\cref{example:community_chain}. At each step, the Markov chain moves to a uniformly chosen vertex of its current community with probability $p=0.99$, and to a uniformly chosen vertex of the entire graph otherwise. (a) Starting from a vertex in the smallest community $A_1$, the R\'enyi-SDPI bound ensures $\chi^2(\nu K^t\|\pi)\leq1$ after $329$ steps, compared with $459$ for the $\chi^2$-SDPI bound. (b) A depiction of the five-community graph on which the random walk is performed.}
            \label{fig:mcmt_chi_squared_comparison}
        \end{figure}
        In~\cref{example:community_chain}, the first transition erases all information about the position of the initial state within its community, so that the subsequent evolution depends only on the community containing the state. Consequently, scaling all community sizes by the same factor leaves the exact divergence unchanged for every integer $t\geq1$, even though the initial divergence $\chi^2(\delta_x\|\pi)=|\X|-1$ grows with the size of the state space. Since $\eta_2(\pi,K)$ depends on the community sizes only through $\frac{|\X|}{|A_1|}$, it is also unchanged under this scaling. Therefore, for every fixed integer $t\geq1$,
        \begin{equation*}
            \lim_{|\X|\to\infty}\frac{|\X|^{\eta_2(\pi,K)^t}-1}{p^{2t}\bigl(|\X|-1\bigr)}=0,
        \end{equation*}
        so that the multiplicative gap between the R\'enyi-SDPI and $\chi^2$-SDPI bounds can be made arbitrarily large.
\section*{Acknowledgment}\noindent
    This work was supported in part by the Swiss
    National Science Foundation under Grant 236494, and was partially conducted during the first author's visit to the Okinawa Institute of Science and Technology.
\bibliographystyle{IEEEtran}
\bibliography{references}
\newpage
\appendices
\section{Technical Results on Monotonicity}\noindent
    In the subsequent proofs, we frequently encounter the need to determine the monotonicity of a ratio of functions. The following standard result from convex analysis unifies these arguments.
    \begin{lemma}\label{lemma:secant_slope}
        Let $I\subseteq\reals_+$ be an interval containing 0, and let $g:I\to\reals$ be a function. Define $f: I\setminus\{0\}\to\reals$ through $f(t)=\frac{g(t)}{t}$.
        \begin{enumerate}[a)]
            \item\label{lemma:secant_slope_convex} If $g$ is convex (resp. strictly convex) on $I$ and $g(0)\leq 0$, then $f$ is non-decreasing (resp. increasing).
            \item\label{lemma:secant_slope_concave} If $g$ is concave (resp. strictly concave) on $I$ and $g(0)\geq 0$, then $f$ is non-increasing (resp. decreasing).
        \end{enumerate}
    \end{lemma}
    \begin{proof}
        Assuming $g$ is convex and $g(0)\leq0$, we have $g\big(\lambda t+(1-\lambda)\cdot0\big)\leq \lambda g(t)$ for $\lambda\in[0,1]$. Letting $s=\lambda t\in I$, the inequality reads $\frac{g(s)}{s}\leq\frac{g(t)}{t}$, hence $f$ is non-decreasing. If strict convexity holds, identical steps yield that $f$ is increasing. For $g$ (strictly) concave and $g(0)\geq0$, the same argument applies but with the inequalities reversed.
    \end{proof}
    We now use \cref{lemma:secant_slope} to collect a number of useful corollaries.
    \begin{corollary}\label{corollary:log_sum_exp}
        Let $a, b\geq0$ such that $0<a+b \leq1$ and let $k\neq0$. Then the function
        \begin{equation*}
            f(t) = \frac{\log(a e^{kt} + b)}{t}
        \end{equation*}
        is increasing for $t > 0$ if $ab>0$, and non-decreasing otherwise.
    \end{corollary}
    \begin{proof}
        We apply~\cref{lemma:secant_slope} with $g(t) = \log(a e^{kt} + b)$.
        Notice that $g(0) = \log(a+b)$ and the second derivative is $g^{\prime\prime}(t) = \frac{ab k^2 e^{kt}}{(a e^{kt} + b)^2}$, so that the sign of $g^{\prime\prime}(t)$ is entirely determined by the sign of $ab$. Hence, $g$ is convex (resp. strictly convex) if $ab \geq 0$ (resp. $ab > 0$) and since $a+b \leq 1$, $g(0) \le \log(1) = 0$, \cref{lemma:secant_slope}\labelcref{lemma:secant_slope_convex} gives that $f(t)$ is non-decreasing (resp. increasing).
    \end{proof}
    \begin{corollary}\label{corollary:log_ratio_increasing}
        For $\lambda\in[0,1]$, the map $t\mapsto \frac{\log\big(1+\lambda(\alpha-1) t\big)}{\log\big(1+(\alpha-1)t\big)}$ is non-decreasing on the interval $(0,\infty)$ for any $\alpha>1$.
    \end{corollary}
    \begin{proof}
        Let $s = \log\big(1+(\alpha-1)t\big)$, which is an increasing bijection from $(0,\infty)$ to $(0, \infty)$. Under this substitution, $t=\frac{e^{s}-1}{\alpha-1}$, and the function becomes $f(s)=\frac{\log(1-\lambda+\lambda e^{s})}{s}$. By~\cref{corollary:log_sum_exp} with $a=\lambda, b=1-\lambda$ and $k=1$, $f(s)$ is non-decreasing. Since the substitution used is increasing, the original map has the same monotonicity as $f(s)$.
    \end{proof}
    \begin{corollary}\label{corollary:log_ratio_decreasing}
        The map $t\mapsto \frac{\log(1-\lambda t)}{\log(1-t)}$ on the interval $(0,1)$ is non-increasing when $\lambda\in[0,1]$, and decreasing when $\lambda\in(0,1)$. Moreover, $\lim_{t\to0}\frac{\log(1-\lambda t)}{\log(1-t)}=\lambda$.
    \end{corollary}
    \begin{proof}
        Let $s = -\log(1-t)$, which is an increasing bijection from $(0,1)$ to $(0, \infty)$. Under this substitution, $t=1-e^{-s}$, and the function becomes $f(s)=\frac{-\log(1-\lambda+\lambda e^{-s})}{s}$. By~\cref{corollary:log_sum_exp} with $a=\lambda, b=1-\lambda$ and $k=-1$, $f(s)$ is non-increasing when $\lambda\in[0,1]$ and decreasing when $\lambda\in(0,1)$. Since the substitution used is increasing, the original map has the same monotonicity as $f(s)$. Finally, the limit is
        \begin{equation*}
                \lim_{t\to 0} \frac{\log\left(1-\lambda t\right)}{\log(1-t)} = \lim_{t\to 0}\frac{\frac{-\lambda}{1-\lambda t}}{\frac{-1}{1-t}} = \lambda.
        \end{equation*}
    \end{proof}
    \begin{corollary}\label{corollary:xlogx_map}
        For $\lambda \in (0,1]$, the map $t \mapsto \frac{(1+\lambda t)\log(1+\lambda t)}{t}$ is increasing on the interval $(0, \infty)$.
    \end{corollary}
    \begin{proof}
        The function is in the form $\frac{g(t)}{t}$ with $g(t) = (1+\lambda t)\log(1+\lambda t)$. The second derivative of $g$ is $g^{\prime\prime}(t) = \frac{\lambda^2}{1+\lambda t}$, and hence $g$ is strictly convex. Moreover since $g(0) = 0$, \cref{lemma:secant_slope}\labelcref{lemma:secant_slope_convex} implies that $\frac{g(t)}{t}$ is increasing.
    \end{proof}
\section{Proofs for~\cref{sec:properties_exact_characterisations}}
    \subsection{Proof of~\cref{prop:renyi_sdpi_functional_form}}\label{appendix:proof_renyi_sdpi_functional_form}\noindent
        The map $\nu\mapsto f=\frac{\mathrm d\nu}{\mathrm d\mu}$ is a bijection between probability measures $\nu\ll\mu$ and $\mathcal D_\mu$, with $\nu=\mu$ if and only if $f=\allones$. Moreover, for every $\nu\ll\mu$ with $\nu\neq\mu$, writing $f=\frac{\mathrm d\nu}{\mathrm d\mu}$, \cref{eq:dual_kernel_formula} allows us to write
        \begin{equation*}
            \frac{D_\alpha(\nu K\|\mu K)}{D_\alpha(\nu\|\mu)} = \frac{\log\left\|K_\mu^\star f\right\|_{L^\alpha(\mu K)}}{\log\left\|f\right\|_{L^\alpha(\mu)}},
        \end{equation*}
        where the identity at $\alpha=\infty$ follows from the corresponding $L^\infty$-representation. Since the admissible probability measures in the definition of $\eta_\alpha(\mu,K)$ correspond precisely to the elements of $\mathcal D_\mu\setminus\{\allones\}$, taking the supremum proves the claim.

        For $f\in\mathcal D_\mu\setminus\{\allones\}$, the variational representation and the fact that $\log\|f\|_{L^\alpha(\mu)}>0$ prove~\cref{eq:renyi_sdpi_functional_inequality} after exponentiating. For $f=\allones$, the inequality holds with equality since $K_\mu^\star\allones=\allones$.
    \subsection{A First-Order Optimality Lemma}\label{appendix:first_order_optimality}\noindent
        Several arguments below rely on the same first-order necessary condition, recorded here as a separate lemma.
        \begin{lemma}\label{lemma:first_order_optimality}
            Consider a pair $(\mu,K)\in\calP(\X)\times\calP(\Y|\X)$, let $\alpha>1$ and $\lambda\in[0,1]$, and suppose that $h\in\mathcal D_\mu$ maximises the functional
            \begin{equation*}
                \Phi_{\alpha,\lambda}(f)\triangleq\log\left\|K_\mu^\star f\right\|_{L^\alpha(\mu K)}-\lambda\log\left\|f\right\|_{L^\alpha(\mu)}
            \end{equation*}
            over $\mathcal D_\mu$. Then,
            \begin{equation}\label{eq:first_order_optimality_condition}
                \frac{K\big((K_\mu^\star h)^{\alpha-1}\big)}{\left\|K_\mu^\star h\right\|_{L^\alpha(\mu K)}^\alpha}\leq\lambda\frac{h^{\alpha-1}}{\left\|h\right\|_{L^\alpha(\mu)}^\alpha}+(1-\lambda)\allones\qquad\mu\text{-almost everywhere},
            \end{equation}
            with equality on $\calS_h\triangleq\left\{x\in\supp(\mu):h(x)>0\right\}$.
        \end{lemma}
        \begin{proof}
            Let us derive the first-order necessary condition for maximising $\Phi_{\alpha, \lambda}$ over the set $\mathcal D_\mu$. Fix $g\in\mathcal D_\mu$. By convexity of $\mathcal D_\mu$,
            \begin{equation*}
                h_t\triangleq(1-t)h+tg\in\mathcal D_\mu,\qquad 0\leq t\leq1.
            \end{equation*}
            Since $\Phi_{\alpha, \lambda}(h_t)\leq\Phi_{\alpha, \lambda}(h)$,
            \begin{equation}\label{eq:foc_general}
                \frac{\Phi_{\alpha, \lambda}(h_t)-\Phi_{\alpha, \lambda}(h)}{t}\leq0 \qquad\Longrightarrow \qquad\left.\frac{\mathrm d}{\mathrm dt}\Phi_{\alpha, \lambda}(h_t)\right|_{t=0^+}\leq0.
            \end{equation}
            Because $\alpha>1$, the map $s\mapsto|s|^\alpha$ is continuously differentiable, including at $s=0$. With the convention $0^{\alpha-1}=0$, direct differentiation and the non-negativity of $h_t$ give
            \begin{align*}
                \left.\frac{\mathrm d}{\mathrm dt}\log\|h_t\|_{L^\alpha(\mu)}\right|_{t=0^+} &= \left.\frac{\mathrm d}{\mathrm dt}\left[\frac{1}{\alpha}\log\left\langle h_t^\alpha,\allones\right\rangle_\mu\right]\right|_{t=0^+}\\
                &= \left.\frac{1}{\alpha}\frac{\left\langle\alpha h_t^{\alpha-1}(g-h), \allones\right\rangle_\mu}{\left\langle h_t^\alpha, \allones\right\rangle_\mu}\right|_{t=0^+}\\
                &= \left.\frac{\left\langle h_t^{\alpha-1},g-h\right\rangle_\mu}{\|h_t\|_{L^\alpha(\mu)}^\alpha}\right|_{t=0^+}\\
                &= \frac{\left\langle h^{\alpha-1},g-h\right\rangle_\mu}{\|h\|_{L^\alpha(\mu)}^\alpha}.
            \end{align*}
            Using the previous computation together with the linearity of $K_\mu^\star$ and the adjoint relation from~\cref{eq:adjoint_relation} leads to
            \begin{align*}
                \left.\frac{\mathrm d}{\mathrm dt}\log\| K_\mu^\star h_t\|_{L^\alpha(\mu K)}\right|_{t=0^+} &= \frac{\left\langle(K_\mu^\star h)^{\alpha-1},K_\mu^\star(g-h)\right\rangle_{\mu K}}{\| K_\mu^\star h\|_{L^\alpha(\mu K)}^\alpha}\\
                &= \frac{\left\langle K\big((K_\mu^\star h)^{\alpha-1}\big),g-h\right\rangle_\mu}{\| K_\mu^\star h\|_{L^\alpha(\mu K)}^\alpha}.
            \end{align*}
            Therefore, if we define
            \begin{equation*}
                G_{h}\triangleq\frac{K\big((K_\mu^\star h)^{\alpha-1}\big)}{\| K_\mu^\star h\|_{L^\alpha(\mu K)}^\alpha}-\lambda\frac{h^{\alpha-1}}{\| h\|_{L^\alpha(\mu)}^\alpha},
            \end{equation*}
            the first-order condition from~\cref{eq:foc_general} becomes
            \begin{equation}\label{eq:foc_precise}
                \big\langle G_{h},g-h\big\rangle_\mu\leq0\qquad\text{for every }g\in\mathcal D_\mu.
            \end{equation}

            Using the adjoint relation from~\cref{eq:adjoint_relation} once more,
            \begin{equation}\label{eq:foc_equality}
                \langle G_{h},h\rangle_\mu = \frac{\left\langle(K_\mu^\star h)^{\alpha-1},K_\mu^\star h\right\rangle_{\mu K}}{\| K_\mu^\star h\|_{L^\alpha(\mu K)}^\alpha}-\lambda\frac{\left\langle h^{\alpha-1},h\right\rangle_\mu}{\| h\|_{L^\alpha(\mu)}^\alpha} = 1-\lambda,
            \end{equation}
            so that~\cref{eq:foc_precise} becomes
            \begin{equation*}
                \langle G_{h},g\rangle_\mu\leq1-\lambda\qquad\text{for every }g\in\mathcal D_\mu,
            \end{equation*}
            and consequently,
            \begin{equation}\label{eq:foc_final}
                G_{h}\leq (1-\lambda)\allones \qquad\mu\text{-almost everywhere}.
            \end{equation}
            Finally, \cref{eq:foc_equality} gives $\left\langle h,(1-\lambda)\allones-G_{h}\right\rangle_\mu=0$, while $(1-\lambda)\allones-G_{h}$ is non-negative on $\supp(\mu)$ by~\cref{eq:foc_final}. It must therefore vanish wherever $h$ is positive, proving equality on $\calS_h$.
        \end{proof}
    \subsection{Proof of~\cref{thm:monotonicity}}\label{appendix:proof_monotonicity}\noindent
        The proof of the three-order comparison in~\cref{lemma:order_insertion_monotonicity} relies on the following consequence of~\cref{lemma:first_order_optimality}, which relates three orders $\alpha$, $\beta$, and $\gamma$.
        \begin{lemma}\label{lemma:first_order_balance_identity}
            Consider a pair $(\mu,K)\in\calP(\X)\times\calP(\Y|\X)$ and let $1\leq\alpha<\beta\leq\infty$. Set $\gamma\triangleq1+\frac{\alpha}{\beta^\prime}$. Then, for any $0\leq \lambda<\eta_\gamma(\mu,K)$, there exists $h\in\mathcal D_\mu$ with support $\calS_h\triangleq\{x\in\supp(\mu): h(x)>0\}$ such that
            \begin{equation}\label{eq:first_order_balance_positivity}
                \left\|K_\mu^\star h\right\|_{L^\gamma(\mu K)}>\left\|h\right\|_{L^\gamma(\mu)}^\lambda,
            \end{equation}
            and
            \begin{equation}\label{eq:first_order_inequality_after_am_gm}
                \frac{\left\langle K_\mu^\star\bigl(h^{\frac{\alpha}{\beta}}\bigr),(K_\mu^\star h)^{\gamma-1}\right\rangle_{\mu K}}{\left\|K_\mu^\star h\right\|_{L^\gamma(\mu K)}^\gamma} \geq \left\| h^{\frac\alpha\beta}\right\|_{L^1(\mu)}^{1-\lambda} \left(\frac{\left\|h\right\|_{L^\alpha(\mu)}^{\alpha}}{\left\|h\right\|_{L^\gamma(\mu)}^{\gamma}}\right)^\lambda.
            \end{equation}
            When $\beta=\infty$, the function $h^\frac{\alpha}{\beta}$ is interpreted as $\mathbbm{1}_{\calS_h}$.
        \end{lemma}
        \begin{proof}
            Consider the functional
            \begin{equation*}
                \Phi(f)\triangleq \log\left\|K_\mu^\star f\right\|_{L^\gamma(\mu K)} - \lambda\log\left\|f\right\|_{L^\gamma(\mu)}
            \end{equation*}
            on $\mathcal{D}_\mu$.
            Finiteness of the alphabets makes $K_\mu^\star$ a continuous linear map, and hence both norms appearing in $\Phi$ are continuous maps. Moreover, as $f$ and $K_\mu^\star f$ are probability densities, the norms take values in $[1,\infty)$, so that composing them with the logarithm preserves continuity. Therefore, $\Phi$ is continuous on the compact set $\mathcal{D}_\mu$ and consequently admits a maximiser $h\in\mathcal D_\mu$.

            By the definition of $\eta_\gamma(\mu,K)$, the inequality $\lambda<\eta_\gamma(\mu,K)$ yields some $f\in\mathcal D_\mu\setminus\{\allones\}$ such that
            \begin{equation*}
                \frac{\log\left\|K_\mu^\star f\right\|_{L^\gamma(\mu K)}}{\log\left\|f\right\|_{L^\gamma(\mu)}}>\lambda.
            \end{equation*}
            The non-constant probability density $f$ satisfies $\|f\|_{L^\gamma(\mu)}>\|f\|_{L^1(\mu)}=1$ because $\gamma>1$ and therefore $\Phi(f)>0$. By maximality, $\Phi(h)\geq\Phi(f)>0$, which is equivalent to~\cref{eq:first_order_balance_positivity}.

            Since $h$ maximises $\Phi$ over $\mathcal D_\mu$, applying~\cref{lemma:first_order_optimality} at order $\gamma$ gives
            \begin{equation*}
                \frac{K\bigl((K_\mu^\star h)^{\gamma-1}\bigr)}{\left\|K_\mu^\star h\right\|_{L^\gamma(\mu K)}^\gamma} \leq \lambda \frac{h^{\gamma-1}}{\left\|h\right\|_{L^\gamma(\mu)}^\gamma} + (1-\lambda)\allones,
            \end{equation*}
            $\mu$-almost everywhere, with equality on $\calS_h$.

            The function $h^{\frac{\alpha}{\beta}}$ vanishes $\mu$-almost everywhere outside $\calS_h$: this is immediate when $\beta<\infty$, while for $\beta=\infty$ it follows from the convention $h^{\frac{\alpha}{\beta}}=\mathbbm{1}_{\calS_h}$. Moreover, $\frac{\alpha}{\beta}+\gamma-1=\alpha$, so that $h^\frac\alpha\beta h^{\gamma-1}=h^\alpha$ $\mu$-almost everywhere. We may therefore multiply the equality in the first-order condition by $h^{\frac\alpha\beta}$ and integrate with respect to $\mu$. Doing so and using~\cref{eq:adjoint_relation} leads to
            \begin{equation}\label{eq:first_order_equality}
                \frac{\left\langle K_\mu^\star\bigl(h^{\frac{\alpha}{\beta}}\bigr),(K_\mu^\star h)^{\gamma-1}\right\rangle_{\mu K}}{\left\|K_\mu^\star h\right\|_{L^\gamma(\mu K)}^\gamma}=(1-\lambda)\left\| h^\frac{\alpha}{\beta}\right\|_{L^1(\mu)}+\lambda\frac{\left\|h\right\|_{L^\alpha(\mu)}^\alpha}{\left\|h\right\|_{L^\gamma(\mu)}^\gamma}.
            \end{equation}
            Applying the weighted AM--GM inequality to the right-hand side of~\cref{eq:first_order_equality} gives~\cref{eq:first_order_inequality_after_am_gm}.
        \end{proof}
        The choice of $\gamma$ is explained by the identities
        \begin{equation*}
            \frac{\alpha}{\beta}+\gamma-1=\alpha,\qquad\frac{\alpha}{\beta}\,\beta=\alpha,\qquad(\gamma-1)\beta^\prime=\alpha,
        \end{equation*}
        which bring all quantities back to order $\alpha$. The first turns the product $h^{\alpha/\beta}h^{\gamma-1}$ into $h^\alpha$ in the preceding proof, while the other two turn the conjugate $\beta$- and $\beta^\prime$-norms arising from H\"older's inequality in the subsequent proof into $\alpha$-norms.
        \begin{lemma}\label{lemma:order_insertion_monotonicity}
            Consider a pair $(\mu,K)\in\calP(\X)\times\calP(\Y|\X)$ and let $1<\beta\leq\infty$. Then,
            \begin{equation*}
                \eta_{1+\frac1{\beta^\prime}}(\mu,K) \leq \eta_\beta(\mu,K).
            \end{equation*}
            Moreover, for every $1<\alpha<\beta$, setting $\gamma \triangleq 1+\frac{\alpha}{\beta^\prime}$,
            \begin{equation*}
                \eta_{\gamma}(\mu,K) \leq \max\big\{\eta_\alpha(\mu,K), \eta_\beta(\mu,K)\big\}.
            \end{equation*}
        \end{lemma}
        \begin{proof}
            We prove both assertions simultaneously, taking $\alpha=1$ for the first and $1<\alpha<\beta$ for the second.

            Suppose, towards a contradiction, that $\eta_\gamma(\mu,K)>\eta_\beta(\mu,K)$ and, when $\alpha>1$, also $\eta_\gamma(\mu,K)>\eta_\alpha(\mu,K)$. The DPI gives $\eta_\gamma(\mu,K)\leq1$, so the preceding strict inequalities allow us to choose $\lambda\in(0,1)$ such that $\eta_\beta(\mu,K) < \lambda< \eta_\gamma(\mu,K)$, and, when $\alpha>1$, also $\eta_\alpha(\mu,K)<\lambda$.

            By~\cref{lemma:first_order_balance_identity}, there exists $h\in\mathcal{D}_\mu$ such that
            \begin{equation}\label{eq:first_order_balance_positivity2}
                \left\|K_\mu^\star h\right\|_{L^\gamma(\mu K)}>\left\|h\right\|_{L^\gamma(\mu)}^\lambda,
            \end{equation}
            and
            \begin{equation}\label{eq:first_order_balance_identity2}
                \frac{\left\langle K_\mu^\star\bigl(h^{\frac{\alpha}{\beta}}\bigr),(K_\mu^\star h)^{\gamma-1}\right\rangle_{\mu K}}{\left\|K_\mu^\star h\right\|_{L^\gamma(\mu K)}^\gamma} \geq \left\| h^{\frac\alpha\beta}\right\|_{L^1(\mu)}^{1-\lambda} \frac{\left\|h\right\|_{L^\alpha(\mu)}^{\alpha\lambda}}{\left\|h\right\|_{L^\gamma(\mu)}^{\gamma\lambda}}.
            \end{equation}
            For $\alpha>1$, the numerator on the left-hand side of~\cref{eq:first_order_balance_identity2} is bounded above by
            \begin{align}
                \left\langle K_\mu^\star\bigl(h^{\frac\alpha\beta}\bigr), (K_\mu^\star h)^{\gamma-1}\right\rangle_{\mu K} &\leq \left\|K_\mu^\star\bigl(h^{\frac\alpha\beta}\bigr)\right\|_{L^\beta(\mu K)} \left\|(K_\mu^\star h)^{\gamma-1}\right\|_{L^{\beta^\prime}(\mu K)}&\left(\substack{\text{H\"older's}\\\text{inequality}}\right)&\nonumber\\
                &= \left\|K_\mu^\star\bigl(h^{\frac\alpha\beta}\bigr)\right\|_{L^\beta(\mu K)}\left\|K_\mu^\star h\right\|_{L^\alpha(\mu K)}^{\gamma-1}&\left(\substack{\text{Since}\\(\gamma-1)\beta^\prime=\alpha}\right)&\nonumber\\
                &\leq\left\|h^{\frac\alpha\beta}\right\|_{L^\beta(\mu)}^{\lambda}\left\| h^{\frac\alpha\beta}\right\|_{L^1(\mu)}^{1-\lambda}\left\|K_\mu^\star h\right\|_{L^\alpha(\mu K)}^{\gamma-1}&\left(\substack{\text{By~\cref{eq:renyi_sdpi_functional_inequality_non_negative_functions} and}\\\eta_\beta(\mu, K)<\lambda}\right)&\nonumber\\
                &= \left\|h\right\|_{L^\alpha(\mu)}^{\frac{\lambda\alpha}{\beta}}\left\| h^{\frac\alpha\beta}\right\|_{L^1(\mu)}^{1-\lambda}\left\|K_\mu^\star h\right\|_{L^\alpha(\mu K)}^{\gamma-1}\nonumber\\
                &\leq \left\|h\right\|_{L^\alpha(\mu)}^{\frac{\lambda\alpha}{\beta}}\left\| h^{\frac\alpha\beta}\right\|_{L^1(\mu)}^{1-\lambda}\left\| h\right\|_{L^\alpha(\mu)}^{(\gamma-1)\lambda}&\left(\substack{\text{By~\cref{eq:renyi_sdpi_functional_inequality_non_negative_functions} and}\\\eta_\alpha(\mu, K)<\lambda}\right)&\label{eq:order_insertion_holder_bound_temp}\\
                &= \left\| h^{\frac\alpha\beta}\right\|_{L^1(\mu)}^{1-\lambda}\left\| h\right\|_{L^\alpha(\mu)}^{\lambda\alpha}&\left(\substack{\text{Since}\\\frac\alpha\beta+\gamma-1=\alpha}\right)&\label{eq:holder_bound_monotonicity}
            \end{align}
            For $\alpha=1$, the same holds but~\cref{eq:order_insertion_holder_bound_temp} becomes an equality since $\left\|K_\mu^\star h\right\|_{L^1(\mu K)}=\left\|h\right\|_{L^1(\mu)}=1=\left\|h\right\|_{L^1(\mu)}^\lambda$.

            Substituting~\cref{eq:holder_bound_monotonicity} into~\cref{eq:first_order_balance_identity2} and rearranging yields
            \begin{equation*}
                \left(\frac{\left\|K_\mu^\star h\right\|_{L^\gamma(\mu K)}}{\left\|h\right\|_{L^\gamma(\mu)}^\lambda}\right)^\gamma \leq 1,
            \end{equation*}
            which contradicts~\cref{eq:first_order_balance_positivity2}.
        \end{proof}
        \begin{proof}[Proof of~\cref{thm:monotonicity}]
            If $\supp(\mu)$ is a singleton, then $\eta_\alpha(\mu,K)=0$ for every $\alpha\in[1,\infty]$ by the convention following~\cref{eq:dd_sdpi_constant}, and the result is immediate. Hence, assume that $|\supp(\mu)|\geq2$.

            The case $\beta=1$ is also immediate, so fix $\beta\in(1,\infty]$. Rather than comparing each lower order with $\beta$ separately, consider the set of orders at which the desired comparison fails:
            \begin{equation*}
                \mathcal{O}_\beta \triangleq \big\{\alpha\in(1,\beta): \eta_\alpha(\mu,K)>\eta_\beta(\mu,K)\big\}.
            \end{equation*}
            The set $\mathcal{O}_\beta$ is open. Indeed, for each fixed $f\in\mathcal D_\mu\setminus\{\allones\}$, the map
            \begin{equation*}
                \Psi_f(\alpha) \triangleq \frac{\log\left\|K_\mu^\star f\right\|_{L^\alpha(\mu K)}}{\log\left\|f\right\|_{L^\alpha(\mu)}}
            \end{equation*}
            is continuous on $(1,\infty)$. The representation of the SDPI constant as a supremum over densities gives
            \begin{equation*}
                \mathcal{O}_\beta = \bigcup_{f\in\mathcal D_\mu\setminus\{\allones\}} \big\{\alpha\in(1,\beta): \Psi_f(\alpha)>\eta_\beta(\mu, K)\big\},
            \end{equation*}
            and since each set in this union is open by continuity of the corresponding $\Psi_f$, $\mathcal{O}_\beta$ is open.

            Suppose, towards a contradiction, that $\mathcal{O}_\beta$ is non-empty. Since $\mathcal{O}_\beta$ is open, it is a disjoint union of open intervals, none of which can be enlarged while remaining in $\mathcal{O}_\beta$. Fix one such interval and denote it by $(a,b)$. Since $(a,b)\subseteq(1,\beta)$, its endpoints satisfy $1\leq a<b\leq\beta$. If $a>1$, then $a\notin\mathcal{O}_\beta$ for otherwise openness would allow the interval to be extended to the left, contradicting its construction. Consequently,
            \begin{equation}\label{eq:monotonicity_proof_ub_a}
                \eta_a(\mu, K)\leq\eta_\beta(\mu, K) \qquad\text{whenever }a>1.
            \end{equation}
            The same reasoning applies at $b$, yielding
            \begin{equation}\label{eq:monotonicity_proof_ub_b}
                \eta_b(\mu, K)\leq\eta_\beta(\mu, K)
            \end{equation}
            when $b<\beta$, while equality holds when $b=\beta$. Combining~\cref{eq:monotonicity_proof_ub_a,eq:monotonicity_proof_ub_b} with~\cref{lemma:order_insertion_monotonicity}, we obtain
            \begin{equation}\label{eq:order_insertion_bound}
                \eta_{1+\frac{a}{b^\prime}}(\mu,K) \leq
                \begin{cases}
                    \eta_b(\mu,K), &\text{if } a=1\\[1mm]
                    \max\bigl\{\eta_a(\mu, K),\eta_b(\mu, K)\bigr\}, & \text{if }a>1
                \end{cases}
                \leq \eta_\beta(\mu,K).
            \end{equation}
            However, the order appearing on the left-hand side lies strictly inside $(a,b)$. Indeed, $1+\frac{a}{b^\prime} = \frac1b b + \frac1{b^\prime}a\in(a,b)$ when $b<\infty$ and $1 + \frac{a}{b^\prime} = a + 1\in(a,b)$ when $b=\infty$. Since $(a,b)\subseteq\mathcal{O}_\beta$, it follows that
            \begin{equation*}
                \eta_{1+\frac{a}{b^\prime}}(\mu, K) > \eta_\beta(\mu, K),
            \end{equation*}
            contradicting~\cref{eq:order_insertion_bound}. Therefore, $\mathcal{O}_\beta$ must be empty, and hence
            \begin{equation}\label{eq:monotonicity_without_1}
                \eta_\alpha(\mu, K) \leq \eta_\beta(\mu, K), \qquad 1<\alpha\leq\beta.
            \end{equation}
            It remains to include the endpoint $\alpha=1$. Fix $\nu\ll\mu$ with $\nu\neq\mu$. For every $\alpha\in(1,\beta)$, the definition of the SDPI constant and~\cref{eq:monotonicity_without_1} give
            \begin{equation*}
                D_\alpha(\nu K\|\mu K) \leq \eta_\alpha(\mu,K)D_\alpha(\nu\|\mu) \leq \eta_\beta(\mu,K)D_\alpha(\nu\|\mu).
            \end{equation*}
            Letting $\alpha\downarrow1$ and using the convergence of R\'enyi Divergence to KL Divergence gives
            \begin{equation*}
                D_{\sf KL}(\nu K\|\mu K) \leq \eta_\beta(\mu,K)D_{\sf KL}(\nu\|\mu).
            \end{equation*}
            Dividing by $D_{\sf KL}(\nu\|\mu)>0$ and taking the supremum over $\nu\ll\mu$, $\nu\neq\mu$, yields
            \begin{equation*}
                \eta_1(\mu,K) = \eta_{\sf KL}(\mu,K) \leq \eta_\beta(\mu,K),
            \end{equation*}
            which completes the proof.
        \end{proof}
    \subsection{Proof of~\cref{thm:convexity}}\label{appendix:proof_convexity}\noindent
        To derive the desired convexity result, the first step is to refine~\cref{lemma:order_insertion_monotonicity}, which was used to establish~\cref{thm:monotonicity}.
        \begin{lemma}\label{lemma:order_insertion_convexity}
            Consider a pair $(\mu,K)\in\calP(\X)\times\calP(\Y|\X)$ and let $1<\alpha< \beta\leq\infty$. Then, setting $\gamma \triangleq 1+\frac{\alpha}{\beta^\prime}$,
            \begin{equation*}
                \eta_{\gamma}(\mu,K) \leq\frac{(\alpha-1)\eta_\alpha(\mu,K) + \eta_\beta(\mu,K)}{\alpha}.
            \end{equation*}
        \end{lemma}
        \begin{proof}
            Set $\lambda\triangleq\frac{(\alpha-1)\eta_\alpha(\mu,K)+\eta_\beta(\mu,K)}{\alpha}$, so that $\eta_\alpha(\mu,K)\leq\eta_\beta(\mu,K)$ by~\cref{thm:monotonicity}, and hence $\lambda\in[0,1]$.

            Suppose, towards a contradiction, that $\eta_\gamma(\mu,K)>\lambda$. By~\cref{lemma:first_order_balance_identity}, there exists $h\in\mathcal{D}_\mu$ such that
            \begin{equation}\label{eq:first_order_balance_positivity3}
                \left\|K_\mu^\star h\right\|_{L^\gamma(\mu K)}>\left\|h\right\|_{L^\gamma(\mu)}^\lambda,
            \end{equation}
            and
            \begin{equation}\label{eq:first_order_balance_identity3}
                \frac{\left\langle K_\mu^\star\bigl(h^{\frac{\alpha}{\beta}}\bigr),(K_\mu^\star h)^{\gamma-1}\right\rangle_{\mu K}}{\left\|K_\mu^\star h\right\|_{L^\gamma(\mu K)}^\gamma} \geq \left\| h^{\frac\alpha\beta}\right\|_{L^1(\mu)}^{1-\lambda} \frac{\left\|h\right\|_{L^\alpha(\mu)}^{\alpha\lambda}}{\left\|h\right\|_{L^\gamma(\mu)}^{\gamma\lambda}}.
            \end{equation}

            The numerator on the left-hand side of~\cref{eq:first_order_balance_identity3} is bounded above by
            \begin{align}
                \left\langle K_\mu^\star\bigl(h^{\frac\alpha\beta}\bigr), (K_\mu^\star h)^{\gamma-1}\right\rangle_{\mu K} &\leq \left\|K_\mu^\star\bigl(h^{\frac\alpha\beta}\bigr)\right\|_{L^\beta(\mu K)} \left\|(K_\mu^\star h)^{\gamma-1}\right\|_{L^{\beta^\prime}(\mu K)}&\left(\substack{\text{H\"older's}\\\text{inequality}}\right)&\nonumber\\
                &= \left\|K_\mu^\star\bigl(h^{\frac\alpha\beta}\bigr)\right\|_{L^\beta(\mu K)}\left\|K_\mu^\star h\right\|_{L^\alpha(\mu K)}^{\gamma-1}&\left(\substack{\text{Since}\\(\gamma-1)\beta^\prime=\alpha}\right)&\nonumber\\
                &\leq\left\|h^{\frac\alpha\beta}\right\|_{L^\beta(\mu)}^{\eta_\beta(\mu, K)}\left\| h^{\frac\alpha\beta}\right\|_{L^1(\mu)}^{1-\eta_\beta(\mu,K)}\left\| h\right\|_{L^\alpha(\mu)}^{(\gamma-1)\eta_\alpha(\mu,K)}&\left(\substack{\text{By~\cref{eq:renyi_sdpi_functional_inequality_non_negative_functions}}}\right)&\nonumber\\
                &\leq \left\| h^{\frac\alpha\beta}\right\|_{L^1(\mu)}^{1-\eta_\beta(\mu,K)} \left\|h\right\|_{L^\alpha(\mu)}^{\alpha\left(\frac{\eta_\beta(\mu,K)}{\beta} + \frac{\eta_\alpha(\mu,K)}{\beta^\prime}\right)}.&\left(\substack{\text{Since}\\(\gamma-1)\beta^\prime=\alpha}\right)&\label{eq:order_insertion_convexity_holder}
            \end{align}
            Substituting~\cref{eq:order_insertion_convexity_holder} into~\cref{eq:first_order_balance_identity3} and using the definition of $\lambda$ to simplify the exponents gives
            \begin{equation}\label{eq:order_insertion_convexity_final_inequality}
                \left(\frac{\left\|K_\mu^\star h\right\|_{L^\gamma(\mu K)}}{\left\|h\right\|_{L^\gamma(\mu)}^\lambda}\right)^\gamma \leq \left(\left\| h^{\frac\alpha\beta}\right\|_{L^1(\mu)}^{\alpha-1}\left\|h\right\|_{L^\alpha(\mu)}^{\alpha(1-\frac\alpha\beta)}\right)^{-\frac{\eta_\beta(\mu,K) - \eta_\alpha(\mu,K)}{\alpha}}.
            \end{equation}
            Finally, using H\"older's inequality with conjugate exponents $\frac{\gamma-1}{\alpha-1}$ and $\frac{\gamma-1}{\gamma-\alpha}$, we obtain
            \begin{align*}
                1 &= \left\langle h,\allones\right\rangle_\mu^{\gamma-1} \\
                &= \left\langle\left(h^{\frac\alpha\beta}\right)^{\frac{\alpha-1}{\gamma-1}},\left(h^\alpha\right)^{\frac{\gamma-\alpha}{\gamma-1}}\right\rangle_\mu^{\gamma-1}\\
                &\leq \left\langle h^{\frac\alpha\beta}, \allones\right\rangle_\mu^{\alpha-1} \left\langle h^\alpha, \allones\right\rangle_\mu^{\gamma-\alpha}\\
                &= \left\| h^{\frac\alpha\beta}\right\|_{L^1(\mu)}^{\alpha-1}\left\|h\right\|_{L^\alpha(\mu)}^{\alpha(1-\frac\alpha\beta)}.
            \end{align*}
            Since $\eta_\alpha(\mu,K)\leq\eta_\beta(\mu,K)$, the right-hand side of~\cref{eq:order_insertion_convexity_final_inequality} is therefore at most one, contradicting~\cref{eq:first_order_balance_positivity3}.

            When $\beta=\infty$, we have $\gamma=\alpha+1$. The same steps apply until~\cref{eq:order_insertion_convexity_final_inequality}, which with our convention $h^{\frac\alpha\infty}=\mathbbm{1}_{\calS_h}$ becomes
            \begin{equation*}
                \left(\frac{\left\|K_\mu^\star h\right\|_{L^\gamma(\mu K)}}{\left\|h\right\|_{L^\gamma(\mu)}^\lambda}\right)^\gamma \leq \left(\mu(\calS_h)^{\alpha-1}\left\|h\right\|_{L^\alpha(\mu)}^\alpha\right)^{-\frac{\eta_\infty(\mu, K)-\eta_\alpha(\mu, K)}{\alpha}} \leq 1,
            \end{equation*}
            where the last inequality follows from $\eta_\alpha(\mu, K)\leq\eta_\infty(\mu, K)$ and $1 = \left\langle h,\mathbbm{1}_{\calS_h}\right\rangle_\mu^\alpha \leq \mu(\calS_h)^{\alpha-1} \left\|h\right\|_{L^\alpha(\mu)}^\alpha$. This again contradicts~\cref{eq:first_order_balance_positivity3}, completing the proof.
        \end{proof}
        We now show that this refined upper bound suffices to establish~\cref{thm:convexity}.
        \begin{proof}[Proof of~\cref{thm:convexity}]
            Set
            \begin{equation*}
                F(\theta)\triangleq(\theta-1)\eta_\theta(\mu,K),\qquad \theta\in(1,\infty).
            \end{equation*}
            Fix $1<\alpha<\beta<\infty$, and let
            \begin{equation*}
                \ell(\theta)\triangleq\frac{\beta-\theta}{\beta-\alpha}F(\alpha)+\frac{\theta-\alpha}{\beta-\alpha}F(\beta)
            \end{equation*}
            denote the chord joining $F(\alpha)$ and $F(\beta)$. Consider
            \begin{equation*}
                \mathcal O_{\alpha,\beta}\triangleq\bigl\{\theta\in(\alpha,\beta):F(\theta)>\ell(\theta)\bigr\}.
            \end{equation*}
            By~\cref{prop:renyi_sdpi_functional_form},
            \begin{equation*}
                \mathcal O_{\alpha,\beta}=\bigcup_{f\in\mathcal D_\mu\setminus\{\allones\}}\left\{\theta\in(\alpha,\beta):(\theta-1)\frac{\log\left\|K_\mu^\star f\right\|_{L^\theta(\mu K)}}{\log\left\|f\right\|_{L^\theta(\mu)}}>\ell(\theta)\right\}.
            \end{equation*}
            Each set in this union is open, since its defining function is continuous in $\theta$. Hence, $\mathcal O_{\alpha,\beta}$ is open.

            Suppose, towards a contradiction, that $\mathcal{O}_{\alpha, \beta}$ is non-empty. Since $\mathcal{O}_{\alpha, \beta}$ is open, it is a disjoint union of open intervals, none of which can be enlarged while remaining in $\mathcal{O}_{\alpha, \beta}$. Fix one such interval and denote it by $(a,b)$. By maximality of this interval,
            \begin{equation*}
                F(a)\leq\ell(a),\qquad F(b)\leq\ell(b),
            \end{equation*}
            with equality whenever $a=\alpha$ or $b=\beta$. Now set $\gamma\triangleq1+\frac{a}{b^\prime}\in(a,b)$. Applying~\cref{lemma:order_insertion_convexity} with $a$ and $b$ in place of $\alpha$ and $\beta$, and multiplying its conclusion by $\gamma-1=\frac{a}{b^\prime}$, gives
            \begin{align*}
                F(\gamma) &\leq\frac{b-1}{b}F(a) + \frac1bF(b)\\
                &\leq \frac{b-1}{b}\ell(a) + \frac1b\ell(b)\\
                &= \ell(\gamma),
            \end{align*}
            where the last equality follows from the affinity of $\ell$. This contradicts $\gamma\in(a,b)\subseteq\mathcal O_{\alpha,\beta}$. Therefore, $F$ lies below the chord joining $F(\alpha)$ and $F(\beta)$. Since $\alpha$ and $\beta$ were arbitrary, $F$ is convex on $(1,\infty)$.
        \end{proof}
    \subsection{Proof of~\cref{corollary:continuity_in_order}}\label{appendix:proof_continuity_in_order}\noindent
        If $\supp(\mu)$ is a singleton, then $\eta_\alpha(\mu,K)=0$ for every $\alpha\in[1,\infty]$ by the convention following~\cref{eq:dd_sdpi_constant}, and the result is immediate. Hence, assume that $|\supp(\mu)|\geq2$.

        By~\cref{thm:convexity}, the function $\alpha\mapsto(\alpha-1)\eta_\alpha(\mu,K)$ is convex on $(1,\infty)$. Since $0\leq\eta_\alpha(\mu,K)\leq1$, this function is real-valued and hence continuous. Dividing by $\alpha-1$ proves the continuity of $\alpha\mapsto\eta_\alpha(\mu,K)$ on $(1,\infty)$.

        We next consider the endpoint at one. Fix $\nu\ll\mu$ with $\nu\neq\mu$ and write $f\triangleq\frac{\dd\nu}{\dd\mu}$. Letting $m_\mu \triangleq \min_{x\in\supp(\mu)}\mu(x)$, finiteness of the alphabets gives
        \begin{equation*}
            0\leq f\leq m_\mu^{-1}\quad\mu\text{-a.e.},
            \qquad
            0\leq K_\mu^\star f\leq\|f\|_{L^\infty(\mu)}
            \leq m_\mu^{-1}\quad\mu K\text{-a.e.},
        \end{equation*}
        where the second bound follows because $K_\mu^\star f$ is a conditional average of $f$.

        For $\alpha>1$, define the functions $g_\alpha(t)\triangleq\frac{t^\alpha-\alpha t+\alpha-1}{\alpha-1}$ and $g_1(t)\triangleq t\log t-t+1$ with the convention $0\log0=0$. Both functions and their first derivatives vanish at $t=1$. Moreover, for every $t\in\big(0,m_\mu^{-1}\big]$,
        \begin{equation*}
            g_\alpha^{\prime\prime}(t)=\alpha t^{\alpha-2}\leq\alpha m_\mu^{1-\alpha}t^{-1}=\alpha m_\mu^{1-\alpha}g_1^{\prime\prime}(t).
        \end{equation*}
        Consequently, the function $t\mapsto\alpha m_\mu^{1-\alpha}g_1(t)-g_\alpha(t)$ is convex and has value and first derivative equal to zero at $t=1$. It is therefore non-negative, and continuity at $t=0$ gives
        \begin{equation}\label{eq:continuity_at_one_g_inequality}
            g_\alpha(t)\leq\alpha m_\mu^{1-\alpha}g_1(t)
        \end{equation}
        throughout the range of both $f$ and $K_\mu^\star f$. Hence, we obtain
        \begin{align*}
            D_\alpha(\nu K\|\mu K) &\leq \calH_\alpha(\nu K\|\mu K)&\left(\substack{\text{Since}\\\frac1{\alpha-1}\log\big(1+(\alpha-1)x\big)\leq x}\right)&\\
            &= \mean{\mu K}{g_\alpha \circ \dualK f}&\left(\substack{\text{Definition of $g_\alpha$}}\right)&\\
            &\leq \alpha m_\mu^{1-\alpha}\mean{\mu K}{g_1 \circ \dualK f}&\left(\substack{\text{By~\cref{eq:continuity_at_one_g_inequality}}}\right)&\\
            &=\alpha m_\mu^{1-\alpha}D_{\sf KL}(\nu K\|\mu K)&\left(\substack{\text{Definition of $g_1$}}\right)&\\
            &\leq\alpha m_\mu^{1-\alpha}\eta_1(\mu,K)D_{\sf KL}(\nu\|\mu)&\left(\substack{\text{Definition of $\eta_1(\mu,K)$}}\right)&\\
            &\leq\alpha m_\mu^{1-\alpha}\eta_1(\mu,K)D_\alpha(\nu\|\mu)&\left(\substack{\text{Monotonicity of}\\\text{R\'enyi Divergences}}\right)&.
        \end{align*}
        Dividing by $D_\alpha(\nu\|\mu)>0$, taking the supremum over $\nu$ and applying~\cref{thm:monotonicity} for the reverse inequality yields
        \begin{equation}\label{eq:continuity_at_one_sdpi_comparison}
            \eta_1(\mu,K)\leq\eta_\alpha(\mu,K)\leq\alpha\left(\min_{x\in\supp(\mu)}\mu(x)\right)^{1-\alpha}\eta_1(\mu,K).
        \end{equation}
        Since the multiplicative factor converges to one as $\alpha\downarrow1$, continuity at one follows.

        Finally,~\cref{thm:monotonicity} shows that the finite-order R\'enyi-SDPI constants form a non-decreasing family bounded above by $\eta_\infty(\mu,K)$. Consequently, their limit exists and satisfies
        \begin{equation*}
            \lim_{\alpha\to\infty}\eta_\alpha(\mu,K)\leq\eta_\infty(\mu,K).
        \end{equation*}
        Conversely, fix $\nu\ll\mu$ with $\nu\neq\mu$. Since $D_\infty(\nu\|\mu)>0$, convergence of R\'enyi Divergences for both $(\nu,\mu)$ and $(\nu K,\mu K)$ to their order-$\infty$ values~\cite[Theorem 6]{van2014renyi} gives
        \begin{equation*}
            \lim_{\alpha\to\infty}\eta_\alpha(\mu,K)\geq\lim_{\alpha\to\infty}\frac{D_\alpha(\nu K\|\mu K)}{D_\alpha(\nu\|\mu)}=\frac{D_\infty(\nu K\|\mu K)}{D_\infty(\nu\|\mu)}.
        \end{equation*}
        Taking the supremum over admissible $\nu$ gives the reverse inequality, and therefore continuity at infinity holds.
    \subsection{Proof of~\cref{corollary:distribution_independent_order_properties}}\label{appendix:proof_distribution_independent_order_properties}\noindent
        By definition, $\eta_\alpha(K)=\sup_{\mu\in\calP(\X)}\eta_\alpha(\mu,K)$. Hence, Part a) follows from~\cref{thm:monotonicity}, since a pointwise supremum of non-decreasing functions is non-decreasing. Similarly, given that $(\alpha-1)\eta_\alpha(K)=\sup_{\mu\in\calP(\X)}(\alpha-1)\eta_\alpha(\mu,K)$ for every $\alpha>1$, Part b) follows from~\cref{thm:convexity}, since a pointwise supremum of convex functions is convex.

        For Part c), since $\alpha\mapsto(\alpha-1)\eta_\alpha(K)$ is convex by Part b) and real-valued, it is continuous on $(1,\infty)$. Dividing by $\alpha-1$ proves continuity of $\alpha\mapsto\eta_\alpha(K)$ on this interval. Finally, continuity at infinity follows from~\cref{corollary:continuity_in_order} and Part a) since
        \begin{equation*}
            \eta_\infty(K)=\sup_{\mu\in\calP(\X)}\lim_{\alpha\to\infty}\eta_\alpha(\mu,K)\leq\lim_{\alpha\to\infty}\sup_{\mu\in\calP(\X)}\eta_\alpha(\mu,K)=\lim_{\alpha\to\infty}\eta_\alpha(K)\leq\eta_\infty(K).
        \end{equation*}
    \subsection{Proof of~\cref{thm:extremal_functions}}\label{appendix:proof_extremal_functions}\noindent
        We first establish the dichotomy. By~\cref{thm:local_chi_squared_lb_on_renyi_sdpi}, $\eta_\alpha(\mu,K)\geq\eta_{\chi^2}(\mu,K)$. If equality holds, then Part~a) follows. Suppose therefore that $\eta_\alpha(\mu, K)>\eta_{\chi^2}(\mu, K)$.

        Recall from~\cref{prop:renyi_sdpi_functional_form} that the SDPI constant can be written as
        \begin{equation*}
            \eta_\alpha(\mu,K) = \sup_{f\in\mathcal D_\mu\setminus\{\allones\}}\frac{\log\| K_\mu^\star f\|_{L^\alpha(\mu K)}}{\log\| f\|_{L^\alpha(\mu)}}.
        \end{equation*}
        Since $\X$ is finite, $\mathcal D_\mu$ is compact and convex in $L^1(\mu)$.

        Let $(f_n)\subseteq\mathcal D_\mu\setminus\{\allones\}$ be a maximising sequence. By compactness of $\mathcal D_\mu$, after passing to a subsequence we have $f_n\to f^\star$ in $L^1(\mu)$ for some $f^\star\in\mathcal D_\mu$. However, this does not yet ensure that $f^\star$ belongs to $\mathcal D_\mu\setminus\{\allones\}$, since it could be equal to $\allones$. We show that this cannot happen. Suppose, for contradiction, that $f^\star=\allones$. Set $h_n=f_n-\allones$. Then,  since convergence of $f_n$ to $f^\star$ also holds in $L^2(\mu)$ due to finiteness of $|\X|$, $\| h_n\|_{L^2(\mu)}^2\to0$, $\langle h_n,\allones\rangle_\mu=0$, and one has the following Taylor expansions:
            \begin{align*}
                \log\| f_n\|_{L^\alpha(\mu)}
                &=\frac{\alpha-1}{2}\| h_n\|_{L^2(\mu)}^2 + o\left(\| h_n\|_{L^2(\mu)}^2\right),\\
                \log\| K_\mu^\star f_n\|_{L^\alpha(\mu K)}
                &=\frac{\alpha-1}{2}\| K_\mu^\star h_n\|_{L^2(\mu K)}^2 + o\left(\| h_n\|_{L^2(\mu)}^2\right).
            \end{align*}
            Therefore,
            \begin{equation*}
                \frac{\log\| K_\mu^\star f_n\|_{L^\alpha(\mu K)}}{\log\| f_n\|_{L^\alpha(\mu)}}
                =\frac{\| K_\mu^\star h_n\|_{L^2(\mu K)}^2}{\| h_n\|_{L^2(\mu)}^2} + o(1),
            \end{equation*}
            and since $(f_n)$ is a maximising sequence, taking $n\to\infty$ and using the definition of $\eta_{\chi^2}(\mu,K)$ gives
            \begin{equation*}
                \eta_\alpha(\mu, K) = \lim_{n\to\infty}\frac{\log\| K_\mu^\star f_n\|_{L^\alpha(\mu K)}}{\log\| f_n\|_{L^\alpha(\mu)}}=\lim_{n\to\infty}\frac{\| K_\mu^\star h_n\|_{L^2(\mu K)}^2}{\| h_n\|_{L^2(\mu)}^2}\leq\eta_{\chi^2}(\mu,K).
            \end{equation*}
        which would contradict $\eta_\alpha(\mu,K)>\eta_{\chi^2}(\mu,K)$. Hence $f^\star\in\mathcal D_\mu\setminus\{\allones\}$, and since the contraction ratio is continuous on this punctured domain, $f^\star$ attains $\eta_\alpha(\mu,K)$. From these properties, we deduce that the corresponding probability measure $\nu^\star$ given through the identity $f^\star=\frac{\dd\nu^\star}{\dd\mu}$ satisfies $\nu^\star\ll\mu$ and $\nu^\star\neq\mu$, and attains $\eta_\alpha(\mu, K)$.

        Next, we show that $\eta_\alpha(\mu,K)<1$. Suppose, towards a contradiction, that $\eta_\alpha(\mu,K)=1$. Since $f^\star$ attains the supremum,~\cref{prop:renyi_sdpi_functional_form} gives $\left\|K_\mu^\star f^\star\right\|_{L^\alpha(\mu K)}=\left\|f^\star\right\|_{L^\alpha(\mu)}$. On the other hand, Jensen's inequality and~\cref{eq:adjoint_relation} yield
        \begin{equation*}
            0\leq\left\langle K_\mu^\star\bigl((f^\star)^\alpha\bigr)-\bigl(K_\mu^\star f^\star\bigr)^\alpha,\allones\right\rangle_{\mu K}=\left\|f^\star\right\|_{L^\alpha(\mu)}^\alpha-\left\|K_\mu^\star f^\star\right\|_{L^\alpha(\mu K)}^\alpha=0.
        \end{equation*}
        The Jensen gap therefore vanishes $\mu K$-almost everywhere. Since $t\mapsto t^\alpha$ is strictly convex, $f^\star$ is constant on the support of $K_\mu^\star(\cdot\mid y)$ for $\mu K$-almost every $y$, and hence
        \begin{equation*}
            \bigl(K_\mu^\star(f^\star-\allones)\bigr)^2=K_\mu^\star\bigl((f^\star-\allones)^2\bigr)\qquad\mu K\text{-almost everywhere}.
        \end{equation*}
        Using this identity together with~\cref{eq:adjoint_relation}, we obtain
        \begin{equation*}
            1\geq\eta_{\chi^2}(\mu,K)\geq\frac{\left\|K_\mu^\star(f^\star-\allones)\right\|_{L^2(\mu K)}^2}{\left\|f^\star-\allones\right\|_{L^2(\mu)}^2}=\frac{\Big\langle\bigl(K_\mu^\star(f^\star-\allones)\bigr)^2,\allones\Big\rangle_{\mu K}}{\Big\langle K_\mu^\star\bigl((f^\star-\allones)^2\bigr),\allones\Big\rangle_{\mu K}}=1,
        \end{equation*}
        contradicting $\eta_\alpha(\mu,K)>\eta_{\chi^2}(\mu,K)$. Therefore, $\eta_{\chi^2}(\mu,K)<\eta_\alpha(\mu,K)<1$, which proves Part~b).

        It remains to establish~\cref{eq:renyi_sdpi_foc}. Let $\nu^\star$ now be an arbitrary extremal distribution for $\eta_\alpha(\mu,K)$, and set $f^\star\triangleq\frac{\dd\nu^\star}{\dd\mu}\in\mathcal D_\mu\setminus\{\allones\}$. Consider the functional
        \begin{equation*}
            \Phi(f)\triangleq\log\left\|K_\mu^\star f\right\|_{L^\alpha(\mu K)}-\eta_\alpha(\mu,K)\log\left\|f\right\|_{L^\alpha(\mu)}.
        \end{equation*}
        For every $f\in\mathcal D_\mu\setminus\{\allones\}$, strictness in the monotonicity of $L^p$-norms gives $\|f\|_{L^\alpha(\mu)}>\|f\|_{L^1(\mu)}=1$. Hence, the definition of $\eta_\alpha(\mu,K)$ implies $\Phi(f)\leq0$. Moreover, $\Phi(\allones)=0$ and $\Phi(f^\star)=0$, so that $f^\star$ maximises $\Phi$ over $\mathcal D_\mu$. Applying~\cref{lemma:first_order_optimality} with $\lambda=\eta_\alpha(\mu,K)$ and $h=f^\star$ gives~\cref{eq:renyi_sdpi_foc}, with equality on $\supp(\nu^\star)$.
    \subsection{Proof of~\cref{thm:boundary_nu}}\label{appendix:proof_boundary_nu}\noindent
        If $\supp(\mu)$ is a singleton, then both sides of~\cref{eq:boundary_nu_maximum_representation} are zero by convention, and the second assertion is vacuous. Hence, assume that $|\supp(\mu)|\geq2$.

        We first prove the second assertion. Suppose that $0<\eta_\alpha(\mu,K)<1$, and let $\nu^\star$ be an arbitrary extremal distribution for $\eta_\alpha(\mu,K)$, with corresponding density $f^\star\triangleq\frac{\dd\nu^\star}{\dd\mu}\in\mathcal D_\mu\setminus\{\allones\}$. Suppose, towards a contradiction, that $\supp(\nu^\star)=\supp(\mu)$. Equality then holds $\mu$-almost everywhere in~\cref{eq:renyi_sdpi_foc}. Integrating it with respect to $\mu$ and using~\cref{eq:adjoint_relation} yields
        \begin{equation}\label{eq:foc_identity_boundary_theorem}
            \frac{\left\|K_\mu^\star f^\star\right\|_{L^{\alpha-1}(\mu K)}^{\alpha-1}}{\left\|K_\mu^\star f^\star\right\|_{L^\alpha(\mu K)}^\alpha} = \eta_\alpha(\mu, K)\frac{\left\|f^\star\right\|_{L^{\alpha-1}(\mu)}^{\alpha-1}}{\left\|f^\star\right\|_{L^\alpha(\mu)}^\alpha} + 1-\eta_\alpha(\mu, K).
        \end{equation}
        Since $f^\star$ is a non-constant density, we have
        \begin{equation*}
            \left\|f^\star\right\|_{L^\alpha(\mu)}^\alpha-\left\|f^\star\right\|_{L^{\alpha-1}(\mu)}^{\alpha-1} = \left\langle(f^\star)^{\alpha-1}-\allones, f^\star-\allones\right\rangle_\mu > 0,
        \end{equation*}
        so that on the right-hand side of~\cref{eq:foc_identity_boundary_theorem}, the ratio of norms lies in the interval $(0,1)$. Moreover, the left-hand side of~\cref{eq:foc_identity_boundary_theorem} can be bounded as
        \begin{align}
            \frac{\left\|K_\mu^\star f^\star\right\|_{L^{\alpha-1}(\mu K)}^{\alpha-1}}{\left\|K_\mu^\star f^\star\right\|_{L^\alpha(\mu K)}^\alpha} &= \frac{\left\|K_\mu^\star f^\star\right\|_{L^{\alpha-1}(\mu K)}^{\alpha-1}}{\left\|f^\star\right\|_{L^\alpha(\mu)}^{\alpha\eta_\alpha(\mu, K)}}&\left(\substack{\text{Since $f^\star$}\\\text{attains $\eta_\alpha(\mu, K)$}}\right)&\nonumber\\
            &\leq \frac{\left\|f^\star\right\|_{L^{\alpha-1}(\mu)}^{(\alpha-1)\eta_{\alpha-1}(\mu, K)}}{\left\|f^\star\right\|_{L^\alpha(\mu)}^{\alpha\eta_\alpha(\mu, K)}}&\left(\substack{\text{By~\cref{eq:renyi_sdpi_functional_inequality} for $\alpha>2$;}\\\text{By~\cref{eq:markov_kernel_mean_preservation} for $\alpha=2$}}\right)&\nonumber\\
            &\leq \left(\frac{\left\|f^\star\right\|_{L^{\alpha-1}(\mu)}^{\alpha-1}}{\left\|f^\star\right\|_{L^\alpha(\mu)}^\alpha}\right)^{\eta_\alpha(\mu, K)}&\left(\substack{\text{By~\cref{thm:monotonicity} and since}\\\text{$\|f^\star\|_{L^{\alpha-1}(\mu)}\geq\|f^\star\|_{L^{1}(\mu)}=1$}}\right)&\nonumber\\
            &< \eta_\alpha(\mu, K)\frac{\left\|f^\star\right\|_{L^{\alpha-1}(\mu)}^{\alpha-1}}{\left\|f^\star\right\|_{L^\alpha(\mu)}^\alpha} + 1-\eta_\alpha(\mu, K),&\left(\substack{\text{Weighted AM--GM}\\\text{inequality}}\right)&\label{eq:boundary_result_final_inequality}
        \end{align}
        where strictness in the last inequality is due to $0<\eta_\alpha(\mu, K)<1$ and the fact that the ratio of norms lies in $(0,1)$. Therefore, \cref{eq:boundary_result_final_inequality} contradicts~\cref{eq:foc_identity_boundary_theorem} and we conclude that $\supp(\nu^\star)\subsetneq\supp(\mu)$.

        It remains to prove~\cref{eq:boundary_nu_maximum_representation}, for which it suffices to exhibit an extremal distribution whose support is strictly contained in $\supp(\mu)$. If $\eta_\alpha(\mu,K)=0$, then any $\delta_x$ with $x\in\supp(\mu)$ is an extremal distribution whose support is strictly contained in $\supp(\mu)$.

        Suppose now that $\eta_\alpha(\mu,K)>0$. Recall that $\eta_{\chi^2}(\mu,K)$ is the supremum of $\frac{\|K_\mu^\star h\|_{L^2(\mu K)}^2}{\|h\|_{L^2(\mu)}^2}$ over non-zero functions $h$ satisfying $\mean{\mu}{h}=0$~\cite[Theorem III.2]{raginsky2016strong}. Since this quotient is invariant under rescaling, its optimisation may be restricted to $\mu$-mean-zero functions with unit $L^2(\mu)$-norm. This set is compact because $\X$ is finite, and hence the supremum is attained by some function $h^\star$. Since $h^\star$ is non-zero and has zero $\mu$-mean, it takes both positive and negative values on $\supp(\mu)$, and we may normalise it so that $\min_{x\in\supp(\mu)}h^\star(x)=-1$. In particular, $\allones+h^\star$ is a non-constant probability density that vanishes somewhere on $\supp(\mu)$.

        Suppose first that $\eta_\alpha(\mu,K)<1$. If $\eta_{\chi^2}(\mu,K)=0$, then $\eta_{\chi^2}(\mu,K)<\eta_\alpha(\mu,K)$ immediately. Otherwise,~\cref{thm:local_chi_squared_lb_on_renyi_sdpi} gives $0<\eta_{\chi^2}(\mu,K)\leq\eta_\alpha(\mu,K)<1$, and we obtain
        \begin{align*}
            \eta_\alpha(\mu,K)&\geq\eta_2(\mu,K)&\left(\substack{\text{By~\cref{thm:monotonicity}}}\right)&\\
            &\geq\frac{\log\left(1+\eta_{\chi^2}(\mu,K)\|h^\star\|_{L^2(\mu)}^2\right)}{\log\left(1+\|h^\star\|_{L^2(\mu)}^2\right)}&\left(\substack{\text{By~\cref{prop:renyi_sdpi_functional_form}}\\\text{applied to $\allones+h^\star$}}\right)&\\
            &>\eta_{\chi^2}(\mu,K).&\left(\substack{\text{Strict concavity of}\\\text{the logarithm}}\right)&
        \end{align*}
        Thus, in either case,~\cref{thm:extremal_functions} guarantees the existence of an extremal distribution, whose support is strictly contained in $\supp(\mu)$ by the preceding argument.

        Finally, suppose that $\eta_\alpha(\mu,K)=1$. By~\cref{thm:extremal_functions}, we then have $\eta_{\chi^2}(\mu,K)=1$, and Jensen's inequality gives
        \begin{equation*}
            0\leq \left\langle K_\mu^\star\big((h^\star)^2\big)-\bigl(K_\mu^\star h^\star\bigr)^2,\allones\right\rangle_{\mu K} =\|h^\star\|_{L^2(\mu)}^2-\left\|K_\mu^\star h^\star\right\|_{L^2(\mu K)}^2=0.
        \end{equation*}
        Equality therefore holds in the corresponding application of Jensen's inequality, so $h^\star$ is constant on the support of $K_\mu^\star(\cdot|y)$ for $\mu K$-almost every $y$. Consequently,
        \begin{equation*}
            \left\|K_\mu^\star(\allones+h^\star)\right\|_{L^\alpha(\mu K)}^\alpha=\left\langle \bigl(K_\mu^\star(\allones+h^\star)\bigr)^\alpha,\allones\right\rangle_{\mu K}=\left\langle K_\mu^\star\bigl((\allones+h^\star)^\alpha\bigr),\allones\right\rangle_{\mu K}=\|\allones+h^\star\|_{L^\alpha(\mu)}^\alpha.
        \end{equation*}
        The probability measure with density $\allones+h^\star$ is therefore an extremal distribution whose support is strictly contained in $\supp(\mu)$, proving~\cref{eq:boundary_nu_maximum_representation}.
    \subsection{Proof of~\cref{prop:infty_sdpi_achieved_on_conditionals}}\label{appendix:proof_infty_sdpi_achieved_on_conditionals}\noindent
        If $\supp(\mu)$ is a singleton, then $\eta_\infty(\mu,K)=0$ and there is no $A\in\Sigma_\X$ satisfying $0<\mu(A)<1$, so the first assertion is immediate and the second is vacuous. Hence, assume that $|\supp(\mu)|\geq2$.

        We will show that both $\eta\leq \eta_\infty(\mu, K)$ and $\eta\geq \eta_\infty(\mu, K)$ hold, where
        \begin{equation*}
            \eta\triangleq\sup_{\substack{A\in\Sigma_\X:\\0<\mu(A)<1}}\frac{D_\infty\left(\mu_{|A} K\|\mu K\right)}{D_\infty\left(\mu_{|A}\|\mu\right)} = \sup_{\substack{A\in\Sigma_\X:\\0<\mu(A)<1}}\frac{D_\infty\left(\mu_{|A} K\|\mu K\right)}{-\log{\mu(A)}}.
        \end{equation*}
        Since $\X$ is finite, there are only finitely many events $A\in\Sigma_\X$ satisfying $0<\mu(A)<1$, so the supremum defining $\eta$ is in fact a maximum.

        We start with the lower bound. Since for any $A\in\Sigma_\X$ with $0<\mu(A)<1$, the probability measure $\mu_{|A}$ is admissible in the optimisation defining $\eta_\infty(\mu, K)$, we obtain $\eta \leq \eta_\infty(\mu, K)$. Moreover when $\eta=1$, the equality is settled since $\eta_\infty(\mu, K)\leq 1$. Hence, assume from now on that $\eta\in[0, 1)$.

        For the upper bound, we introduce the probability measure $\rho_B$, defined for a fixed $B\in\Sigma_\Y$ with $\mu K(B)>0$ by
        \begin{equation*}
            \frac{\dd\rho_B}{\dd\mu} = \frac{K(B|\cdot)}{\mu K(B)} \quad \mu\text{-a.e.}
        \end{equation*}
        For all $A\in\Sigma_\X$ with $0<\mu(A)<1$, observe the following chain of equivalences due to the definition of $\eta$:
        \begin{align}
        D_\infty(\mu_{|A}K\|\mu K)\leq -\eta\log{\mu(A)}&\iff \left\|\frac{\dd \mu_{|A}K}{\dd \mu K}\right\|_{L^\infty(\mu K)} \leq \mu(A)^{-\eta}\nonumber\\
        &\iff \forall B\in\Sigma_\Y, \quad\mu_{|A}K(B) \leq \mu(A)^{-\eta}\mu K(B)\nonumber\\
        &\iff \forall B\in\Sigma_\Y\text{ with }\mu K(B)>0, \quad \rho_B(A) \leq \mu(A)^{1-\eta}.\label{eq:eta_infty_proof_inequality_for_events}
        \end{align}
        This remains true when $\mu(A)\in\{0, 1\}$ since $\rho_B\ll\mu$, and hence~\cref{eq:eta_infty_proof_inequality_for_events} holds in fact for every $A\in\Sigma_\X$. Next, we extend~\cref{eq:eta_infty_proof_inequality_for_events}, which currently holds for events $A\in\Sigma_\X$, to non-negative bounded measurable functions.

        Consider a non-negative measurable function $f\in L^\infty(\mu)$. Defining its superlevel sets as $A_t\triangleq \{x\in\X: f(x)>t\}$, the layer-cake representation gives that for any probability measure $\rho\in\calP(\X)$ satisfying $\rho\ll\mu$,
        \begin{align}
            \int_\X f(x)\dd\rho(x)&=\int_\X\left(\int_0^{\|f\|_{L^\infty(\mu)}}\mathbbm{1}_{A_t}(x)\dd t\right)\dd\rho(x)\nonumber\\
            &= \int_0^{\|f\|_{L^\infty(\mu)}} \left(\int_\X\mathbbm{1}_{A_t}(x)\dd\rho(x)\right)\dd t\nonumber\\
            &=\int_0^{\|f\|_{L^\infty(\mu)}}\rho(A_t)\dd t.\label{eq:proof_eta_infty_layer_cake}
        \end{align}
        Consequently, for any $B\in\Sigma_\Y$ with $\mu K(B)>0$, we obtain
        \begin{align}
            \mean{\rho_B}{f} & =\int_0^{\|f\|_{L^\infty(\mu)}} \rho_B(A_t)\dd t&\left(\substack{\text{By~\cref{eq:proof_eta_infty_layer_cake}}}\right)&\nonumber\\
            &\leq \int_0^{\|f\|_{L^\infty(\mu)}} \mu(A_t)^{1-\eta}\dd t&\left(\substack{\text{By~\cref{eq:eta_infty_proof_inequality_for_events}}}\right)&\nonumber\\
            &\leq \left(\int_0^{\|f\|_{L^\infty(\mu)}}\mu(A_t)\dd t\right)^{1-\eta}\left(\int_0^{\|f\|_{L^\infty(\mu)}}1\dd t\right)^\eta&\left(\substack{\text{H\"older's}\\\text{ inequality}}\right)&\nonumber\\
            &= \left(\int_0^{\|f\|_{L^\infty(\mu)}}\mu(A_t)\dd t\right)^{1-\eta}\|f\|_{L^\infty(\mu)}^\eta\nonumber\\
            &= \left(\int f\dd\mu\right)^{1-\eta}\|f\|_{L^\infty(\mu)}^\eta.&\left(\substack{\text{By~\cref{eq:proof_eta_infty_layer_cake}}}\right)&\label{eq:eta_infty_proof_inequality_for_functions}
        \end{align}
        The step with H\"older's inequality holds for $0<\eta<1$ using the H\"older's conjugates $\frac1{1-\eta}$ and $\frac1\eta$. When $\eta=0$, the same inequality holds with equality.

        Now, let $\nu\in\calP(\X)$ with $0<D_\infty(\nu\|\mu)<\infty$, and set $f=\frac{\dd\nu}{\dd\mu}$. Applying~\cref{eq:eta_infty_proof_inequality_for_functions} yields
        \begin{equation*}
            \frac{\nu K(B)}{\mu K(B)}= \int_\X \frac{K(B|x)}{\mu K(B)}\frac{\dd\nu}{\dd\mu}(x)\dd\mu(x)=\mean{\rho_B}{\frac{\dd\nu}{\dd\mu}}\leq \left\|\frac{\dd\nu}{\dd\mu}\right\|_{L^\infty(\mu)}^\eta
        \end{equation*}
        for every $B\in\Sigma_\Y$ with $\mu K(B)>0$, or equivalently $D_\infty(\nu K\|\mu K)\leq \eta D_\infty(\nu\|\mu)$. Rearranging and taking supremum over all admissible $\nu$ gives $\eta_\infty(\mu, K)\leq \eta$.

        Suppose now that $0<\eta_\infty(\mu,K)<1$, so that $0<\eta<1$, and let $\nu^\star$ be an arbitrary extremal distribution with density $f^\star\triangleq\frac{\dd\nu^\star}{\dd\mu}$. Since $\Y$ is finite and $\nu^\star$ is extremal, there exists $y\in\supp(\mu K)$ such that $\mean{\rho_{\{y\}}}{f^\star}=K_\mu^\star f^\star(y)=\|f^\star\|_{L^\infty(\mu)}^\eta$. Since $\mean{\mu}{f^\star}=1$, both sides of~\cref{eq:eta_infty_proof_inequality_for_functions}, applied with $B=\{y\}$ and $f=f^\star$, are equal. Equality must therefore hold in the corresponding application of H\"older's inequality. Writing $A_t^\star\triangleq\{x\in\X:f^\star(x)>t\}$, this implies that $t\mapsto\mu(A_t^\star)$ is constant almost everywhere on $(0,\|f^\star\|_{L^\infty(\mu)})$. As $\X$ is finite, $f^\star$ can consequently take only the values $0$ and $\|f^\star\|_{L^\infty(\mu)}$ $\mu$-almost everywhere. Moreover, since $f^\star$ is measurable, $A_0^\star\in\Sigma_\X$, and $f^\star=\|f^\star\|_{L^\infty(\mu)}\mathbbm{1}_{A_0^\star}$ $\mu$-almost everywhere. Finally, $\mean{\mu}{f^\star}=1$ gives $\|f^\star\|_{L^\infty(\mu)}=\frac1{\mu(A_0^\star)}$. Since $\nu^\star\neq\mu$, this norm is strictly larger than one, so $0<\mu(A_0^\star)<1$ and $\nu^\star=\mu_{|A_0^\star}$.
    \subsection{Proof of~\cref{prop:eta_parameter_renyi_binary}}\label{appendix:proof_eta_parameter_renyi_binary}\noindent
        For any divergence $D$, let $\eta_D^{(2)}(K)$ denote the supremum in~\cref{eq:di_sdpi_constant} restricted to pairs of probability measures supported on at most two common points.

        We first prove Part~a).

        \textbf{Case $\alpha\in\{0,1\}$:} The result follows from~\cite[Theorem~1]{ordentlich2021strong}, since $\eta_0(K)=\eta_{\sf revKL}(K)=\eta_{\sf KL}(K)=\eta_1(K)$.

        \textbf{Case $\alpha\in(0,1)$:} Applying~\cref{thm:local_chi_squared_lb_on_renyi_sdpi} to binary-supported pairs of probability measures, together with~\cite[Theorem~2]{ordentlich2021strong} and~\cref{prop:distribution_independent_sdpi_equal_chi_square_sdpi}, gives
        \begin{equation*} \eta_\alpha(K)\geq\eta_\alpha^{(2)}(K)\geq\eta_{\chi^2}^{(2)}(K)=\eta_{\chi^2}(K)=\eta_\alpha(K). \end{equation*}

        \textbf{Case $\alpha\in(1,\infty)$:} Fix any $(\nu,\mu)$ with $0<D_\alpha(\nu\|\mu)<\infty$, and set $f\triangleq\frac{\dd\nu}{\dd\mu}$. Let $H_\alpha(\nu\|\mu)\triangleq e^{(\alpha-1)D_\alpha(\nu\|\mu)}$ denote the Hellinger integral. Consider the compact convex set
        \begin{equation*}
            \mathcal{S}_f \triangleq \left\{\widehat\mu\in\calP\bigl(\supp(\mu)\bigr):\left\langle f,\allones\right\rangle_{\widehat\mu}=1\right\}.
        \end{equation*}
        Since $\mathcal{S}_f$ is the intersection of a probability simplex with a hyperplane, its extreme points are supported on at most two points in $\supp(\mu)$. Moreover, as $\mu\in\mathcal S_f$, there exist $1\leq m\leq|\supp(\mu)|$, extreme points $\mu^{(1)},\ldots,\mu^{(m)}$ of $\mathcal S_f$, and coefficients $\theta_1,\ldots,\theta_m>0$ summing to one such that $\mu=\sum_{j=1}^m\theta_j\mu^{(j)}$. Also, let $\nu^{(j)}$ be the probability measure defined through
        \begin{equation}\label{eq:binary_nu_definition}
            \frac{\dd\nu^{(j)}}{\dd\mu^{(j)}}=f\quad\mu^{(j)}\text{-a.e.},
        \end{equation}
        so that $\nu=\sum_{j=1}^m\theta_j\nu^{(j)}$. Since $\nu\neq\mu$, at least one component satisfies $\nu^{(j)}\neq\mu^{(j)}$. For every such $j$, the pair $(\nu^{(j)},\mu^{(j)})$ is admissible, and the definition of $\eta_\alpha^{(2)}(K)$ gives
        \begin{equation}\label{eq:binary_support_non_trivial_component}
            H_\alpha\bigl(\nu^{(j)}K\|\mu^{(j)}K\bigr)\leq H_\alpha\bigl(\nu^{(j)}\|\mu^{(j)}\bigr)^{\eta_\alpha^{(2)}(K)}.
        \end{equation}
        For the remaining indices, both sides of~\cref{eq:binary_support_non_trivial_component} equal one. Hence, \cref{eq:binary_support_non_trivial_component} holds for every $j\in[m]$ and we obtain
        \begin{align*}
            H_\alpha(\nu K\|\mu K) &\leq\sum_{j=1}^m\theta_j H_\alpha\bigl(\nu^{(j)}K\|\mu^{(j)}K\bigr)&\left(\substack{\text{Joint convexity}\\\text{of $H_\alpha$}}\right)&\\
            &\leq\sum_{j=1}^m\theta_j H_\alpha\bigl(\nu^{(j)}\|\mu^{(j)}\bigr)^{\eta_\alpha^{(2)}(K)}&\left(\substack{\text{By~\cref{eq:binary_support_non_trivial_component}}}\right)&\\
            &\leq\left(\sum_{j=1}^m\theta_j H_\alpha\bigl(\nu^{(j)}\|\mu^{(j)}\bigr)\right)^{\eta_\alpha^{(2)}(K)}&\left(\substack{\text{Jensen's}\\\text{inequality}}\right)&\\
            &=\left(\sum_{j=1}^m\theta_j\langle f^\alpha, \allones\rangle_{\mu^{(j)}}\right)^{\eta_\alpha^{(2)}(K)}\\
            &=\langle f^\alpha, \allones\rangle_{\sum_{j=1}^m\theta_j\mu^{(j)}}^{\eta_\alpha^{(2)}(K)}&\left(\substack{\text{By~\cref{eq:binary_nu_definition}}}\right)&\\
            &=H_\alpha(\nu\|\mu)^{\eta_\alpha^{(2)}(K)}.
        \end{align*}
        Taking logarithms and rearranging gives $\frac{D_\alpha(\nu K\|\mu K)}{D_\alpha(\nu\|\mu)}\leq{\eta_\alpha^{(2)}(K)}$, and taking the supremum over admissible $(\nu,\mu)$ yields $\eta_\alpha(K)\leq\eta_\alpha^{(2)}(K)$. The reverse inequality is immediate.

        \textbf{Case $\alpha=\infty$:} For every $\beta\in(1,\infty)$, $0<D_\beta(\nu\|\mu)<\infty$ implies $\nu\ll\mu$, and therefore
        \begin{equation*}
            \eta_\beta^{(2)}(K)=\sup_{\substack{\mu\in\calP(\X):\\|\supp(\mu)|\leq2}}\eta_\beta(\mu,K)\leq\sup_{\substack{\mu\in\calP(\X):\\|\supp(\mu)|\leq2}}\eta_\infty(\mu,K)=\eta_\infty^{(2)}(K),
        \end{equation*}
        where the inequality follows from~\cref{thm:monotonicity}. With this, Part c) in~\cref{corollary:distribution_independent_order_properties} and the finite-order cases proven above lead to
        \begin{equation*}
            \eta_\infty(K)=\lim_{\beta\to\infty}\eta_\beta(K)=\lim_{\beta\to\infty}\eta_\beta^{(2)}(K)\leq\eta_\infty^{(2)}(K)\leq\eta_\infty(K).
        \end{equation*}

        We now prove Part~b). Fix $\alpha\in[2,\infty]$. Since admissibility implies $\nu\ll\mu$, Part~a) reduces the optimisation to reference measures $\mu$ supported on at most two points. We may further restrict to $|\supp(\mu)|=2$: if $|\supp(\mu)|=1$, then the constraint set defining $\eta_\alpha(\mu,K)$ is empty and $\eta_\alpha(\mu,K)=0$ by convention. Such reference measures therefore do not affect the supremum.

        For $\alpha\in[2,\infty)$,~\cref{thm:boundary_nu} restricts the optimisation defining $\eta_\alpha(\mu,K)$ to probability measures whose support is strictly contained in $\supp(\mu)$. At order $\alpha=\infty$,~\cref{prop:infty_sdpi_achieved_on_conditionals} restricts it to conditional distributions of $\mu$ on non-trivial events. Since $\mu$ is supported on two points, either restriction forces $\nu$ to be a Dirac measure. Writing $x$ for its support point and $x^\prime$ for the other point in $\supp(\mu)$ gives
        \begin{equation*}
            \nu=\delta_x,\qquad\mu=(1-t)\delta_x+t\delta_{x^\prime}
        \end{equation*}
        for some $t\in(0,1)$, proving Part~b).
    \subsection{Proof of~\cref{prop:infty_renyi_sdpi_closed_form}}\label{appendix:proof_infty_renyi_sdpi_closed_form}\noindent
        By~\cref{prop:eta_parameter_renyi_binary} and~\cref{prop:infty_sdpi_achieved_on_conditionals}, it suffices to optimise over pairs of the form $\nu = \delta_x$ and $\mu_t = (1-t)\delta_x + t\delta_{x^\prime}$ for distinct $x, x^\prime\in\X$, where $t\in(0,1)$. The corresponding ratio of $\infty$-R\'enyi Divergences is given by
        \begin{align*}
            \frac{D_\infty(\delta_x K\|\mu_t K)}{D_\infty(\delta_x\|\mu_t)} &= \frac{\log \left\|\frac{\dd K(\cdot|x)}{\dd\mu_t K}\right\|_{L^\infty(\mu_t K)}}{\log\frac1{1-t}}\\
            &= \frac{\log \max_{y\in\supp(\mu_t K)}\frac{K(y|x)}{(1-t)K(y|x) +t K(y|x^\prime)}}{\log\frac1{1-t}}\\
            &= \frac{\log \max_{y\in\supp(K(\cdot|x))}\frac{K(y|x)}{(1-t)K(y|x) +t K(y|x^\prime)}}{\log\frac1{1-t}}\\
            &= \frac{-\log\left(1-t + t \min_{y\in\supp(K(\cdot|x))}\frac{K(y|x^\prime)}{K(y|x)}\right)}{\log\frac1{1-t}}\\
            &= \frac{\log \bigg(1-t \left(1 -\min_{y\in\supp(K(\cdot|x))}\frac{K(y|x^\prime)}{K(y|x)}\right)\bigg)}{\log(1-t)}.
        \end{align*}
        Now that the ratio has been determined, it remains to find its maximum value over $t\in(0,1)$. Since $\lambda\triangleq 1 -\min_{y\in\supp(K(\cdot|x))}\frac{K(y|x^\prime)}{K(y|x)}\in[0,1]$, we can use~\cref{corollary:log_ratio_decreasing} which states that $t \mapsto \frac{\log(1-\lambda t)}{\log(1-t)}$ is decreasing on $(0,1)$ and $\lim_{t\to0}\frac{\log(1-\lambda t)}{\log(1-t)}=\lambda$. Therefore, the supremum is approached when $t\to 0$, and so
        \begin{equation*}
            \sup_{t\in(0,1)} \frac{D_\infty(\delta_x K\|\mu_t K)}{D_\infty(\delta_x\|\mu_t)} = \lim_{t\to 0} \frac{D_\infty(\delta_x K\|\mu_t K)}{D_\infty(\delta_x\|\mu_t)}= 1 -\min_{y\in\supp(K(\cdot|x))} \frac{K(y|x^\prime)}{K(y|x)}.
        \end{equation*}
        Since $x, x^\prime\in\X$ were arbitrary, the final expression for the SDPI constant becomes
        \begin{equation*}
            \eta_\infty(K) = \sup_{x, x^\prime\in\X}\left\{1 -\min_{y\in\supp(K(\cdot|x))} \frac{K(y|x^\prime)}{K(y|x)}\right\}.
        \end{equation*}
    \subsection{Proof of~\cref{prop:distribution_independent_sdpi_equal_chi_square_sdpi}}\label{appendix:proof_distribution_independent_sdpi_equal_chi_square_sdpi}\noindent
        The cases ${\alpha\in\{0,1\}}$ reduce to the distribution-independent SDPI constant of KL Divergence, an instance of $\varphi$-Divergence with operator convex $\varphi$, for which it is known that $\eta_{\varphi}(K)=\eta_{\chi^2}(K)$~\cite[Corollary III.1]{raginsky2016strong}.

        For $\alpha\in(0,1)$, \cref{thm:local_chi_squared_lb_on_renyi_sdpi,corollary:sdpi_ub_lb_hellinger} yield $\eta_{\chi^2}(K)
        \leq \eta_\alpha(K)
        \leq \eta_{\calH_\alpha}(K)
        = \eta_{\chi^2}(K)$, where the last equality follows again from \cite[Corollary III.1]{raginsky2016strong}, since the Hellinger Divergence is generated by an operator-convex function when $\alpha\in(0,1)$. Hence $\eta_\alpha(K)=\eta_{\chi^2}(K)$.
    \subsection{Proof of~\cref{prop:non_trivial_sdpi_characterisation}}\label{appendix:proof_non_trivial_sdpi_characterisation}\noindent
        We start with Part a). Since the boundary cases $\alpha\in\{0,1\}$ are already known~\cite{cohen1993relative}, let us consider $\alpha\in(0,1)$.
        \begin{itemize}
            \item \textbf{``If'' direction $\Longleftarrow$:} Suppose that there exist $x, x^\prime$ such that $\supp\big(K(\cdot|x)\big)$ and $\supp\big(K(\cdot|x^\prime)\big)$ are disjoint. Then~\cite[Theorem 4.1]{cohen1993relative} gives $\eta_{\chi^2}(K)=1$, and since $\eta_\alpha(K)\geq\eta_{\chi^2}(K)$ this means $\eta_\alpha(K)=1$.
            \item \textbf{``Only if'' direction $\Longrightarrow$:} We prove the contrapositive. Assume that for all $x, x^\prime\in\X$, $\supp\big(K(\cdot|x)\big)$ and $\supp\big(K(\cdot|x^\prime)\big)$ intersect. But then using that $\|\nu-\mu\|_{\sf TV}<1\iff \supp(\nu)\cap\supp(\mu)\neq\emptyset$ together with the characterisation ${\eta_{\sf TV}(K) = \sup_{x, x^\prime\in\X}\|K(\cdot|x)-K(\cdot|x^\prime)\|_{\sf TV}}$ gives $\eta_{\sf TV}(K)<1$. It remains to invoke Part a) of~\cref{prop:bound_on_distribution_independent_sdpi} and we get $\eta_\alpha(K)\leq\eta_{\sf TV}(K)<1$.
        \end{itemize}
        Let us now proceed with Part b) and analyse the case $\alpha>1$.
        \begin{itemize}
            \item \textbf{``If'' direction $\Longleftarrow$:} Suppose that there exist two rows in $K$ with different support. In other words, there exists $x, x^\prime\in\X$ and $y\in\Y$ such that $K(y|x^\prime)=0$ but $K(y|x)>0$. Letting $\nu = \delta_x$ and $\mu_\varepsilon = (1-\varepsilon)\delta_{x^\prime}+\varepsilon\delta_x$ where $\varepsilon\in(0,1)$, we find that $\renyiDiv{\nu}{\mu_\varepsilon}=-\log\varepsilon$ and
            \begin{align*}
                e^{(\alpha-1)\renyiDiv{\nu K}{\mu_\varepsilon K}} &= \sum_{y^\prime\in\Y}K (y^\prime|x)^\alpha\big((1-\varepsilon)K(y^\prime|x^\prime) +\varepsilon K(y^\prime|x)\big)^{1-\alpha}\\
                &\geq K (y|x)^\alpha\big((1-\varepsilon)\cdot 0 +\varepsilon K(y|x)\big)^{1-\alpha}\\
                &= \varepsilon^{1-\alpha}K(y|x).
            \end{align*}
            Consequently for all $\varepsilon\in(0,1)$,
            \begin{equation*}
                \eta_\alpha(K) \geq \frac{\renyiDiv{\nu K}{\mu_\varepsilon K}}{\renyiDiv{\nu}{\mu_\varepsilon}} \geq 1 + \frac{\log K(y|x)}{(\alpha-1)\log\frac1\varepsilon} \overset{\varepsilon\to 0}{\longrightarrow} 1,
            \end{equation*}
            and since $\eta_\alpha(K)\leq 1$ always, $\eta_\alpha(K)=1$.
            \item \textbf{``Only if'' direction $\Longrightarrow$:} We prove the contrapositive. Assume that for all $x, x^\prime\in\X$, $\supp\big(K(\cdot|x)\big)=\supp\big(K(\cdot|x^\prime)\big)\equiv \calS$. Let us first consider the limiting regime in which $\renyiDiv{\nu}{\mu}\to\infty$. Since all rows of $K$ have common support, we have
            \begin{align*}
                \renyiDiv{\nu K}{\mu K}&\leq D_\infty(\nu K\|\mu K) \\
                &= \log\max_{y\in\supp(\mu K)}\frac{\nu K(y)}{\mu K (y)}\\
                &\leq \log\max_{y\in\calS}\frac{\max_{x\in\X}K(y|x)}{\min_{x^\prime\in\X}K(y|x^\prime)} \\
                &<\infty.
            \end{align*}
            Hence, we would obtain $\frac{\renyiDiv{\nu K}{\mu K}}{\renyiDiv{\nu}{\mu}} \to 0$, which is obviously strictly smaller than 1.

            A second limiting regime is when $\renyiDiv{\nu}{\mu}\to0$, which is known to yield the local contraction $\eta_{\chi^2}(K)$~\cite[Theorem 2]{jin2024properties}. Since we assume all conditionals of $K$ have identical support, they are in particular not disjoint, so that $\eta_{\chi^2}(K)<1$. With the limiting cases out of the way, assume we have $0<\renyiDiv{\nu}{\mu}<\infty$, and let $f=\rnd$. One obtains
            \begin{align*}
                e^{(\alpha-1)\renyiDiv{\nu K}{\mu K}} &= \left\langle\allones, (\dualK f)^\alpha\right\rangle_{\mu K}&\left(\substack{\text{By~\cref{eq:dual_kernel_formula}}}\right)&\\
                &< \left\langle\allones, \dualK( f^\alpha)\right\rangle_{\mu K}&\left(\substack{\text{Jensen's}\\\text{inequality}}\right)&\\
                &=\left\langle K(\allones),  f^\alpha\right\rangle_{\mu}&\left(\substack{\text{By~\cref{eq:adjoint_relation}}}\right)&\\
                &=\left\langle \allones,  f^\alpha\right\rangle_{\mu}\\
                &= e^{(\alpha-1)\renyiDiv{\nu}{\mu}},
            \end{align*}
            where Jensen's inequality is strict because $\nu\neq\mu$ implies that $f$ is non-constant on $\supp(\mu)$, and the assumption that all conditionals of $K$ have the same support ensures that $K_\mu^\star(\cdot|y)$ has support $\supp(\mu)$ for every $y\in\calS$.

            The two limiting arguments above rule out contraction ratios approaching one when $D_\alpha(\nu\|\mu)$ tends to zero or infinity. In the remaining regime, compactness of $\calP(\X)^2$ yields a limiting pair, at which the strict Jensen inequality above again prevents the contraction ratio from approaching one. Hence, $\eta_\alpha(K)<1$.

            Finally, let us prove Part b) for $\alpha=\infty$. In this case, we find from~\cref{prop:infty_renyi_sdpi_closed_form} that $\eta_\infty(K)=1$ is equivalent to the statement
            \begin{equation*}
                \exists x, x^\prime\in\X, \exists y\in\Y\text{ such that } K(y|x)>0\text{ and }K(y|x^\prime)=0.
            \end{equation*}
            In other terms, this is equivalent to the existence of $x, x^\prime\in\X$ such that $\supp\big(K(\cdot|x)\big)\neq\supp\big(K(\cdot|x^\prime)\big)$.
        \end{itemize}
    \subsection{Proof of~\cref{prop:renyi_samorodnitsky}}\label{appendix:proof_renyi_samorodnitsky}\noindent
        Before proceeding with the proof, we establish a useful lemma:
        \begin{lemma}\label{lemma:renyi_sdpi_side_information}
            Let $\alpha\in(1,\infty]$, $K\in\calP(\Y|\X)$, and $\nu,\gamma\in\calP(\Z\times\X)$ such that $\nu\ll\gamma$. Denote by $\nu_{\Z}$ and $\gamma_{\Z}$ their $\Z$-marginals. Then
            \begin{equation}\label{eq:renyi_sdpi_side_information}
                D_\alpha\big(\nu(\Id_{\Z}\otimes K)\|\gamma(\Id_{\Z}\otimes K)\big) \leq \eta_\alpha(K)D_\alpha(\nu\|\gamma) + \big(1-\eta_\alpha(K)\big) D_\alpha(\nu_{\Z}\|\gamma_{\Z}).
            \end{equation}
            If the reference measure takes the form $\gamma=\rho\otimes\mu$ for some $\rho\in\calP(\Z)$ and $\mu\in\calP(\X)$, then the same inequality holds but with $\eta_\alpha(\mu,K)$ replacing $\eta_\alpha(K)$.
        \end{lemma}

        \begin{proof}
            We have
            \begin{equation*}
                \nu_{\Z}(z)\triangleq\sum_{x\in\X}\nu(z,x),\qquad \gamma_{\Z}(z)\triangleq\sum_{x\in\X}\gamma(z,x),
            \end{equation*}
            and since $\nu\ll\gamma$, for every $z\in\supp(\nu_{\Z})$ we can define
            \begin{equation*}
                \nu(x|z)\triangleq\frac{\nu(z,x)}{\nu_{\Z}(z)},\qquad \gamma(x|z)\triangleq\frac{\gamma(z,x)}{\gamma_{\Z}(z)}.
            \end{equation*}
            Write $\eta$ for the applicable coefficient on the right-hand side of~\cref{eq:renyi_sdpi_side_information}. In the $\gamma=\rho\otimes\mu$ case, $\gamma(\cdot|z)=\mu$ for any $z\in\supp(\nu_\Z)$, so the SDPI holds with $\eta=\eta_\alpha(\mu,K)$. In general, it holds with $\eta=\eta_\alpha(K)$. The cases $\eta\in\{0,1\}$ follow directly, so assume $\eta\in(0,1)$.

            For $1<\alpha<\infty$, the definition of R\'enyi-SDPI and H\"older's inequality give
            \begin{align*}
                e^{(\alpha-1)D_\alpha\left(\nu(\Id_{\Z}\otimes K)\middle\|\gamma(\Id_{\Z}\otimes K)\right)}&=\sum_{z\in\supp(\nu_{\Z})}\gamma_{\Z}(z)\left(\frac{\nu_{\Z}(z)}{\gamma_{\Z}(z)}\right)^\alpha e^{(\alpha-1)D_\alpha\big(\nu(\cdot|z)K\|\gamma(\cdot|z)K\big)}\\
                &\leq\sum_{z\in\supp(\nu_{\Z})}\left[\gamma_{\Z}(z)\left(\frac{\nu_{\Z}(z)}{\gamma_{\Z}(z)}\right)^\alpha e^{(\alpha-1)D_\alpha\big(\nu(\cdot|z)\|\gamma(\cdot|z)\big)}\right]^\eta \left[\gamma_{\Z}(z)\left(\frac{\nu_{\Z}(z)}{\gamma_{\Z}(z)}\right)^\alpha\right]^{1-\eta}\\
                &\leq\left[\sum_{z\in\supp(\nu_{\Z})}\gamma_{\Z}(z)\left(\frac{\nu_{\Z}(z)}{\gamma_{\Z}(z)}\right)^\alpha e^{(\alpha-1)D_\alpha\left(\nu(\cdot|z)\middle\|\gamma(\cdot|z)\right)}\right]^\eta \left[\sum_{z\in\supp(\nu_{\Z})}\gamma_{\Z}(z)\left(\frac{\nu_{\Z}(z)}{\gamma_{\Z}(z)}\right)^\alpha\right]^{1-\eta}.
            \end{align*}
            Taking logarithms and dividing by $(\alpha-1)$ proves the claim for finite $\alpha$.

            For $\alpha=\infty$, only the definition of the R\'enyi-SDPI is needed, which leads to
            \begin{align*}
                e^{D_\infty\big(\nu(\Id_{\Z}\otimes K)\|\gamma(\Id_{\Z}\otimes K)\big)}&=\max_{z\in\supp(\nu_{\Z})}\frac{\nu_{\Z}(z)}{\gamma_{\Z}(z)}e^{D_\infty\big(\nu(\cdot|z)K\|\gamma(\cdot|z)K\big)}\\
                &\leq\left[\max_{z\in\supp(\nu_{\Z})}\frac{\nu_{\Z}(z)}{\gamma_{\Z}(z)}e^{D_\infty\big(\nu(\cdot|z)\|\gamma(\cdot|z)\big)}\right]^\eta\left[\max_{z\in\supp(\nu_{\Z})}\frac{\nu_{\Z}(z)}{\gamma_{\Z}(z)}\right]^{1-\eta}.
            \end{align*}
            Taking logarithms concludes the proof.
        \end{proof}

        \begin{proof}[Proof of~\cref{prop:renyi_samorodnitsky}]
            For brevity, write $\eta_i\triangleq\eta_\alpha(\mu_i,K_i)$ for $i\in[n]$. We start with the proof of Part a). Let $\rho\in\calP(\Z)$ and let $\nu\ll\rho\otimes\mu_1\otimes\cdots\otimes\mu_n$ be a probability measure on $\Z\times\X_1\times\cdots\times\X_n$. For $T\subseteq[n]$, let $\nu_{\Z,T}$ denote its marginal on $\Z\times\prod_{i\in T}\X_i$, and set $\nu_{\Z,\varnothing}\triangleq\nu_{\Z}$. Moreover, set $\mu_T\triangleq\bigotimes_{i\in T}\mu_i$, with the convention that $\mu_\varnothing$ denotes the (unique) probability measure on a singleton. We claim that
            \begin{equation}\label{eq:renyi_samorodnitsky_auxiliary}
                D_\alpha\big(\nu(\Id_{\Z}\otimes K_1\otimes\cdots\otimes K_n)\|\rho\otimes\mu_1K_1\otimes\cdots\otimes\mu_nK_n\big) \leq \sum_{T\subseteq[n]}\left(\prod_{i\in T}\eta_i\prod_{i\in[n]\setminus T}(1-\eta_i)\right)D_\alpha\left(\nu_{\Z,T}\middle\|\rho\otimes\mu_T\right),
            \end{equation}
            and prove it by induction on $n$.

            For $n=1$, this reduces to the version of~\cref{lemma:renyi_sdpi_side_information} with the reference measure being a product distribution. Suppose that the claim holds for $n-1$. Since the component kernels act on distinct coordinates, the overall Markov kernel may be applied in two stages:
            \begin{equation*}
            \begin{aligned}
                \nu(\Id_{\Z}\otimes K_1\otimes\cdots\otimes K_n)=\Bigl[\nu\bigl(\Id_{\Z\times\prod_{i=1}^{n-1}\X_i}\otimes K_n\bigr)\Bigr]\bigl(\Id_{\Z}\otimes K_1\otimes\cdots\otimes K_{n-1}\otimes\Id_{\Y_n}\bigr).
            \end{aligned}
            \end{equation*}
            After the first stage, which consists in only applying $K_n$ to the $n$-th coordinate of $\nu$, the resulting measure remains absolutely continuous with respect to
            \begin{equation*}
                \rho\otimes\mu_1\otimes\cdots\otimes\mu_{n-1}\otimes\mu_nK_n.
            \end{equation*}
            We may therefore apply the induction hypothesis to $K_1,\ldots,K_{n-1}$, treating $\Z\times\Y_n$ as the auxiliary space, with reference measure $\rho\otimes\mu_nK_n$. For every $T\subseteq[n-1]$, the marginal appearing on the right-hand side of the induction hypothesis is $\nu_{\Z,T\cup\{n\}}\bigl(\Id_{\Z\times\prod_{i\in T}\X_i}\otimes K_n\bigr)$, with reference measure $\rho\otimes\mu_T\otimes\mu_nK_n$. Applying the induction hypothesis, and then applying~\cref{lemma:renyi_sdpi_side_information} to each resulting term with auxiliary space $\Z\times\prod_{i\in T}\X_i$ and auxiliary reference $\rho\otimes\mu_T$, gives
            \begin{align*}
                D_\alpha\big(\nu(\Id_{\Z}&\otimes K_1\otimes\cdots\otimes K_n)\|\rho\otimes\mu_1K_1\otimes\cdots\otimes\mu_nK_n\big)\\
                &\leq\sum_{T\subseteq[n-1]}\left(\prod_{i\in T}\eta_i\prod_{i\in[n-1]\setminus T}(1-\eta_i)\right)D_\alpha\left(\nu_{\Z,T\cup\{n\}}(\Id\otimes K_n)\|\rho\otimes\mu_T\otimes\mu_nK_n\right)\\
                &\leq\sum_{T\subseteq[n-1]}\left(\prod_{i\in T}\eta_i\prod_{i\in[n-1]\setminus T}(1-\eta_i)\right)\!\!\bigg[\eta_nD_\alpha\left(\nu_{\Z,T\cup\{n\}}\|\rho\otimes\mu_T\otimes\mu_n\right)+(1-\eta_n)D_\alpha\left(\nu_{\Z,T}\|\rho\otimes\mu_T\right)\bigg]\\
                &=\sum_{S\subseteq[n]}\left(\prod_{i\in S}\eta_i\prod_{i\in[n]\setminus S}(1-\eta_i)\right)D_\alpha\left(\nu_{\Z,S}\|\rho\otimes\mu_S\right).
            \end{align*}
            In the first inequality, $\Id$ acts on $\Z$ and the coordinates in $T$. In the last equality, the terms multiplied by $\eta_n$ correspond to the subsets $S\subseteq[n]$ containing $n$, whereas those multiplied by $1-\eta_n$ correspond to the subsets not containing $n$. This proves~\cref{eq:renyi_samorodnitsky_auxiliary}.

            To conclude, it remains to take $\Z$ to be a singleton, $\rho$ its (unique) probability measure, and $\nu$ the probability measure in the proposition. Then $\nu_{\Z,T}=\nu_T$, and \cref{eq:renyi_samorodnitsky_auxiliary} becomes
            \begin{equation*}
                D_\alpha\big(\nu(K_1\otimes\cdots\otimes K_n)\|\mu_1K_1\otimes\cdots\otimes\mu_nK_n\big)\leq\sum_{T\subseteq[n]}\left(\prod_{i\in T}\eta_i\prod_{i\in[n]\setminus T}(1-\eta_i)\right)D_\alpha(\nu_T\|\mu_T),
            \end{equation*}
            as desired.

            Let us now prove Part b). For $\alpha\in(1,\infty]$, consider arbitrary auxiliary measures $\nu,\gamma\in\calP(\Z\times\X_1\times\cdots\times\X_n)$ with $\nu\ll\gamma$, and let $\nu_{\Z,T}$ and $\gamma_{\Z,T}$ denote their respective marginals on $\Z\times\prod_{i\in T}\X_i$. The same induction proves
            \begin{equation*}
                D_\alpha\big(\nu(\Id_{\Z}\otimes K_1\otimes\cdots\otimes K_n)\|\gamma(\Id_{\Z}\otimes K_1\otimes\cdots\otimes K_n)\big)\leq\sum_{T\subseteq[n]}\left(\prod_{i\in T}\eta_\alpha(K_i)\prod_{i\in[n]\setminus T}\bigl(1-\eta_\alpha(K_i)\bigr)\right)D_\alpha(\nu_{\Z,T}\|\gamma_{\Z,T}).
            \end{equation*}
            Indeed, at the induction step both measures are passed through $K_n$, and the general part of~\cref{lemma:renyi_sdpi_side_information} is applied with the conditional reference on the $n$-th coordinate. Taking $\Z$ to be a singleton gives the result.

            For $\alpha\in(0,1)$, consider arbitrary $\nu,\gamma\in\calP(\X_1\times\cdots\times\X_n)$. No absolute-continuity assumption is needed in this case, and the desired inequality follows instead from the Samorodnitsky-type inequality for $\varphi$-Divergences. Indeed, the Hellinger Divergence is generated by an operator convex function when $\alpha\in(0,1)$, so that by~\cite[Corollary III.1]{raginsky2016strong} and~\cref{prop:distribution_independent_sdpi_equal_chi_square_sdpi}
            \begin{equation*}
                \eta_{\calH_\alpha}(K_i)=\eta_{\chi^2}(K_i)=\eta_\alpha(K_i),\qquad i\in[n]
            \end{equation*}
            and by~\cite[Theorem 7]{makur2020comparison},
            \begin{equation}\label{eq:hellinger_samorodnitsky}
                \calH_\alpha\big(\nu(K_1\otimes\cdots\otimes K_n)\|\gamma(K_1\otimes\cdots\otimes K_n)\big)\leq \sum_{T\subseteq[n]}\left(\prod_{i\in T}\eta_{\calH_\alpha}(K_i)\prod_{i\in[n]\setminus T}\bigl(1-\eta_{\calH_\alpha}(K_i)\bigr)\right)\calH_\alpha(\nu_{T}\|\gamma_{T}).
            \end{equation}
            Writing $w_T$ for the coefficient of the $T$-th term in the preceding subset sum and using $\sum_{T\subseteq[n]}w_T=1$, we have
            \begin{align*}
                e^{(\alpha-1)D_\alpha\big(\nu( K_1\otimes\cdots\otimes K_n)\|\gamma(K_1\otimes\cdots\otimes K_n)\big)} &= 1+(\alpha-1)\calH_\alpha\big(\nu( K_1\otimes\cdots\otimes K_n)\|\gamma(K_1\otimes\cdots\otimes K_n)\big)\\
                &\geq 1+(\alpha-1)\sum_{T\subseteq[n]}w_T\calH_\alpha(\nu_{T}\|\gamma_{T})&\left(\substack{\text{By~\cref{eq:hellinger_samorodnitsky}}}\right)&\\
                &=\sum_{T\subseteq[n]}w_T\left(1+(\alpha-1)\calH_\alpha(\nu_{T}\|\gamma_{T})\right)\\
                &=\sum_{T\subseteq[n]}w_T\exp\left((\alpha-1)D_\alpha(\nu_{T}\|\gamma_{T})\right)\\
                &\geq\exp\left((\alpha-1)\sum_{T\subseteq[n]}w_TD_\alpha(\nu_{T}\|\gamma_{T})\right).&\left(\substack{\text{Jensen's}\\\text{inequality}}\right)&
            \end{align*}
            Taking logarithms and dividing by $\alpha-1<0$ concludes the proof.
        \end{proof}
    \subsection{Proof of~\cref{prop:tensorisation_of_sdpi_constant}}\label{appendix:proof_tensorisation_of_sdpi_constant}\noindent
        We start with a proof of Part a). Consider $\mu=\mu_1\otimes\cdots\otimes\mu_n$ and $K=K_1\otimes \cdots\otimes K_n$, where $(\mu_i, K_i)\in\calP(\X_i)\times\calP(\Y_i|\X_i)$.

        First, we establish the lower bounds. For any $i\in[n]$, we have
        \begin{align*}
            \eta_\alpha(\mu, K) &= \sup_{\substack{\nu\in\calP(\X_1\times\cdots\times \X_n):\\0<\renyiDiv{\nu}{\mu}<\infty}} \frac{D_\alpha\big(\nu (K_1\otimes\cdots\otimes K_n)\|\mu_1 K_1\otimes\cdots\otimes\mu_n K_n\big)}{D_\alpha(\nu\|\mu_1\otimes\cdots\otimes\mu_n)}\\
            &\geq \sup_{\substack{\nu_i\in\calP(\X_i):\\0<D_\alpha(\nu_i\|\mu_i)<\infty}}\frac{D_\alpha\big(\mu_1K_1\otimes\cdots\otimes\mu_{i-1}K_{i-1}\otimes\nu_iK_i\otimes\mu_{i+1}K_{i+1}\otimes\cdots\otimes\mu_n K_n\|\mu_1 K_1\otimes\cdots\otimes\mu_n K_n\big)}{D_\alpha(\mu_1\otimes\cdots\otimes\mu_{i-1}\otimes\nu_i\otimes\mu_{i+1}\otimes\cdots\otimes\mu_n\|\mu_1\otimes\cdots\otimes\mu_n)}\\
            &= \sup_{\substack{\nu_i\in\calP(\X_i):\\0<D_\alpha(\nu_i\|\mu_i)<\infty}}\frac{D_\alpha\big(\nu_i K_i\|\mu_i K_i) + \sum_{j\neq i}D_\alpha(\mu_j K_j\|\mu_j K_j)}{D_\alpha(\nu_i\|\mu_i) + \sum_{j\neq i}D_\alpha(\mu_j\|\mu_j)}\\
            &= \eta_\alpha(\mu_i, K_i).
        \end{align*}
        Since it holds for every $i$, we obtain $\eta_\alpha(\mu, K) \geq\max_{i\in[n]}\eta_\alpha(\mu_i, K_i)$. Moreover, this naturally leads to the bound for the distribution-independent SDPI constant:
        \begin{align*}
            \eta_\alpha(K) &= \sup_{\gamma\in\calP(\X_1\times\cdots\times\X_n)} \eta_\alpha(\gamma, K) \\
            &\geq\sup_{\gamma_1\otimes\cdots\otimes\gamma_n\in\calP(\X_1\times\cdots\times\X_n)} \eta_\alpha(\gamma_1\otimes\cdots\otimes\gamma_n, K_1\otimes\cdots\otimes K_n) \\
            &\geq \sup_{\gamma_1\otimes\cdots\otimes\gamma_n\in\calP(\X_1\times\cdots\times\X_n)} \max_{i\in[n]}\eta_\alpha(\gamma_i, K_i) \\
            &= \max_{i\in[n]}\eta_\alpha(K_i).
        \end{align*}
        Let us move on to the upper bounds, which follow directly from~\cref{prop:renyi_samorodnitsky}. Indeed, consider any $\nu\in\calP(\X_1\times\cdots\times\X_n)$ with $\nu\ll\mu$. Since marginalisation is a Markov kernel, the DPI yields
        \begin{equation*}
            D_\alpha(\nu_T\|\mu_T)\leq D_\alpha(\nu\|\mu),
        \end{equation*}
        for any $T\subseteq[n]$. Using this in~\cref{eq:dd_renyi_samorodnitsky_tensorisation} with the additional fact that $D_\alpha(\nu_\varnothing\|\mu_\varnothing)=0$, we obtain
        \begin{align*}
            D_\alpha(\nu K\|\mu K) &\leq \sum_{T\subseteq[n]}\left(\prod_{i\in T}\eta_\alpha(\mu_i, K_i)\prod_{j\in[n]\setminus T}\big(1-\eta_\alpha(\mu_j, K_j)\big)\right)D_\alpha(\nu_T\|\mu_T)\\
            &\leq \sum_{\substack{T\subseteq[n]:\\T\neq\varnothing}}\left(\prod_{i\in T}\eta_\alpha(\mu_i, K_i)\prod_{j\in[n]\setminus T}\big(1-\eta_\alpha(\mu_j, K_j)\big)\right)D_\alpha(\nu\|\mu)\\
            &= D_\alpha(\nu\|\mu)\left[1-\prod_{i=1}^n\bigl(1-\eta_\alpha(\mu_i, K_i)\bigr)\right]
        \end{align*}
        Dividing by $D_\alpha(\nu\|\mu)$ and taking supremum over all admissible $\nu$ gives $\eta_\alpha(\mu, K)\leq 1-\prod_{i=1}^n\bigl(1-\eta_\alpha(\mu_i, K_i)\bigr)$, as desired.

        For Part b), using the distribution-independent inequality from~\cref{prop:renyi_samorodnitsky} displayed in~\cref{eq:di_renyi_samorodnitsky_tensorisation}, the same reasoning yields the bound on $\eta_\alpha(K)$.

\section{Proofs for~\cref{sec:bounds_on_contraction_coefficient}}\noindent
    \subsection{Proof of~\cref{lemma:bound_between_renyi_chi_squared}}\label{appendix:proof_bound_between_renyi_chi_squared}\noindent
        We start with a weaker version of the result to illustrate the approach. By Taylor's Theorem with the Lagrange form of the remainder, for any $t\geq 0$, there exists some $c$ lying between $1$ and $t$ such that
        \begin{equation*}
            t^\alpha = 1 + \alpha\big(t - 1\big) + \frac{\alpha(\alpha-1)}{2}c^{\alpha-2}\big(t - 1\big)^2.
        \end{equation*}
        Selecting $t=\frac{\dd\nu}{\dd\mu}(x)$ and integrating w.r.t. $\mu$ yields
        \begin{equation*}
            \int \left(\frac{\dd\nu}{\dd\mu}\right)^\alpha\dd\mu = 1 + \frac{\alpha(\alpha-1)}{2}\int_\X c(x)^{\alpha-2}\left(\frac{\dd\nu}{\dd\mu}(x) - 1\right)^2\dd\mu(x)
        \end{equation*}
        for some $c(x)$ lying between $1$ and $\frac{\dd\nu}{\dd\mu}(x)$. Since the integration is over $\mu$, we can use the bound $c(x)\leq\|\frac{\dd\nu}{\dd\mu}\|_{L^\infty(\mu)}=M$ and the fact that $x\mapsto x^{\alpha-2}$ is non-decreasing when $\alpha\geq 2$ to get
        \begin{equation*}
            \int \left(\frac{\dd\nu}{\dd\mu}\right)^\alpha\dd\mu \leq 1 + \frac{\alpha(\alpha-1)}{2}M^{\alpha-2}\chi^2(\nu\|\mu),
        \end{equation*}
        or equivalently
        \begin{equation}\label{eq:pinsker_type_inequality_weak}
            D_\alpha(\nu\|\mu)\leq \frac1{\alpha-1}\log\left(1 + \frac{\alpha(\alpha-1)}{2}M^{\alpha-2}\chi^2(\nu\|\mu)\right)
        \end{equation}
        To get the stronger version of the result, we start similarly with Taylor's Theorem but this time with the integral remainder:
        \begin{equation*}
            t^\alpha = 1 + \alpha\big(t - 1\big) + \alpha(\alpha-1)\int_1^t (t-s)s^{\alpha-2}\dd s
        \end{equation*}
        Rearranging, this can be written as $C_\alpha(t) = \frac{\alpha(\alpha-1)}{(t-1)^2}\int_1^t s^{\alpha-2}(t-s)\dd s$. Applying the substitution $s= 1+u(t-1)$ for $u\in[0,1]$ yields
        \begin{equation}
            C_\alpha(t) = \alpha(\alpha-1)\int_0^1 (1-u)\big(1+u(t-1)\big)^{\alpha-2}\dd u.
        \end{equation}
        Since $\alpha\geq 2$, the term $(1+u(t-1))^{\alpha-2}$ is non-decreasing in $t$ for any fixed $u\in[0,1]$, and consequently $C_\alpha(t)$ is non-decreasing since it is the integral of non-decreasing functions. From this monotonicity, we directly obtain that $C_\alpha(\frac{\dd\nu}{\dd\mu}(x))\leq C_\alpha(M)$ for $\mu$-a.e. $x$, which equivalently reads
        \begin{equation*}
            \left(\frac{\dd\nu}{\dd\mu}(x)\right)^\alpha \leq 1+\alpha\left(\frac{\dd\nu}{\dd\mu}(x) -1\right) + C_\alpha(M)\left(\frac{\dd\nu}{\dd\mu}(x) - 1\right)^2
        \end{equation*}
        After integration w.r.t. $\mu$ and rearranging, we obtain the desired result:
        \begin{equation*}
            D_\alpha(\nu\|\mu) \leq \frac1{\alpha-1}\log\left(1+C_\alpha\left(M\right)\chi^2(\nu\|\mu)\right).
        \end{equation*}
        Notice it is stronger than~\cref{eq:pinsker_type_inequality_weak} because
        \begin{align*}
            C_\alpha(M) &= \alpha(\alpha-1)\int_0^1 (1-u)\big(1+u(M-1)\big)^{\alpha-2}\dd u\\
            &\leq \alpha(\alpha-1)\int_0^1 (1-u)M^{\alpha-2}\dd u\\
            &= \frac{\alpha(\alpha-1)}{2}M^{\alpha-2}.
        \end{align*}
    \subsection{Proof of~\cref{prop:ub_on_eta_alpha_by_eta_chi_squared}}\label{appendix:proof_ub_on_eta_alpha_by_eta_chi_squared}\noindent
        To lower bound the R\'enyi Divergence for $\alpha>0$, we use Pinsker's and Gilardoni's inequality~\cite{gilardoni2010pinsker}
        \begin{equation}
            D_\alpha(\nu\|\mu)\geq \min(1, \alpha)\cdot2\|\nu-\mu\|_{\sf TV}^2.\label{eq:gilardoni_inequality}
        \end{equation}
        Combined with~\cite[Equation 72]{makur2020comparison}, which states that
        \begin{equation*}
            \chi^2(\nu\|\mu)\leq \frac{2\|\nu-\mu\|_{\sf TV}^2}{\min_{x\in\supp(\mu)}\mu(x)},
        \end{equation*}
        we find
        \begin{equation}\label{eq:pinsker_type_inequality_renyi_chi_squared}
            D_\alpha(\nu\|\mu)\geq \min(1, \alpha)\cdot\chi^2(\nu\|\mu)\min_{x\in\supp(\mu)}\mu(x).
        \end{equation}
        For the upper bound we use~\cref{lemma:bound_between_renyi_chi_squared} and~\cref{eq:renyi_ub_by_chi_squared} together with $\log(1+x)\leq x$ for $x\geq 0$ to get
        \begin{equation}\label{eq:reverse_pinsker_type_inequality_renyi_chi_squared}
            D_\alpha(\nu\|\mu) \leq \begin{cases}
                \chi^2(\nu\|\mu),\quad\text{if }\alpha\in[0, 2]\\
                C_\alpha\Big(\frac1{\min_{x\in\supp(\mu)}\mu(x)}\Big)\frac{\chi^2(\nu\|\mu)}{\alpha-1}, \quad \text{if }\alpha>2.
            \end{cases}
        \end{equation}
        Both~\cref{eq:pinsker_type_inequality_renyi_chi_squared,eq:reverse_pinsker_type_inequality_renyi_chi_squared} hold whenever $\nu\ll\mu$. For $\alpha\geq1$, they can therefore be used to obtain
        \begin{align*}
            \eta_\alpha(\mu, K) &= \sup_{\substack{\nu\in\calP(\X):\\0<D_\alpha(\nu\|\mu)<\infty}}\frac{D_\alpha(\nu K\|\mu K)}{D_\alpha(\nu\|\mu)}\\
            &=\sup_{\substack{\nu\in\calP(\X):\\\nu\ll\mu\\\nu\neq\mu}}\frac{D_\alpha(\nu K\|\mu K)}{D_\alpha(\nu\|\mu)}\\
            &\leq \begin{cases}
                \sup_{\substack{\nu\in\calP(\X):\\\nu\ll\mu\\\nu\neq\mu}} \frac{\chi^2(\nu K\|\mu K)}{\chi^2(\nu\|\mu)\min_{x\in\supp(\mu)}\mu(x)},\quad \text{if }\alpha\in[1,2]\\
                \sup_{\substack{\nu\in\calP(\X):\\\nu\ll\mu\\\nu\neq\mu}} \frac{\frac{1}{\alpha-1}\chi^2(\nu K\|\mu K)C_\alpha\left(\frac1{\min_{y\in\supp(\mu K)}\mu K(y)}\right)}{\chi^2(\nu\|\mu)\min_{x\in\supp(\mu)}\mu(x)}, \quad\text{if }\alpha>2
            \end{cases} \\
            &= \begin{cases}
                \frac{\eta_{\chi^2}(\mu, K)}{\min_{x\in\supp(\mu)}\mu(x)},\quad \text{if }\alpha\in[1,2]\\
                C_\alpha\left(\frac1{\min_{y\in\supp(\mu K)}\mu K(y)}\right) \frac{\eta_{\chi^2}(\mu, K)}{(\alpha-1)\min_{x\in\supp(\mu)}\mu(x)}, \quad\text{if }\alpha>2.
            \end{cases}
        \end{align*}
        The case $\alpha\in(0,1)$ is more delicate because the supremum when computing $\eta_\alpha(\mu, K)$ is over
        \begin{equation*}
            \{\nu\in\calP(\X): 0<D_\alpha(\nu\|\mu)<\infty\}=\{\nu\in\calP(\X): \supp(\nu)\cap\supp(\mu)\neq\emptyset \text{ and } \nu\neq\mu\},
        \end{equation*}
        which does not necessarily coincide with the optimisation space $\{\nu\in\calP(\X): \nu\ll\mu \text{ and } \nu\neq\mu\}$ in the definition of $\eta_{\chi^2}(\mu, K)$. However, requiring $\mu$ to be full-support ensures
        \begin{equation*}
            \{\nu\in\calP(\X): \supp(\nu)\cap\supp(\mu)\neq\emptyset \text{ and } \nu\neq\mu\} = \{\nu\in\calP(\X): \nu\ll\mu \text{ and } \nu\neq\mu\},
        \end{equation*}
        so that by~\cref{eq:pinsker_type_inequality_renyi_chi_squared,eq:reverse_pinsker_type_inequality_renyi_chi_squared},
        \begin{align*}
            \eta_\alpha(\mu, K) &= \sup_{\substack{\nu\in\calP(\X):\\\supp(\nu)\cap\supp(\mu)\neq\emptyset\\\nu\neq\mu}}\frac{D_\alpha(\nu K\|\mu K)}{D_\alpha(\nu\|\mu)}\\
            &=\sup_{\substack{\nu\in\calP(\X):\\\nu\ll\mu\\\nu\neq\mu}}\frac{D_\alpha(\nu K\|\mu K)}{D_\alpha(\nu\|\mu)}\\
            &\leq \sup_{\substack{\nu\in\calP(\X):\\\nu\ll\mu\\\nu\neq\mu}}\frac{\chi^2(\nu K\|\mu K)}{\alpha\chi^2(\nu\|\mu)\min_{x\in\supp(\mu)}\mu(x)}\\
            &= \frac{\eta_{\chi^2}(\mu, K)}{\alpha\min_{x\in\X}\mu(x)}.
        \end{align*}
    \subsection{Comparing SDPI Constants of Functionally Related Divergences}\label{appendix:sdpi_relationship_concave_function}\noindent
        \begin{lemma}\label{lemma:sdpi_relationship_concave_function}
            Consider any divergence ${D(\cdot\|\cdot):\calP(\X)\times\calP(\X)\to\reals_+}$ that satisfies the DPI. Let ${g:\reals_+\to\reals_+}$ be a non-zero function such that $g(0)=0$, and let $D^{(g)}(\nu\|\mu)=g\big(D(\nu\|\mu)\big)$. Then, for any admissible pair $(\mu, K)$:
            \begin{enumerate}[a)]
                \item If $g$ is convex, then
                \begin{equation*}
                    \eta_{D^{(g)}}(\mu, K) \leq \eta_{D}(\mu, K)\quad\text{and}\quad    \eta_{D^{(g)}}(K) \leq \eta_{D}(K),
                \end{equation*}

                \item If $g$ is concave, then
                \begin{equation*}
                    \eta_{D^{(g)}}(\mu, K) \geq \eta_{D}(\mu, K)\quad\text{and}\quad  \eta_{D^{(g)}}(K) \geq \eta_{D}(K).
                \end{equation*}
            \end{enumerate}
        \end{lemma}
        \begin{proof}
            Let us consider the case where $g$ is convex first. Since $g$ is non-negative and $g(0)=0$, $g$ must be non-decreasing so that $D^{(g)}$ also satisfies the DPI. Moreover by~\cref{lemma:secant_slope}\labelcref{lemma:secant_slope_convex}, the function $x\mapsto\frac{g(x)}{x}$ is non-decreasing for $x>0$. Set $x=D(\nu\|\mu)$ and $y=D(\nu K\|\mu K)$. If $y=0$, then $\frac{D^{(g)}(\nu K\|\mu K)}{D^{(g)}(\nu\|\mu)}=\frac{D(\nu K\|\mu K)}{D(\nu\|\mu)}=0$. Otherwise, $0<y\leq x$, and the monotonicity of $t\mapsto\frac{g(t)}{t}$ gives
            \begin{equation*}
                \frac{g(y)}{y}\leq\frac{g(x)}{x}\iff\frac{D^{(g)}(\nu K\|\mu K)}{D^{(g)}(\nu\|\mu)}\leq\frac{D(\nu K\|\mu K)}{D(\nu\|\mu)}.
            \end{equation*}
            Taking the appropriate supremum on each side and using $\{\nu\in\calP(\X): 0<D^{(g)}(\nu\|\mu)<\infty\}\subseteq\{\nu\in\calP(\X): 0<D(\nu\|\mu)<\infty\}$ gives $\eta_{D^{(g)}}(\mu, K) \leq \eta_{D}(\mu, K)$ and further taking supremum over $\mu$ gives the distribution-independent result. If $g$ is concave, then its non-negativity, the assumption that it is non-zero, and $g(0)=0$ imply that $g(x)>0$ for every $x>0$. Hence, $D$ and $D^{(g)}$ have the same admissible pairs, and the same argument with reversed inequalities proves Part b).
        \end{proof}
        \Cref{corollary:sdpi_ub_lb_hellinger} follows directly from~\cref{lemma:sdpi_relationship_concave_function} with the function $g(x)=\frac{e^{(\alpha-1)x}-1}{\alpha-1}$, since $\calH_\alpha(\nu\|\mu)=g\big(D_\alpha(\nu\|\mu)\big)$.
    \subsection{Proof of~\cref{prop:bound_on_distribution_independent_sdpi}}\label{appendix:proof_bound_on_distribution_independent_sdpi}\noindent
        We first prove Part a). The cases $\alpha\in\{0,1\}$ reduce to the known inequality $\eta_{\sf KL}(K)\leq\eta_{\sf TV}(K)$. For $\alpha\in(0,1)$, \Cref{corollary:sdpi_ub_lb_hellinger} gives $\eta_\alpha(K)\leq\eta_{\calH_\alpha}(K)$ and since the Hellinger Divergence is in the family of $\varphi$-Divergences, it satisfies $\eta_{\calH_\alpha}(K)\leq\eta_{\sf TV}(K)$.

        Part b) is proven by considering $\X=\{0,1\}$ and the Z-channel $K = \big[\begin{smallmatrix}
            1 & 0\\
            1-\lambda & \lambda
        \end{smallmatrix}\big]$ with $\lambda\in(0, 1)$, for which $\eta_{\sf TV}(K)=\lambda$. Considering the input and output R\'enyi Divergences between the Dirac mass $\delta_1$ and $\mu_\varepsilon \triangleq [1-\varepsilon, \varepsilon]$, we find their ratio to be
        \begin{align*}
            \frac{\renyiDiv{\delta_1 K}{\mu_\varepsilon K}}{\renyiDiv{\delta_1}{\mu_\varepsilon}} &= \frac{\frac1{\alpha-1}\log\left((1-\lambda)^\alpha(1-\lambda\varepsilon)^{1-\alpha} + \lambda^\alpha(\lambda\varepsilon)^{1-\alpha}\right)}{\log \frac1\varepsilon} \\
            &= 1 - \frac{\frac1{\alpha-1}\log\left(\lambda + (1-\lambda)^\alpha(\frac1\varepsilon-\lambda )^{1-\alpha}\right)}{\log\frac1\varepsilon}.
        \end{align*}
        Letting $u(\varepsilon) = \log\frac1\varepsilon$ and $g(x) = \frac1{\alpha-1}\log\left(\lambda + (1-\lambda)^\alpha(e^x-\lambda )^{1-\alpha}\right)$, the ratio can be expressed as $1 - \frac{g\big(u(\varepsilon)\big)}{u(\varepsilon)}$, and we find that
        \begin{equation*}
            \lim_{\varepsilon\to0}\left(1-\frac{g\big(u(\varepsilon)\big)}{u(\varepsilon)}\right) = 1 - \lim_{u\to\infty}\frac{\frac1{\alpha-1}\log\left(\lambda + \frac{(1-\lambda)^\alpha}{(e^u-\lambda )^{\alpha-1}}\right)}{u}= 1.
        \end{equation*}
        But by definition of the R\'enyi-SDPI constant, one has
        \begin{equation*}
            \lim_{\varepsilon\to0}\frac{\renyiDiv{\delta_1 K}{\mu_\varepsilon K}}{\renyiDiv{\delta_1}{\mu_\varepsilon}} \leq \eta_\alpha(K)\leq 1,
        \end{equation*}
        so that $\eta_\alpha(K) = 1$ and in particular $\eta_\alpha(K) > \eta_{\sf TV}(K)$.
    \subsection{Proof of~\cref{prop:sdpi_ub_hellinger}}\label{appendix:proof_sdpi_ub_hellinger}\noindent
        We begin with the distribution-dependent setting. If $M_\alpha(\mu,K)=0$, then $K(\cdot|x)=\mu K$ for every $x\in\supp(\mu)$. Hence, $\nu K=\mu K$ for every $\nu\ll\mu$, so that $\eta_\alpha(\mu,K)=0$. We may therefore assume that $M_\alpha(\mu,K)>0$.

        To build intuition, we first establish a weaker bound. Using the SDPI for Hellinger Divergences, we have
        \begin{align*}
            \eta_\alpha(\mu, K) &= \sup_{\substack{\nu\in\calP(\X):\\0<D_\alpha(\nu\|\mu)<\infty}}\frac{\log\big(1+(\alpha-1)\calH_\alpha(\nu K\|\mu K)\big)}{\log\big(1+(\alpha-1)\calH_\alpha(\nu\|\mu)\big)} \\
            &\leq \sup_{\substack{\nu\in\calP(\X):\\0<D_\alpha(\nu\|\mu)<\infty}}\frac{\log\big(1+\eta_{\calH_\alpha}(\mu, K)(\alpha-1)\cdot\calH_\alpha(\nu\|\mu)\big)}{\log\big(1+(\alpha-1)\calH_\alpha(\nu\|\mu)\big)}.
        \end{align*}
        Observe that for any $\nu\in\calP(\X)$ satisfying $0<D_\alpha(\nu\|\mu)<\infty$, the convexity of $\calH_\alpha$ in its first argument implies $0<\calH_\alpha(\nu\|\mu)\leq \sup_{x\in\supp(\mu)}\calH_\alpha(\delta_x\|\mu) \triangleq C_\mu$. We can therefore write
        \begin{align}
            \eta_\alpha(\mu, K) &\leq \sup_{t\in(0, C_\mu]}\frac{\log\big(1+(\alpha-1)\eta_{\calH_\alpha}(\mu, K)\cdot t\big)}{\log\big(1+(\alpha-1)t\big)}\nonumber\\
            &= \frac{\log\big(1+(\alpha-1)\eta_{\calH_\alpha}(\mu, K) \cdot C_\mu\big)}{\log\big(1+(\alpha-1)C_\mu\big)},\label{eq:ub_sdpi_via_hellinger_weak}
        \end{align}
        where the equality follows because the map $t\mapsto \frac{\log(1+(\alpha-1)\lambda t)}{\log(1+(\alpha-1)t)}$ is non-decreasing for any $\lambda\in[0,1]$ and $\alpha>1$ (see~\cref{corollary:log_ratio_increasing}). Applying the same steps to bound $\eta_\alpha(K)$ would yield a trivial bound of 1; however, it turns out that this initial idea can be refined to improve the bound for both SDPI constants. We flesh out the details for the distribution-dependent case.

        The refinement is obtained by providing an additional bound on $\calH_\alpha(\nu K\|\mu K)$. By joint convexity of $\calH_\alpha$ and linearity of $K$, one finds $\calH_\alpha(\nu K\|\mu K) \leq \sup_{x\in\supp(\mu)}\calH_\alpha(\delta_x K\|\mu K)=M_\alpha(\mu, K)$. Consequently,
        \begin{align}
            \eta_\alpha(\mu, K) &\leq \sup_{\substack{\nu\in\calP(\X):\\0<D_\alpha(\nu\|\mu)<\infty}}\frac{\log\Big(1+(\alpha-1)\min\big(\eta_{\calH_\alpha}(\mu, K)\cdot\calH_\alpha(\nu\|\mu), M_\alpha(\mu, K)\big)\Big)}{\log\big(1+(\alpha-1)\calH_\alpha(\nu\|\mu)\big)}\nonumber\\
            &\leq \sup_{t\in(0, C_\mu]}\frac{\log\Big(1+(\alpha-1)\min\big(\eta_{\calH_\alpha}(\mu, K)\cdot t, M_\alpha(\mu, K)\big)\Big)}{\log\big(1+(\alpha-1)t\big)}.\label{eq:ub_sdpi_via_hellinger_strong_intermediate_computation}
        \end{align}
        To evaluate this supremum, we analyse the objective function in \cref{eq:ub_sdpi_via_hellinger_strong_intermediate_computation}. Since $M_\alpha(\mu,K)>0$ implies $\eta_{\calH_\alpha}(\mu,K)>0$, the quantity $t^\star \triangleq \frac{M_\alpha(\mu, K)}{\eta_{\calH_\alpha}(\mu, K)}$ is well-defined. Note that by the definition of $\eta_{\calH_\alpha}(\mu, K)$, we have $M_\alpha(\mu, K) \leq \eta_{\calH_\alpha}(\mu, K) C_\mu$, ensuring $t^\star \leq C_\mu$.

        For $0 < t \leq t^\star$, the function simplifies to $t\mapsto \frac{\log(1+(\alpha-1)\eta_{\calH_\alpha}(\mu, K)\cdot t)}{\log(1+(\alpha-1)t)}$. As established, this is non-decreasing, so its maximum on this interval is attained at $t^\star$. For $t^\star < t \leq C_\mu$, the function simplifies to $t\mapsto \frac{\log(1+(\alpha-1)M_\alpha(\mu, K))}{\log(1+(\alpha-1)t)}$, which is strictly decreasing. Thus, the global supremum on $(0, C_\mu]$ is achieved at $t^\star$, and we conclude that
        \begin{equation}
            \eta_\alpha(\mu, K) \leq \frac{\log\big(1+(\alpha-1)M_\alpha(\mu, K)\big)}{\log\left(1+(\alpha-1)\frac{M_\alpha(\mu, K)}{\eta_{\calH_\alpha}(\mu, K)}\right)}.
        \end{equation}
        This refinement is never worse than~\cref{eq:ub_sdpi_via_hellinger_weak} because $t^\star \leq C_\mu$, meaning the ratio of logarithms is capped at a lower value than what the unrefined version of the inequality would allow at $t=C_\mu$.

        For the distribution-independent setting, if $M_\alpha(K)=0$, then all the conditionals $K(\cdot|x)$ coincide, so that $\eta_\alpha(K)=0$. Otherwise, repeating the preceding argument with $\eta_{\calH_\alpha}(K)$ in place of $\eta_{\calH_\alpha}(\mu,K)$ and using $\calH_\alpha(\nu K\|\mu K)\leq\sup_{x,x^\prime\in\X}\calH_\alpha(\delta_xK\|\delta_{x^\prime}K)=M_\alpha(K)$ yields the corresponding bound on $\eta_\alpha(K)$.
    \subsection{Proof of~\cref{prop:bounds_on_renyi_infty_sdpi_by_tv}}\label{appendix:proof_bounds_on_renyi_infty_sdpi_by_tv}\noindent
        The result follows after upper and lower bounding the \mbox{$\infty$-R\'enyi Divergence} by the Total Variation Distance. Fix $\nu\in\calP(\X)$ such that $\nu\ll\mu$. First, notice that the following identity holds:
        \begin{equation*}
            \left\|\frac{\dd\nu}{\dd\mu}\right\|_{L^\infty(\mu)}= \sup_{\substack{A\in\Sigma_\X:\\\mu(A)>0}} \frac{\nu(A)}{\mu(A)}.
        \end{equation*}
        Let us also recall that ${\|\nu-\mu\|_{\sf TV} \triangleq \sup_{A\in\Sigma_\X}|\nu(A)-\mu(A)|}$, and in particular ${\|\nu-\mu\|_{\sf TV} = \sup_{A\in\Sigma_\X:\mu(A)>0}\{\nu(A)-\mu(A)\}}$ since $\nu\ll\mu$, so that in this case
        \begin{equation*}
            D_\infty(\nu\|\mu) = \log\sup_{\substack{A\in\Sigma_\X:\\\mu(A)>0}} \frac{\nu(A)}{\mu(A)}\leq \sup_{\substack{A\in\Sigma_\X:\\\mu(A)>0}} \frac{\nu(A)}{\mu(A)}-1\leq \frac1{\min_{x\in\supp(\mu)}\mu(x)}\|\nu-\mu\|_{\sf TV},
        \end{equation*}
        where we used that $\log(x)\leq x-1$ for $x>0$. Moreover, since $\log(x)\geq1-1/x$ for $x>0$,
        \begin{equation*}
            D_\infty(\nu\|\mu) = \log\sup_{\substack{A\in\Sigma_\X:\\\mu(A)>0}} \frac{\nu(A)}{\mu(A)} = \log\sup_{\substack{A\in\Sigma_\X:\\\nu(A),\mu(A)>0}} \frac{\nu(A)}{\mu(A)} \geq \sup_{\substack{A\in\Sigma_\X:\\\nu(A),\mu(A)>0}} \left\{1-\frac{\mu(A)}{\nu(A)}\right\} \geq \|\nu-\mu\|_{\sf TV},
        \end{equation*}
        where in the last inequality we used $1/\nu(A)\geq 1$. The bounds can now be stated, but it is important to observe the set inclusion
        \begin{align*}
            \{\nu\in\calP(\X): 0<D_\infty(\nu\|\mu)<\infty\} &= \{\nu\in\calP(\X): \nu\ll\mu \text{ and } \nu\neq\mu\}\\
            &\subseteq \{\nu\in\calP(\X): \nu\neq\mu\}\\
            &= \{\nu\in\calP(\X): 0<\|\nu-\mu\|_{\sf TV}<\infty\},
        \end{align*}
        due to the fact that Total Variation Distance is always bounded. Using this and the bounds just derived leads to
        \begin{align*}
            \eta_\infty(\mu, K) &= \sup_{\substack{\nu\in\calP(\X):\\\nu\neq\mu\\\nu\ll\mu}} \frac{D_\infty(\nu K\|\mu K)}{D_\infty(\nu\|\mu)} \\
            &\leq \frac1{\min_{y\in\supp(\mu K)}\mu K(y)}\cdot\sup_{\substack{\nu\in\calP(\X):\\\nu\neq\mu\\\nu\ll\mu}} \frac{\|\nu K-\mu K\|_{\sf TV}}{\|\nu-\mu\|_{\sf TV}} \\
            &\leq \frac1{\min_{y\in\supp(\mu K)}\mu K(y)}\cdot\sup_{\substack{\nu\in\calP(\X):\\\nu\neq\mu}} \frac{\|\nu K-\mu K\|_{\sf TV}}{\|\nu-\mu\|_{\sf TV}} \\
            &= \frac{1}{\min_{y\in\supp(\mu K)}\mu K(y)} \eta_{\sf TV}(\mu, K).
        \end{align*}
    \subsection{Proof of~\cref{lemma:distribution_dependent_eta_tv}}
\label{appendix:proof_distribution_dependent_eta_tv}\noindent
        For every $\nu\in\calP(\X)$ with $\nu\neq\mu$, let $\rho_\nu^\pm \triangleq \frac{(\nu-\mu)_\pm}{\|\nu-\mu\|_{\sf TV}}$ denote the probability measures that separate the signed measure $\nu-\mu$ into its positive and negative parts. Hence, $\rho_\nu^+\in\calP(\X)$, $\rho_\nu^-\in\calP\big(\supp(\mu)\big)$.

        Suppose first that $\max_{x\in\X,x^\prime\in\supp(\mu)} \|K(\cdot|x)-K(\cdot|x^\prime)\|_{\sf TV}>0$. Choose $x_\star\in\X$ and $x_\star^\prime\in\supp(\mu)$ attaining the maximum. Necessarily $x_\star\neq x_\star^\prime$, and for $0<\varepsilon<\mu(x_\star^\prime)$, define $\nu_\varepsilon \triangleq \mu+\varepsilon\bigl(\delta_{x_\star}-\delta_{x_\star^\prime}\bigr)$. Then, we have
        \begin{equation*}
            \max_{x\in\X,x^\prime\in\supp(\mu)} \|K(\cdot|x)-K(\cdot|x^\prime)\|_{\sf TV} =\frac{\|\nu_\varepsilon K-\mu K\|_{\sf TV}}{\|\nu_\varepsilon-\mu\|_{\sf TV}} \leq\eta_{\sf TV}(\mu,K).
        \end{equation*}
        On the other hand, we have the upper bound
        \begin{align*}
            \eta_{\sf TV}(\mu,K)&=\sup_{\substack{\nu\in\calP(\X):\\\nu\neq\mu}}
              \frac{\sup_{\substack{f\in\F(\Y):\\0\leq f\leq1}}\big(\mean{\nu K}{f}-\mean{\mu K}{f}\big)}{\|\nu-\mu\|_{\sf TV}}&\left(\substack{\text{Variational representation}\\\text{of TV Distance}}\right)&\\
            &=\sup_{\substack{\nu\in\calP(\X):\\\nu\neq\mu}}\frac{\sup_{\substack{f\in\F(\Y):\\0\leq f\leq1}}\big(\mean{\nu}{Kf}-\mean{\mu}{Kf}\big)}{\|\nu-\mu\|_{\sf TV}}\\
            &=\sup_{\substack{\nu\in\calP(\X):\\\nu\neq\mu}}\sup_{\substack{f\in\F(\Y)\\0\leq f\leq1}}\left(\mean{\rho_\nu^+}{Kf}-\mean{\rho_\nu^-}{Kf}\right)&\left(\substack{\text{By definition of $\rho_\nu^\pm$}}\right)&\\
            &\leq \sup_{\substack{f\in\F(\Y)\\0\leq f\leq1}}\bigl(\max_{x\in\X} Kf(x)-\min_{x^\prime\in\supp(\mu)}Kf(x^\prime)\bigr)\\
            &= \sup_{\substack{f\in\F(\Y)\\0\leq f\leq1}} \max_{\substack{x\in\X,\\x^\prime\in\supp(\mu)}}\bigl(Kf(x)-Kf(x^\prime)\bigr)\\
            &= \max_{\substack{x\in\X,\\x^\prime\in\supp(\mu)}}\sup_{\substack{f\in\F(\Y):\\0\leq f\leq1}}\left(\mean{K(\cdot|x)}{f} - \mean{K(\cdot|x^\prime)}{f}\right)\\
            &= \max_{\substack{x\in\X,\\x^\prime\in\supp(\mu)}}
              \|K(\cdot|x)-K(\cdot|x^\prime)\|_{\sf TV}&\left(\substack{\text{Variational representation}\\\text{of TV Distance}}\right)&,
        \end{align*}
        from which we conclude the desired equality.

        If $\max_{x\in\X,x^\prime\in\supp(\mu)} \|K(\cdot|x)-K(\cdot|x^\prime)\|_{\sf TV}=0$, the last chain of inequalities directly gives $\eta_{\sf TV}(\mu,K)\leq 0$, hence the desired identity also holds in this case and this concludes the proof.
    \subsection{Proof of~\cref{prop:eta_tv_lower_bounds_eta_inf}}\label{appendix:proof_eta_tv_lower_bounds_eta_inf}\noindent
        We first prove Part a). Since $\mu$ has full support, \cref{lemma:distribution_dependent_eta_tv} gives $\eta_{\sf TV}(\mu, K)=\max_{x, x^\prime\in\X}\|K(\cdot|x)-K(\cdot|x^\prime)\|_{\sf TV}=\eta_{\sf TV}(K)$. Let us denote by $x_1, x_2\in\X$ the elements which achieve this maximum. The proof proceeds as follows: we first find a lower bound on $\eta_\infty(\mu, K)$ in terms of $x_1, x_2$, and eventually connect it back to $\eta_{\sf TV}(\mu, K)$. To this end, consider the two subsets $A_i=\X\setminus\{x_i\}$ for $i=1,2$, and notice that
        \begin{align*}
            \max_{y\in\supp(\mu K)} \frac{\mu_{|A_i}K(y)}{\mu K(y)}&= \max_{y\in\supp(\mu K)} \frac{\sum_{x\in\X}K(y|x)\mu_{|A_i}(x)}{\mu K(y)}\\
            &= \max_{y\in\supp(\mu K)} \frac{\frac1{\mu(A_i)}\sum_{x\in A_i}K(y|x)\mu(x)}{\mu K(y)}\\
            &= \max_{y\in\supp(\mu K)} \frac{\frac1{\mu(A_i)}\big(\mu K(y)-\mu(x_i)K(y|x_i)\big)}{\mu K(y)}\\
            &= \frac1{1-\mu(x_i)}\left(1-\mu(x_i)\cdot\min_{y\in\supp(\mu K)}\frac{K(y|x_i)}{\mu K(y)}\right).
        \end{align*}
        Hence we get
        \begin{align}
            \frac{D_\infty(\mu_{|A_i}K\|\mu K)}{D_\infty(\mu_{|A_i}\|\mu)} &= \frac{\log\left(\max_{y\in\supp(\mu K)} \frac{\mu_{|A_i}K(y)}{\mu K(y)}\right)}{-\log\big(\mu(A_i)\big)}\nonumber\\
            &= 1 - \frac{\log\left(1-\mu(x_i)\cdot\min_{y\in\supp(\mu K)}\frac{K(y|x_i)}{\mu K(y)}\right)}{\log\big(1-\mu(x_i)\big)}\nonumber\\
            &\geq 1-\min_{y\in\supp(\mu K)}\frac{K(y|x_i)}{\mu K(y)},\label{eq:eta_inf_eta_tv_proof_ratio_D_inf}
        \end{align}
        where in the last inequality we used that $\lambda_i\triangleq\min_{y\in\supp(\mu K)}\frac{K(y|x_i)}{\mu K(y)}\in[0,1]$ and $\frac{\log(1-\lambda_i t)}{\log(1-t)}\leq \lambda_i$ for $t\in(0,1)$ when $\lambda_i\in[0,1]$ (see~\cref{corollary:log_ratio_decreasing}). Using~\cref{eq:eta_inf_eta_tv_proof_ratio_D_inf} with~\cref{prop:infty_sdpi_achieved_on_conditionals} yields
        \begin{equation}
            \eta_\infty(\mu, K)=\max_{\substack{A\in\Sigma_\X:\\0<\mu(A)<1}}\frac{D_\infty(\mu_{|A}K\|\mu K)}{D_\infty(\mu_{|A}\|\mu)} \geq \max\left(1-\lambda_1, 1-\lambda_2\right)= 1-\min(\lambda_1, \lambda_2).\label{eq:eta_inf_eta_tv_proof_lb_on_eta_inf}
        \end{equation}
        Let us now focus on $\eta_{\sf TV}(\mu, K)$ for a moment. By the variational representation of Total Variation Distance, there exists a set $A\subseteq\Y$ such that
        \begin{equation*}
            K(A|x_1)-K(A|x_2) = \eta_{\sf TV}(\mu, K) \iff \eta_{\sf TV}(\mu, K) = 1-\sum_{y\in A^c}K(y|x_1) -\sum_{y\in A} K(y|x_2).
        \end{equation*}
        We also know $\lambda_i\mu K(y) \leq K(y|x_i)\;\forall y\in\supp(\mu K)$, and moreover since $\mu$ has full support we have that $\mu K(y)=0\iff K(y|x)=0\;\forall x\in\X$, or in other words even when $y\not\in\supp(\mu K)$, we can write $\lambda_i\mu K(y) \leq K(y|x_i)$. Consequently,
        \begin{equation}\label{eq:bound_eta_tv_temp}
            \eta_{\sf TV}(\mu, K) \leq 1-\sum_{y\in A^c}\lambda_1\mu K(y) -\sum_{y\in A} \lambda_2 \mu K(y) \leq 1 - \min(\lambda_1, \lambda_2)
        \end{equation}
        and combining with~\cref{eq:eta_inf_eta_tv_proof_lb_on_eta_inf} concludes the proof.

        Next, we show that we have $\eta_{\sf TV}(\mu, K)=\eta_\infty(\mu, K)$ if and only if $\eta_{\sf TV}(\mu, K)=0$ or $\eta_{\sf TV}(\mu, K)=1$.
        \begin{itemize}
            \item \textbf{``If'' direction $\Longleftarrow$:} If $\eta_{\sf TV}(\mu, K)=1$,  \cref{eq:eta_tv_lower_bounds_eta_inf} directly implies $\eta_\infty(\mu, K)=1$. If $\eta_{\sf TV}(\mu, K)=0$, this means $\max_{x, x^\prime\in\X}\|K(\cdot|x)-K(\cdot|x^\prime)\|_{\sf TV}=0$, that is, all the conditionals $\{K(\cdot|x)\}_{x\in\X}$ are equal. Consequently, every input distribution maps to the same output distribution, causing the divergence $D_\infty(\nu K\|\mu K)$ to vanish for all $\nu$ and yielding $\eta_\infty(\mu, K)=0$.
            \item \textbf{``Only if'' direction $\Longrightarrow$:} We show the contrapositive, so assume $\eta_{\sf TV}(\mu, K)\in(0,1)$. By~\cref{eq:bound_eta_tv_temp}, we see that $\min(\lambda_1, \lambda_2) <1$. Let us consider two cases:
            \begin{enumerate}
                \item $\min(\lambda_1, \lambda_2)=0$: In this case, \cref{eq:eta_inf_eta_tv_proof_lb_on_eta_inf} gives $\eta_\infty(\mu, K)\geq 1-\min(\lambda_1, \lambda_2)=1 > \eta_{\sf TV}(\mu, K)$.
                \item $\min(\lambda_1, \lambda_2)>0$: Let us show that the inequality in~\cref{eq:eta_inf_eta_tv_proof_ratio_D_inf} becomes strict. This can be seen through~\cref{corollary:log_ratio_decreasing}, which implies that $\frac{\log(1-\lambda t)}{\log(1-t)}< \lambda$ for $t,\lambda\in(0,1)$, and in particular
                \begin{equation*}
                    \eta_\infty(\mu, K)\geq\max_{i=1,2}\left\{1 - \frac{\log\left(1-\mu(x_i)\cdot\lambda_i\right)}{\log\big(1-\mu(x_i)\big)}\right\}> 1-\min(\lambda_1, \lambda_2)\geq \eta_{\sf TV}(\mu, K).
                \end{equation*}
            \end{enumerate}
        \end{itemize}
        For Part b), fix any fully-supported $\mu\in\calP(\X)$. By~\cref{lemma:distribution_dependent_eta_tv} and Part~a),
        \begin{equation*}
            \eta_{\sf TV}(K)=\eta_{\sf TV}(\mu,K)\leq\eta_\infty(\mu,K)\leq\eta_\infty(K).
        \end{equation*}
        If $\eta_{\sf TV}(K)=0$, then all the conditionals of $K$ coincide and hence $\eta_\infty(K)=0$, while $\eta_{\sf TV}(K)=1$ directly implies $\eta_\infty(K)=1$ by the preceding inequality. Conversely, if $\eta_{\sf TV}(K)\in(0,1)$, the strict inequality in Part~a) gives
        \begin{equation*}
            \eta_{\sf TV}(K)=\eta_{\sf TV}(\mu,K)<\eta_\infty(\mu,K)\leq\eta_\infty(K),
        \end{equation*}
        which concludes the proof.
\section{Proofs for~\cref{sec:applications}}\noindent
    \subsection{Proof of~\cref{thm:renyi_ldp_implies_contraction}}\label{appendix:proof_renyi_ldp_implies_contraction}\noindent
        Suppose first that $\alpha\in[2,\infty)$. Part~b) of~\cref{prop:eta_parameter_renyi_binary} reduces the optimisation defining $\eta_\alpha(K)$ to pairs $\nu=\delta_x$ and $\mu=(1-t)\delta_x+t\delta_{x^\prime}$, where $x\neq x^\prime$ and $t\in(0,1)$. Moreover, the RLDP assumption implies that $K(\cdot|x)$ and $K(\cdot|x^\prime)$ have the same support, so we find

        \begin{align*}
            e^{-D_\alpha\left(K(\cdot|x)\|(1-t)K(\cdot|x)+tK(\cdot|x^\prime)\right)}
            &=\left(\mean{K(\cdot|x)}{\left((1-t)+t\frac{\dd K(\cdot|x^\prime)}{\dd K(\cdot|x)}\right)^{1-\alpha}}\right)^{\frac1{1-\alpha}}\\
            &\geq(1-t)+t\left(\mean{K(\cdot|x)}{\left(\frac{\dd K(\cdot|x^\prime)}{\dd K(\cdot|x)}\right)^{1-\alpha}}\right)^{\frac1{1-\alpha}}
            &\left(\substack{\text{Reverse Minkowski}\\\text{inequality}}\right)&\\
            &=(1-t)+te^{-D_\alpha\left(K(\cdot|x)\|K(\cdot|x^\prime)\right)}\\
            &\geq(1-t)+te^{-\varepsilon}&\left(\substack{\text{Since $K$ is}\\\text{$(\alpha, \varepsilon)$-RLDP}}\right)&\\
            &=e^{-\varepsilon}+\bigl(1-e^{-\varepsilon}\bigr)(1-t)\\
            &\geq(1-t)^{1-e^{-\varepsilon}}.&\left(\substack{\text{Weighted AM--GM}\\\text{inequality}}\right)&.
        \end{align*}
        Taking negative logarithms and using $D_\alpha\bigl(\delta_x\|(1-t)\delta_x+t\delta_{x^\prime}\bigr)=-\log(1-t)$ gives
        \begin{equation}\label{eq:rldp_implies_rsdpi_contraction_bound_before_supremum}
            \frac{D_\alpha\left(K_x\|(1-t)K_x+tK_{x^\prime}\right)}{D_\alpha\bigl(\delta_x\|(1-t)\delta_x+t\delta_{x^\prime}\bigr)}\leq\bigl(1-e^{-\varepsilon}\bigr).
        \end{equation}
        Taking the supremum over $x,x^\prime$ and $t$ proves $\eta_\alpha(K)\leq 1-e^{-\varepsilon}$. If $\varepsilon>0$, the bound is strict. Indeed, the weighted AM--GM inequality is strict for every $t\in(0,1)$. Moreover, for each $x,x^\prime$, the contraction ratio in~\cref{eq:rldp_implies_rsdpi_contraction_bound_before_supremum} extends continuously to $[0,1]$ as a function of $t$ and equals zero at both endpoints. Since $\X$ is finite, the overall supremum is therefore attained for some $x,x^\prime\in\X$ and $t\in(0,1)$, where the weighted AM--GM inequality is strict.

        At order $\infty$,~\cref{prop:infty_renyi_sdpi_closed_form} can equivalently be written as
        \begin{equation*}
            \eta_\infty(K)=1-\exp\left(-\sup_{x,x^\prime\in\X}D_\infty\bigl(K(\cdot|x)\|K(\cdot|x^\prime)\bigr)\right).
        \end{equation*}
        Since $\varepsilon$-LDP coincides with $(\infty,\varepsilon)$-R\'enyi LDP, the desired equivalence follows directly.

    \subsection{Proof of~\cref{thm:ldp_sharp_renyi_contraction}}\label{appendix:proof_ldp_sharp_renyi_contraction}\noindent
        The case $\varepsilon=0$ is immediate, since the LDP condition forces all the conditionals $K(\cdot|x)$ to coincide, and hence $\eta_\alpha(K)=0=\Upsilon_{\alpha,0}$. Suppose therefore that $\varepsilon>0$.

        By Part~a) of~\cref{prop:eta_parameter_renyi_binary}, it suffices to consider pairs of probability measures $(\nu, \mu)$ supported on two elements $x,x^\prime\in\X$, so let us write
        \begin{equation*}
            \nu=(1-s)\delta_x+s\delta_{x^\prime},\qquad\mu=(1-t)\delta_x+t\delta_{x^\prime},
        \end{equation*}
        where $s,t\in[0,1]$. The $\varepsilon$-LDP condition ensures that $K(\cdot|x)$ and $K(\cdot|x^\prime)$ have the same support and that
        \begin{equation*}
            e^{-\varepsilon}\leq r\triangleq\frac{\dd K(\cdot|x^\prime)}{\dd K(\cdot|x)}\leq e^\varepsilon,\qquad K(\cdot|x)\text{-almost everywhere}.
        \end{equation*}

        Fix $\alpha\in(0,\infty)\setminus\{1\}$ and consider the function
        \begin{equation*}
            h(u)\triangleq(1-s+su)^\alpha(1-t+tu)^{1-\alpha}.
        \end{equation*}
        Its second derivative is $h^{\prime\prime}(u)=\alpha(\alpha-1)(s-t)^2(1-s+su)^{\alpha-2}(1-t+tu)^{-\alpha-1}$, so that $h$ is convex when $\alpha>1$ and concave when $\alpha\in(0,1)$. Suppose first that $\alpha>1$. Then, we have
        \begin{align*}
            e^{(\alpha-1)D_\alpha(\nu K\|\mu K)}&=\mean{K(\cdot|x)}{h(r)}\\
            &=\mean{K(\cdot|x)}{h\left(\frac{e^\varepsilon - r}{e^\varepsilon -e^{-\varepsilon}}e^{-\varepsilon}+\frac{r - e^{-\varepsilon}}{e^\varepsilon -e^{-\varepsilon}}e^{\varepsilon}\right)}\\
            &\leq\mean{K(\cdot|x)}{\frac{e^\varepsilon - r}{e^\varepsilon -e^{-\varepsilon}}h(e^{-\varepsilon})+\frac{r - e^{-\varepsilon}}{e^\varepsilon -e^{-\varepsilon}}h(e^{\varepsilon})}&\left(\substack{\text{Convexity of $h$}}\right)&\\
            &=\frac{e^\varepsilon}{1+e^\varepsilon}h(e^{-\varepsilon})+\frac1{1+e^\varepsilon}h(e^\varepsilon)&\left(\substack{\text{Since $\mean{K(\cdot|x)}{r}=1$}}\right)&\\
            &=e^{(\alpha-1)D_\alpha\left(\nu\mathrm{BSC}\left(\frac1{1+e^\varepsilon}\right)\middle\|\mu\mathrm{BSC}\left(\frac1{1+e^\varepsilon}\right)\right)}.
        \end{align*}
        When $\alpha\in(0,1)$, concavity reverses the inequality. Since $\alpha-1<0$, both cases yield
        \begin{equation}\label{eq:ldp_direct_renyi_comparison}
            D_\alpha(\nu K\|\mu K)\leq D_\alpha\left(\nu\mathrm{BSC}\left(\frac1{1+e^\varepsilon}\right)\middle\|\mu\mathrm{BSC}\left(\frac1{1+e^\varepsilon}\right)\right).
        \end{equation}
        Taking limits in~\cref{eq:ldp_direct_renyi_comparison} as $\alpha\to1$ and $\alpha\to\infty$ gives the corresponding inequalities at orders one and infinity. Dividing~\cref{eq:ldp_direct_renyi_comparison} by $D_\alpha(\nu\|\mu)$ and using the definition of the contraction coefficient gives
        \begin{equation*}
            \frac{D_\alpha(\nu K\|\mu K)}{D_\alpha(\nu\|\mu)}\leq\eta_\alpha\left(\mathrm{BSC}\left(\frac1{1+e^\varepsilon}\right)\right)=\Upsilon_{\alpha,\varepsilon}.
        \end{equation*}
        Taking the supremum over all admissible $s,t,x,x^\prime$ proves the result. For $\alpha=0$, recall that $\eta_0$ is defined through reverse KL Divergence. One can thus use~\cref{eq:ldp_direct_renyi_comparison} at $\alpha=1$ and interchange $\nu$ and $\mu$. Applying the same remaining steps as the other orders yields the result for $\alpha=0$.
    \subsection{Proof of~\cref{corollary:asymptotic_renyi_contraction}}\label{appendix:proof_asymptotic_renyi_contraction}\noindent
        By reversibility and the spectral characterisation of the $\chi^2$-SDPI constant, $\eta_{\chi^2}(\pi,K^t)=\eta_{\chi^2}(\pi,K)^t$~\cite[Equation 86]{makur2020comparison}. Moreover, irreducibility ensures that $\pi$ has full support, while stationarity gives $\pi K^t=\pi$. Thus, combining~\cref{thm:local_chi_squared_lb_on_renyi_sdpi,prop:ub_on_eta_alpha_by_eta_chi_squared}, together with the corresponding lower bound $\eta_{\chi^2}(\pi,K^t)\leq\eta_{\sf KL}(\pi,K^t)$ when $\alpha=1$, shows that for every $\alpha\in(0,\infty)$ there exists a finite constant $c_{\alpha,\pi}$, independent of $t$, such that
        \begin{equation*}
            \eta_{\chi^2}(\pi,K)^t\leq\eta_\alpha(\pi,K^t)\leq c_{\alpha,\pi}\eta_{\chi^2}(\pi,K)^t.
        \end{equation*}
        Taking $t$-th roots and letting $t\to\infty$ concludes the proof.
    \subsection{Proof of~\cref{prop:mcmt}}\label{appendix:proof_mcmt}\noindent
        Iterating the R\'enyi-SDPI and using $\pi K=\pi$ gives $D_\alpha(\nu K^t\|\pi)\leq\eta_\alpha(\pi,K)^tD_\alpha(\nu\|\pi)$. Applying the increasing function $x\mapsto\frac{e^{(\alpha-1)x}-1}{\alpha-1}$ to both sides of this inequality and using the relationship between R\'enyi and Hellinger Divergences yields~\cref{eq:mcmt_non_linear_sdpi}.
    \subsection{Computations for~\cref{example:community_chain}}\label{appendix:partition_chain_computations}\noindent
        Abawonse et al.~\cite[Theorem~3 and Equation~(7)]{abawonse2026generalized} computed $\eta_2$ for resampling kernels. The following lemma extends their result to a broader class of Markov chains, including the one considered in~\cref{example:community_chain}.
        \begin{lemma}\label{lemma:eta_2_partition_chain}
            Let $K\in\calP(\X|\X)$ induce a reversible Markov chain with stationary distribution $\pi$, and consider a partition $\mathcal{A}=\{A_1,\dots,A_q\}$ of $\supp(\pi)$ with $q\geq2$. Define the Markov operators $E_\mathcal{A}$ and $\Pi$ through $E_\mathcal{A}f\triangleq\sum_{j=1}^q\mean{\pi_{|A_j}}{f}\mathbbm{1}_{A_j}$ and $\Pi f\triangleq\mean{\pi}{f}\allones$. If, as operators on $L^2(\pi)$,
            \begin{equation}\label{eq:eta_2_partition_chain_assumption}
                K^2=\lambda E_\mathcal{A}+(1-\lambda)\Pi
            \end{equation}
            for some $\lambda\in[0,1]$, then for every integer $t\geq1$
            \begin{equation}\label{eq:eta_2_partition_chain}
                \eta_{\chi^2}(\pi,K^t)=\lambda^t\qquad\text{and}\qquad \eta_2(\pi,K^t)=\frac{\log\Big(1+\lambda^t\left((\pi_{\min}^{\mathcal A})^{-1}-1\right)\Big)}{\log\big((\pi_{\min}^{\mathcal A})^{-1}\big)},
            \end{equation}
            where $\pi_{\min}^{\mathcal A}\triangleq\min_{j\in[q]}\pi(A_j)$. Moreover, for every integer $t\geq1$ and every $\nu\in\calP(\X)$ with $\nu\ll\pi$,
            \begin{equation}\label{eq:exact_chi_squared_partition_chain}
                \chi^2(\nu K^t\|\pi)=\lambda^t\chi^2(\nu E_{\mathcal A}\|\pi).
            \end{equation}
        \end{lemma}
        \begin{proof}
            Since $E_{\mathcal A}$ and $\Pi$ are conditional-expectation operators, we have $E_{\mathcal A}^2=E_{\mathcal A}$, $\Pi^2=\Pi$, and $E_{\mathcal A}\Pi=\Pi E_{\mathcal A}=\Pi$. Using these identities together with~\cref{eq:eta_2_partition_chain_assumption} inductively gives
            \begin{equation}\label{eq:expression_K_2t}
                K^{2t}=\lambda^tE_{\mathcal A}+(1-\lambda^t)\Pi.
            \end{equation}
            Fix $\nu\ll\pi$ and let $f\triangleq\frac{\dd\nu}{\dd\pi}$. We then have
            \begin{align}
                \chi^2(\nu K^t\|\pi K^t)&=\left\|\left(K^t\right)_\pi^\star f\right\|_{L^2(\pi K^t)}^2-1&\left(\substack{\text{By~\cref{eq:dual_kernel_formula}}}\right)&\nonumber\\
                &=\ip{f}{K^{2t}f}_\pi-1&\left(\substack{\text{By~\cref{eq:adjoint_relation} and}\\\text{since $K^t=(K^t)_\pi^\star$}}\right)&\nonumber\\
                &=\lambda^t\ip{f}{E_{\mathcal A}f}_\pi+(1-\lambda^t)\ip{f}{\Pi f}_\pi-1&\left(\substack{\text{By~\cref{eq:expression_K_2t}}}\right)&\nonumber\\
                &=\lambda^t\left(\|E_{\mathcal A}f\|_{L^2(\pi)}^2-1\right)&\left(\substack{\text{Since $\Pi f=\allones$ and}\\\ip{f}{E_{\mathcal A}f}_\pi=\|E_{\mathcal A}f\|_{L^2(\pi)}^2}\right)&\nonumber\\
                &=\lambda^t\chi^2(\nu E_{\mathcal A}\|\pi)&\left(\substack{\text{Since $\pi E_{\mathcal A}=\pi$ and}\\E_{\mathcal A}f=\frac{\dd\nu E_{\mathcal A}}{\dd\pi}}\right)&.\label{eq:proof_exact_chi_squared_partition_chain}
            \end{align}
            This proves~\cref{eq:exact_chi_squared_partition_chain}. Together with the DPI for $\chi^2$-Divergence,~\cref{eq:proof_exact_chi_squared_partition_chain} gives $\eta_{\chi^2}(\pi,K^t)\leq\lambda^t$. Conversely, taking $\nu=\pi_{|A_j}$ for any $j\in[q]$ gives $\nu E_{\mathcal A}=\nu$, proving that $\eta_{\chi^2}(\pi,K^t)=\lambda^t$.

            Moreover, the definition of $E_{\mathcal A}$ gives $\nu E_{\mathcal A}=\sum_{j=1}^q\nu(A_j)\pi_{|A_j}$. The convexity of $\chi^2$ in its first argument therefore yields
            \begin{equation}\label{eq:convexity_bound_on_chi_squared}
                \chi^2(\nu E_{\mathcal A}\|\pi)\leq\max_{j\in[q]}\chi^2(\pi_{|A_j}\|\pi)=\left(\pi_{\min}^{\mathcal A}\right)^{-1}-1.
            \end{equation}
            If $\nu E_{\mathcal A}=\pi$,~\cref{eq:proof_exact_chi_squared_partition_chain} gives $D_2(\nu K^t\|\pi)=0$. Otherwise,
            \begin{align}
                \frac{D_2(\nu K^t\|\pi)}{D_2(\nu\|\pi)}&=\frac{\log\left(1+\lambda^t\chi^2(\nu E_{\mathcal A}\|\pi)\right)}{\log\left(1+\chi^2(\nu\|\pi)\right)}&\left(\substack{\text{By~\cref{eq:proof_exact_chi_squared_partition_chain} and}\\\text{the definition of $D_2$}}\right)&\nonumber\\
                &\leq\frac{\log\left(1+\lambda^t\chi^2(\nu E_{\mathcal A}\|\pi)\right)}{\log\left(1+\chi^2(\nu E_{\mathcal A}\|\pi)\right)}&\left(\substack{\text{By DPI for $\chi^2$-Divergence}}\right)&\nonumber\\
                &\leq\frac{\log\Big(1+\lambda^t\left((\pi_{\min}^{\mathcal A})^{-1}-1\right)\Big)}{\log\big((\pi_{\min}^{\mathcal A})^{-1}\big)}.&\left(\substack{\text{By~\cref{corollary:log_ratio_increasing}}\\\text{and~\cref{eq:convexity_bound_on_chi_squared}}}\right)&\label{eq:mcmt_lemma_eta_2_ub}
            \end{align}
            Taking the supremum over admissible $\nu$ establishes the ``$\leq$'' direction in~\cref{eq:eta_2_partition_chain}. For the reverse inequality, fix $j\in[q]$ such that $\pi(A_j)=\pi_{\min}^{\mathcal A}$ and take $\nu=\pi_{|A_j}$. For this choice, all the inequalities in~\cref{eq:mcmt_lemma_eta_2_ub} are equalities, proving the expression for $\eta_2(\pi,K^t)$.
        \end{proof}
        For the Markov chain considered in~\cref{example:community_chain}, its Markov kernel satisfies
        \begin{equation*}
            K=pE_{\mathcal A}+(1-p)\Pi,\qquad K^2=p^2E_{\mathcal A}+(1-p^2)\Pi.
        \end{equation*}
        Since $E_{\mathcal A}$ and $\Pi$ are self-adjoint on $L^2(\pi)$, the chain is reversible with respect to $\pi$. Moreover, $\pi$ is uniform and $A_1$ has minimum cardinality, so that $\pi_{\min}^{\mathcal A}=|A_1|/|\X|$. Applying~\cref{lemma:eta_2_partition_chain} with $t=1$ and $\lambda=p^2$ therefore gives
        \begin{equation*}
            \eta_{\chi^2}(\pi,K)=p^2,\qquad \eta_2(\pi,K)=\frac{\log\left(1+p^2\left(\frac{|\X|}{|A_1|}-1\right)\right)}{\log\left(\frac{|\X|}{|A_1|}\right)}.
        \end{equation*}
        Finally, for every $x\in A_1$, we have $\delta_xE_{\mathcal A}=\pi_{|A_1}$ and $\chi^2(\pi_{|A_1}\|\pi)=\frac{|\X|}{|A_1|}-1$. Hence,~\cref{eq:exact_chi_squared_partition_chain} gives
        \begin{equation*}
            \chi^2(\delta_xK^t\|\pi)=p^{2t}\left(\frac{|\X|}{|A_1|}-1\right),\qquad t\geq1.
        \end{equation*}
\end{document}